\documentclass[10pt,oneside,a4paper,reqno]{amsart}  
\usepackage{mathrsfs}

\usepackage{amssymb}

\usepackage{amssymb, bbm,enumerate}
\usepackage{geometry}                
\usepackage[T1]{fontenc}
\usepackage[english]{babel} 
\usepackage{graphicx}
\usepackage{color}
\usepackage[colorlinks=true, linkcolor=magenta, citecolor=cyan, urlcolor=blue]{hyperref}

\usepackage{scalerel,stackengine}
\stackMath
\newcommand\reallywidehat[1]{%
\savestack{\tmpbox}{\stretchto{%
  \scaleto{%
    \scalerel*[\widthof{\ensuremath{#1}}]{\kern-.6pt\bigwedge\kern-.6pt}%
    {\rule[-\textheight/2]{1ex}{\textheight}}
  }{\textheight}%
}{0.5ex}}%
\stackon[1pt]{#1}{\tmpbox}%
}
\renewcommand{\phi}{\varphi}
\renewcommand{\ker}{\Ker}
\renewcommand{\Re}{\textup{Re }}
\renewcommand{\Im}{\textup{Im }}

\newcommand{\mc}[1]{\mathcal{#1}}

\newcommand{\mb}[1]{\mathbb{#1}}
\newcommand{\ms}[1]{\mathscr{#1}}
\newcommand{\mbb}[1]{\mathbbm{#1}}
\newcommand{\tint}{{\textstyle\int}}
\DeclareMathOperator{\dive}{div}

\def\P{P_{>N}}
\def\be{\begin{equation}}
\def\ee{\end{equation}}
\def\bea{\begin{eqnarray}}
\def\eea{\end{eqnarray}}

\def\nn{\nonumber}
\def\T{\mathbb{T}}
\def\C{\mathbb{C}}
\def\Z{\mathbb{Z}}
\def\N{\mathbb{N}}
\def\B{\mathscr{B}}
\def\Ga{\mathscr{G}}

\DeclareMathSymbol{\leqslant}{\mathalpha}{AMSa}{"36} 
\DeclareMathSymbol{\geqslant}{\mathalpha}{AMSa}{"3E} 
\DeclareMathSymbol{\eset}{\mathalpha}{AMSb}{"3F}     
\renewcommand{\leq}{\;\leqslant\;}                   
\renewcommand{\geq}{\;\geqslant\;}                   

\DeclareMathOperator{\supp}{supp}

\DeclareMathOperator{\Ker}{Ker}
\DeclareMathOperator{\Tr}{Tr}

\DeclareMathOperator{\dd}{dd}
\DeclareMathOperator{\Span}{span}

\DeclareMathOperator{\degb}{\overline{deg}}

\newcommand{\E}[1]{\mathbb{E}\left[#1\right]}

\def\a{\alpha}
\def\e{\varepsilon}
\def\d{\delta}
\def\g{\gamma}

\def\b{\beta}
\def\D{\Delta}

\def\r{\rho}
\def\s{\sigma}

\def\R{\mathbb{R}}
\def\C{\mathbb{C}}
\def\E{\mathbb{E}}

\theoremstyle{plain}
\newtheorem{theorem}{Theorem}[section]
\newtheorem{lemma}[theorem]{Lemma}
\newtheorem{proposition}[theorem]{Proposition}
\newtheorem{corollary}[theorem]{Corollary}

\theoremstyle{definition}

\theoremstyle{remark}
\newtheorem{remark}[theorem]{Remark}

\numberwithin{equation}{section}

\definecolor{light}{gray}{.9}

\author{Giuseppe Genovese}
\address{Giuseppe Genovese: Institut f\"ur Mathematik, Universit\"at Z\"urich,
CH-8057 Z\"urich, Switzerland.}
\email{giuseppe.genovese@math.uzh.ch}

\author{Renato Luc\`a}
\address{Department Mathematik und Informatik, Universit\"at Basel Spiegelgasse 1, CH-4051 Basel, Switzerland}
\email{renato.luca@unibas.ch}

\author{Daniele Valeri}
\address{School of Mathematics and Statistics, University of Glasgow, G12 8QQ Glasgow, UK}
\email{daniele.valeri@glasgow.ac.uk}

\title[Invariant measures for the DNLS equation]
{Invariant measures for the periodic
derivative nonlinear Schr\"odinger equation}

\date{\today}



\makeindex
\begin{document}

\begin{abstract}
We construct invariant measures associated to the integrals of motion of the periodic derivative nonlinear Schr\"odinger equation (DNLS) for small data in $L^2$ and we show these measures to be absolutely continuous with respect to the Gaussian measure. The key ingredient of the proof is the analysis of the gauge group of transformations associated to DNLS. As an intermediate step for our main result, we prove quasi-invariance with respect to the gauge maps of the Gaussian measure on $L^2$ with covariance $(\mathbb I+(-\D)^k)^{-1}$ for any $k\geq2$. 
\end{abstract}

\maketitle

\section{Introduction}

In this paper we continue our studies on the periodic DNLS equation
\be\label{eq:DNLS}
\left\{
\begin{array}{rcl}
i\partial_t \psi + \psi'' & = & i\beta\left(\psi|\psi|^2\right)'  \\
\psi(x,0) & = &  \psi_{0}(x) \,,\quad x\in\T\,,
\end{array}
\right.
\ee
where $\psi(x,t) : \mathbb{T} \times \mathbb{R} \rightarrow \mathbb{C}$,
$\psi_{0}(x) : \mathbb{T} \rightarrow \mathbb{C}$,
$\psi'(x,t)$ denotes the derivative of $\psi$ with respect to $x$, and 
$\b\in\R$ is a real parameter. We denote by $\Phi_t$ the associated flow-map. 

This is a dispersive nonlinear model describing the motion along the longitudinal direction of a circularly polarized wave, generated in a low density plasma by an external magnetic field \cite{moj}. It is an integrable system \cite{KN78} (see also \cite{DSK13}),
in the sense that there is an infinite sequence of linearly independent quantities (integrals of motion) which 
are conserved by the flow of \eqref{eq:DNLS} for sufficiently regular solutions.

In our previous work \cite{GLV16} we constructed a family of Gibbs measures, supported on
Sobolev spaces of increasing regularity, associated to the integrals of motion of the DNLS equation.
In this paper we construct a sequence of measures invariant along the flow. We prove these measures to be absolutely continuous with respect to the Gaussian measures, however we cannot show them to coincide with the Gibbs measures (albeit this is expected to be true). 

The studies of PDEs from the perspective of statistical mechanics started with the seminal paper by Lebowitz, Rose and Speer
\cite{Leb}, where the periodic one dimensional NLS equation was studied by introducing the statistical
ensembles naturally associated to the Hamiltonian functional.
Successively, starting from the paper \cite{B94} Bourgain gave fundamental contributions to the development of this field, for a comprehensive exposition we refer to \cite{Bo} and the references therein. 
For integrable PDEs one can profit from an infinite number of higher 
Hamiltonian functionals, in order to construct infinitely many invariant Gibbs measures. This was originally noted 
by Zhidkov \cite{Zh01}, in the context of
Korteweg-de Vries (KdV) equation and cubic nonlinear Schr\"odinger (NLS) equation on $\T$.
A similar result was achieved in the last years for the Benjamin-Ono equation 
on $\T$ via a series of papers by 
Tzvetkov, Visciglia and Deng \cite{TzPTRF,TV13a,TV13b,TV14,D14,DTV14}. In this case (likewise for DNLS) a more careful construction of the (invariant) measures is 
required compared to KdV and NLS. 
Recently a renewed interest invested the subject and numerous works treating different aspects of it for a large class of equations appeared. However we are not attempting here to give an exhaustive account of the literature, but we focus just on DNLS. 

The construction of the Gibbs measure associated to the energy functional of the DNLS equation
\be\label{eq:energia}
E_1[\psi]
 =
\frac12 \|\psi\|_{\dot{H}^1}^2
+\frac{3i}4\beta\int|\psi|^2\psi'\bar\psi
+\frac{\beta^2}4\|\psi\|_{L^6}^6 \, ,
\ee
was achieved in \cite{TT10}, while in \cite{NOR-BS12} and \cite{NR-BSS11} the measure was proven to be invariant. 
The proof of \cite{NOR-BS12} uses quite a sophisticated strategy, that we briefly explain. 
The best local well-posedness result for the DNLS equation in Sobolev spaces is for data in $H^{1/2}$ \cite{Herr}, which falls outside the support of the Gibbs measure (for existence of weak solutions below $H^{1/2}$ see \cite{tak16}). However in \cite{GH} local well-posedness in the Fourier-Lebesgue spaces~$\mathcal FL^{s,r}(\T)$ with $r\in[2,\infty)$ and $(s-1)r<-1$ is proven. The authors of \cite{NOR-BS12} showed these spaces to be of full measure or more precisely that $(\imath,H^1,\mathcal FL^{s,r})$ is an abstract Wiener space, where $\imath\,:\,H^1\mapsto\mathcal FL^{s,r}$ is the inclusion map (see \cite{kuo}). Then they establish energy growth estimates for each single solution, in a localised-in-time version of the Bourgain space $X^{2/3-,1/2}_{3}$, for initial data in $\mathcal FL^{2/3-,3}$. 



In the present paper, we opt for a more probabilistic approach, closer to the one developed by
Tzvetkov and Visciglia in the context of the Benjamin-Ono equation. 
As in \cite{NOR-BS12}, the gauge transformation introduced by Herr in the periodic setting (see \cite{Herr})
constitutes one of the main ingredient of our proof, even though we make a different gauge choice. Indeed our gauge simplifies the integrals of motion rather than the equation. More precisely we consider a one-parameter family of gauge transformations (see \eqref{gauge_change} in Section \ref{sec:gauge}), and for each integral of motion we select an appropriate value of the parameter. One of the advantages of this approach is that we can simply work in Sobolev spaces, without introducing any auxiliary functional space.

In the rest of the introduction we present the set-up in which our main Theorem \ref{Th:Main} is stated and we explain the strategy of the proof.

\subsection{Set-up and Main Results}

We introduce here the objects we are going to deal with. 
According to a standard notation we denote by $H^s(\T)$, $s \geq 0$, the completion of $C^{\infty}(\T)$ with 
respect to the norm induced by the inner product
$$
(f, g)_{H^s}:=\sum_{n\in\Z} (1+n^{2s})  f(n) \bar g(n) \,,
$$
where $f(n)$ is the $n$-th Fourier coefficient of $f$. For every $s \geq 0$, $H^s(\T)$ is a separable
Hilbert space, with $H^{0}(\T)=L^2(\T)$. A function in $H^s(\T)$ is represented as a 
sequence $\{f(n)\}_{n\in \mb Z}$ such that $\sum_{|n|\leq N}(1+ n^{2s})|f(n)|^2$ converges as $N \to \infty$. 
We also use the homogeneous Sobolev spaces $\dot H^s(\T)$, defined as the completion of $C^{\infty}(\T)$ with respect to 
the norm induced by the inner product
$$
(f, g)_{\dot H^s}:=\sum_{n\in\Z} n^{2s}  f(n) \bar g(n) \,.
$$
For $N\geq0$, we consider the canonical projections 
\begin{equation}\label{eq:proj}
P_{N} : L^{2}(\T) \mapsto E_N:=\Span_{\mb C}\{e^{inx}\,:\, |n|\leq N\} \, ,
\end{equation}
defined as
$$
P_N f := \sum_{|n| \leq N}  e^{inx} f(n) \,,\quad P_{N=\infty}f=f\,.  
$$
The orthogonal projections are $P_{>N}:=\mathbb{I}-P_N$.

The space $L^2(\T)$ can be equipped with a measurable-space structure as follows. 
Let $A\in\ms B(\mb C^{2N+1})$ be a Borel subset of $\mb C^{2N+1}$. We introduce the cylindrical sets
\begin{equation}\label{CylSets}
M_N(A):=\{ f \in L^2(\T)\,:\,(f(-N),\ldots,f(N))\in A\}
\,.
\end{equation}
Hence, we define
$\mathcal T_{N}:=\{M_N(A)\}_{A\in\mathscr{B}(\C^{2N+1})}$
and $\mathcal T := \bigcup_{N \in \N} \mathcal T_{N}$, namely the algebra of cylindrical sets with Borel basis given by \eqref{CylSets}.
We also denote by $\s(\mathcal T)$ the smallest $\s$-algebra generated by $\mathcal T$.
For any $N \in \N_0 := \N \cup \{ 0 \}$, $\mathcal T_{N}$ is isomorphic to $\mathscr{B}(\C^{2N+1})$, therefore we can identify 
$\s(\mathcal T)$ with  
$\mathscr{B}(L^2(\T))$, namely the Borel $\s$-algebra on~$L^{2}(\T)$.
The Lebesgue measure on $\C^{2N+1}$ naturally induces a measure on $(M_{N}(\C^{2N+1}), \mathcal T_{N})$, which (with a little abuse) 
we still refer to as Lebesgue measure, through  
\begin{equation}\label{NewLebesgue}
| M_{N} (A)| := c_N \int_{A} \bigg( \prod_{|n| \leq N} df(n)d\bar f(n) \bigg) \, .
\end{equation}
where $c_N$ is a suitable constant. 
%
%
%
For any $k\in\N$, let $\mathbb{I}+(-\D)^k$ be the closure in $L^2(\T)$ of the operator 
$1+\left(-\frac{d^{2}}{d x^{2}}\right)^k$ acting on $C^{\infty}(\T)$. This is a positive self-adjoint 
operator with a trivial kernel and its inverse $(\mathbb{I}+(-\D)^k)^{-1}$ is bounded and trace class.
Therefore there exists a centred Gaussian probability measure with covariance $(\mathbb{I}+(-\D)^k)^{-1}$ 
(a standard reference is \cite{kuo}), which we denote by $\g_k$, such that 
\begin{equation}\label{Def:gammaK}
\g_k(M_{N}(A)):= \frac{1}{Z_N}
\int_A \bigg( \prod_{|n| \leq N} df(n)d\bar f(n) \bigg) e^{-\frac12\sum_{|n|\leq N}(1+n^{2k})|f(n)|^2 }
\, ,
\end{equation}
where $Z_{N}$ is the normalisation constant.

For any $k\in\N$ the triple $(L^2(\T), \B(L^2(\T)), \g_k)$ is a Gaussian probability space. It is worth to recall that each measure $\g_k$ 
concentrates on functions with less than $k-\frac12$ derivatives in $L^2(\T)$. More precisely
$$\g_k \Bigg( \bigcap_{s < k-\frac{1}{2}}  H^{s}(\T) \Bigg)=1 \, ,\quad \g_k \Bigg(H^{k-\frac12}(\T) \Bigg)=0\,. $$  
With $L^p(\g_k)$ we denote the $L^p$ spaces associated to $\g_k$. 

\medskip

We can now introduce the Gibbs measures constructed in \cite{GLV16}. 
Let $R>0$ and let $\chi_R(x):=\chi(x/R)$ where $\chi: \R \to [0,1]$ is a smooth compactly supported function  
such that $\chi(x) = 1$ for $|x| \leq 1/2$ and $\chi(x) = 0$ for $|x| > 1$.
For $k\geq2$, let us fix $R_m>0$, for $m=0,\dots,k-1$, and
define the $k$-th Gibbs measure associated to the DNLS equation by
\be\label{eq:GIBBS-measure}
\rho_{k}(A):=  \int_A
\left(
\prod_{m=0}^{k-1}
\chi_{R_m}\left(E_{m} [\psi] \right)
\right)
e^{- Q_k [\psi] }\gamma_{k}(d\psi), \quad A\in\B(L^2(\T))\,,
\ee
where
\begin{equation}\label{Def:Qk}
Q_k [\psi] := E_k[\psi]-\frac12\| \psi \|_{\dot H^k}^2 
\end{equation}
and $E_{1}, \ldots, E_{k}$ are integrals of motion of the DNLS equation; see \eqref{integrals_DNLS}, \eqref{20170524:eq1}. The measure $\r_k$ must be understood as the weak limit of the following sequence of measures. Given $A\in \B(L^2(\T))$, we denote
\begin{equation}\label{SecCyl}
\mc M_N (A) := \{ f \in L^{2}(\T) :   P_N f \in P_N A    \}
\end{equation}
the (cylindrical) set 
of all the $L^{2}(\T)$ functions with base $P_N A$. 
Then we define
$$
\r_{k,N}(A):= \int_{\mc M_N(A)}
\left(
\prod_{m=0}^{k-1}
\chi_{R_m}\left(E_{m} [P_N\psi] \right)
\right)
e^{- Q_k [P_N\psi] }\gamma_{k}(d\psi) \,.
$$
Then the main result of \cite{GLV16} can be reformulated as follows:

\begin{theorem}[\cite{GLV16}]\label{Th:Old}
Let $k\geq 2$ and let $R_0$ sufficiently small. Then $(L^2(\T),\B(L^2(\T)),\r_k)$ is a probability space.
Moreover, there exists $p_0=p_0(R_0,\dots,R_{k-1},k,|\beta|)>1$ such that, for all $1\leq p<p_0$, the Radon-Nykodim derivative
$\frac{d\r_k}{d\g_k}$ belongs to $L^{p_0}(\g_k)$. We can take $p_0$ arbitrarily large (but not $p_0 = \infty$) provided we choose a sufficiently smaller $R_0$.
\end{theorem}

The small mass condition $R_0 \ll 1$ deserves few comments. 
As already mentioned, the periodic DNLS equation has been shown to be locally well-posed for initial data in $H^{s\geq 1/2}$ in \cite{Herr}. 
Then, a standard procedure allows to globalise the local $H^1$ solutions
with $\| \psi_{0} \|_{L^{2}} < \delta$, as long as~$\delta$ is sufficiently small, by using the integral of motion $E_1$ and the  
Gagliardo--Nirenberg inequality
\be\label{eq:GN}
\|f\|^3_{L^6(\mathbb{T})}\leq
\|f\|_{\dot{H}^{1}(\mathbb{T})}\|f\|_{L^2(\mathbb{T})}^{2}
+
\frac{1}{2 \pi} \|f\|^3_{L^2(\mathbb{T})}\,.
\ee
However this approach does not give the best possible value for $\delta$, which is an interesting open problem. 
The highest value of the mass for which global existence in $H^1(\T)$ holds is $\delta=2\sqrt{\pi / |\beta|}$. This was shown for the non-periodic framework in \cite{Wu13} and the argument was adapted to periodic DNLS in \cite{MO15} (similarly, for the best result in $H^{\frac12}(\T)$ see \cite{Mon17}). Existence of global solution of DNLS on $\R$ without any condition on the mass has been proven by inverse scattering method in \cite{Sul}.
In our approach the small mass condition is required to prove the integrability of the Gibbs densities. 
However from the standpoint of integrable systems it is reasonable to think that in the role of $R_0$ could be replaced by any $R_j$ (once one looks at sufficiently regular solutions).

The main contribution of the present paper follows.

\begin{theorem}\label{Th:Main}
Let $k\geq2$ and let $R_0$ be sufficiently small. Then there exists a probability measure $\hat\r_k$ on $(L^2(\T),\mathscr{B}(L^2(\T)))$ a.c. w.r.t. $\g_k$, such that the flow-map~$\Phi_t$ associated to DNLS is measure preserving
in~$(L^2(\T),\mathscr{B}(L^2(\T)),\hat\r_k)$.
\end{theorem}

An immediate but significant corollary follows by the Poincar\'e recurrence Theorem:
\begin{corollary}\label{Cr:Poincare}
Let $k \geq 2$ and let $\psi$ be a solution of the DNLS equation with initial datum $\psi(x, 0)\in H^s(\T)$ with $s < k-\frac12$. 
For $\hat \r_k$-a.e. $\psi(x, 0)$ there exists a divergent sequence $\{t_n\}_{n\in\N}$ such that 
$$
\lim_{n\to\infty}\|\psi(x, t_n)-\psi(x, 0)\|_{H^s}=0 \, .
$$
\end{corollary}
An analog conclusion for $k=1$ follows from \cite{NOR-BS12}. To the best of our knowledge, these are the sole known results on the long-time behaviour of the DNLS equation. 

\subsection{Strategy of the Proof}

An alternative formulation of our main Theorem \ref{Th:Main} is that for any $k\geq2$ (in fact $k\geq1$ considering also the result of \cite{NOR-BS12}) the DNLS equation has the structure of an infinite-dimensional (Hamiltonian) dynamical system.
Since the earlier works of Bourgain and Zhidkov it has revealed useful to approximate the infinite dimensional problem with 
a finite dimensional one, by considering the evolution of the first $|n| \leq N$ Fourier modes of the solutions. These systems are actually Hamiltonian,
but in general they do not preserve all the integrals of motion. This is often a major issue to cope with in this class of problems. However one expects the integrals of motion to be conserved in the limit $N \to \infty$. Following an approach developed by Tzvetkov and Visciglia for the 
Benjamin-Ono equation, we will show that the derivative of the integrals of motion along the flow of the truncated systems vanishes in the $L^2(\gamma_k)$ mean. Actually, as first observed in \cite{TV13b}, one can reduce to consider only the derivative at the initial time, which is a crucial simplification.

It is helpful to recall that the integrals of motion of DNLS have the following
form
\begin{equation}\label{hcqeip}
E_{k}[\psi] = \frac12\|\psi\|^2_{\dot H^k} - \frac 12 \b(2k+1) \Im \int \psi^{(k)} \bar\psi^{(k-1)} |\psi|^{2} +\mbox{remainders}\,,
\qquad
k \geq 2 \,,
\end{equation}
where
we consider as remainders all the terms that are bounded in the support of $\gamma_k$.
The difficulty to show the asymptotic (w.r.t. $N$) conservation of $E_{k}$ comes form the second addendum in the r.h.s. of equation \eqref{hcqeip}. Notably the integrals of motion of the Benjamin-Ono equation have an analog structure. However in that case a convenient cancellation coming from the symmetries of the problem simplifies substantially the computations.
We cannot find a similar cancellation here. Nevertheless it is possible to eliminate the troubling term using a suitable gauge transformation. 
As already mentioned, these gauge transformations form a one-parameter group $\Ga_\a$ indexed by $\alpha \in \R$ (see \eqref{gauge_change}). 
A generic gauge choice yields the following expression for the integrals of motion of the gauged equation    
\begin{equation}\nonumber
\mc E_k [\phi] = \frac12\|\phi\|^2_{\dot{H}^k}
+ ik\alpha\mu\int \bar \phi^{(k)} \phi^{(k-1)}
-\frac 12 \left((2k+2)\alpha+(2k+1)\beta\right)
\Im \int \phi^{(k)}\bar\phi^{(k-1)} |\phi|^2
+ \mbox{remainders}
\,,
\end{equation}
where $\phi = \Ga_{\a} \psi$ is the solution of the gauged equation and 
we shortened 
$$
\mu := \frac{1}{2\pi} \| \phi \|^{2}_{L^{2}} = \frac{1}{2\pi} \| \psi \|^{2}_{L^{2}} \, .
$$
In general we will use the notation $\mu[f]:= \frac{1}{2\pi}\| f \|^2_{L^2}$ and we will simply write $\mu$ if there will be no ambiguity.
We recover \eqref{hcqeip} as $\alpha=0$.
Setting 
\begin{equation}\label{OCOA}
\alpha = -\frac{2k+1}{2k+2}\beta \, ,
\end{equation}
we reduce to
\begin{equation}\nonumber
\mc E_k [\phi] = \frac12\|\phi\|^2_{\dot{H}^k}
- i k \frac{2k+1}{2k+2}\beta  \mu \int \bar \phi^{(k)} \phi^{(k-1)}
+ \mbox{remainders}
\,.
\end{equation} 
This form of the integrals of motion is much more suitable in order to prove the asymptotic conservation property and such a reduction is the crux of our proof.
Of course also the flow of DNLS changes accordingly to the gauge transformation: indeed our gauge choice leads to a somewhat more involved form
for the equation (see \eqref{eq:GDNLS}). However this does not introduce significant 
difficulties, as the form of the nonlinearity is essentially the same and we are working with rather regular solutions (at least in $H^{s > 5/4}$). 
It is worthy to point out the difference with what is usually done in the low regularity theory for DNLS,
where the choice of the gauge parameter $\alpha = -\beta$ aims to simplify the equation; see also Remark \ref{EdRem}. 

The next step is to define for any $k\geq2$ a gauged Gibbs measure starting from the gauge-transformed (or {\em gauged}) integrals of motion and to prove its invariance w.r.t. the gauged flow. This requires some groundwork, namely a careful analysis of the DNLS-flow and gauge-flow maps, in order to adapt the strategy of \cite{TV13b}.

The invariance of the gauged Gibbs measure under the gauged flow easily implies the invariance of the push-forward of it through $\Ga_\a$ under the DNLS flow. This will be our invariant measure $\hat \r_k$. We stress that in principle one expects $\r_k=\hat \r_k$. The missing step to show the invariance of the Gibbs measures is the proof of absolute continuity of the pull-back $\g_k\circ \Ga_\a$ w.r.t. $\g_k$ with the explicit density. So far what we can prove is the following theorem:
\begin{theorem}\label{th:gauge}
Let $R_0>0$ small enough and $$\tilde\g_k(A)=\g_k(A\cap\{ f \in L^2 : \mu[f] \leq R_0\})\,.$$
Then for any $k\geq2$ and $\a \in \R$ the measure $\tilde \g_k \circ \Ga_\a$ is absolutely continuous w.r.t.~$ \g_k$.
\end{theorem}

The absolute continuity of $\hat\r_k$ w.r.t. $\g_k$ is then a direct consequence of Theorem \ref{th:gauge}; see Remark \ref{AcRemark}.

The change of variable formula for $k=1$ was established in \cite{NR-BSS11}. This is however a very special case, as the typical trajectories for $\g_1$ are complex Brownian bridges, whose properties are crucially employed in the argument of \cite{NR-BSS11}. For more regular processes one cannot expect to reproduce the same proof and some new idea is needed. Since the work of Ramer \cite{ramer}, much attention has been given to the transformation properties of Gaussian measures under anticipative transformations (as the DNLS gauge is). However the gauge group does not match the typology of transformations studied by Ramer onward. For this reason the study of the quasi-invariance of $\g_k$ under the gauge map is of independent interest. 

Recently Tzvetkov \cite{sigma} proposed a strategy for proving quasi-invariance of the Gaussian measure under a one-parameter group of transformations via a {\em soft} argument, which does not provide the explicit density. The method has been successively refined in \cite{OT1, NLW, OT3}. We use this approach to prove Theorem \ref{th:gauge}, from which we deduce the absolutely continuity of $\hat \r_k$ w.r.t. $\g_k$. To prove that the Gibbs measures are invariant, one should known the exact form of the densities after the change of variables given by the gauge also for $k\geq2$. As in the case $k=1$, these densities should complete exactly the part of the integrals of motion missing in the gauged Gibbs measure. 
We do not give here further details, leaving the discussion of this problem to future works. 

\begin{remark}\label{EdRem}
Finally we come back briefly to the case $k=1$. 
As we have noted, the gauge transformation relative to the choice \eqref{OCOA} introduces a significant simplification in the form of the conservation laws. This allows us to obtain the key Proposition \ref{Prop: energie}, which can be easily extended to the case $k=1$. 
In \cite[Theorem4.2]{NOR-BS12} they prove a (deterministic) analogous of this proposition, which is again the main step in the proof of the invariance of the measure, but their choice $\alpha = - \beta$ does not give any simplification in this part. 
Our argument does not immediately cover the case $k = 1$, since the missing step is to extend the stability Proposition \ref{Prop:nearness} to low regularity.
In particular, the works of Herr \cite{Herr}
and Gr\"unrock--Herr \cite{GH} suggest that such a result should be harder to achieve under our choice \eqref{OCOA} of $\alpha$ in place of the usual $\alpha = - \beta$.   
\end{remark}

\subsection{Organisation of the paper}
The paper is organised as follows.
In Section \ref{sec:gauge} we introduce the gauge transformation used throughout the paper and
we analyse the way the DNLS equation and its integrals of motion change according to it.
This leads to the gauged DNLS equation \eqref{eq:GDNLS}, which we refer to as GDNLS, and to the gauged integrals of motion $\mc E_{\ell}$, $\ell\in\N_0$,
defined in \eqref{eq:def-energie-gauged}.
The main results of this section are Corollary \ref{DanCor2} and \ref{Nontiene} where the explicit representation of the gauged integrals of motion
$\mc E_\ell$ is obtained.
In Section \ref{sect:flows} we introduce the truncated GDNLS equation \eqref{20140509:DNLSapprox} and we show
in Proposition \ref{LebMeasPres} that its
flow preserves the Lebesgue measure \eqref{NewLebesgue}.
Moreover, we show that this flow is close to the one of the GDNLS equations
in a suitable Sobolev topology for short time, see Proposition \ref{Prop:nearness}.
Section \ref{sect:Wick} is devoted to the study of the asymptotic conservation of the integrals of motion in the probabilistic sense.
We show in Proposition \ref{Prop: energie} that the $L^2$ norm w.r.t. $\g_k$ ($k\geq 2$) of the time derivative calculated at $t=0$ of 
the integrals of motion $\mc E_\ell$  
vanishes as $N\to\infty$ (namely as the truncation disappears).
In order to prove this result we need to study the asymptotic conservation of the monomials appearing in the explicit form of the integrals of motion $\mc E_\ell$ obtained in Section 2.
For most of them we have also convergence $\g_k$-a.s. for $k\geq2$,
as proved in Lemma \ref{LemmaObv} and \ref{UsLem}.
However this is can not be easily proved in general. The more complicated terms are handled using Wick 
theorem in Lemma \ref{Lemma:Wick6} and \ref{Lemma:Wick4}.
In Section \ref{Sect:Proof} we perform the construction of the invariant measures of Theorem \ref{Th:Main}.
First we construct the gauged Gibbs measures, essentially repeating the argument used for the Gibbs measures in \cite{GLV16} starting by the gauged integrals of motions. Then, following \cite{TV13b},
we prove that they are invariant under the flow of the GDNLS equation, using the results of Section \ref{sect:flows} and Section \ref{sect:Wick}. 
The last two sections are devoted to the proof of absolute continuity of the invariant measures w.r..t the Gaussian measures via the quasi-invariance of the latter under the gauge flow. In Section \ref{sect:gauge1} we introduce and analyse the truncated gauge flow in analogy with what was done in Section \ref{sect:flows}. In Section \ref{sect:gauge1} we exploit the argument of \cite{sigma} to prove Theorem \ref{th:gauge}. This ultimates also the proof of Theorem \ref{Th:Main}.

\subsection*{Notations}
Throughout $f$ will always be a generic complex function, $\psi$ a solution of DNLS, $\phi$ the image of $\psi$ via the gauge transformation, $u$ can be either $\phi$ or $\bar\phi$. 
We denote by $f(n)$ the $n$-th Fourier coefficient of $f:\T \to \C$. 
We set $\mu[f]:= \frac{1}{2\pi}\| f \|^2_{L^2}$ and sometimes we shorten this simply writing $\mu$. 
For the ease of notation, $\E[\cdot]$ denotes the expectation value w.r.t. $\g_k$, regardless of $k$. Anyway the particular $\g_k$ considered will be always clear from the context. 
$B^{s}(R)$ is the ball of center zero and radius $R$ in the topology induced by $\| \cdot \|_{H^{s}}$. We write $X \lesssim Y$ to denote that $X \leq C Y$ for some positive constant $C$ independent on $X,Y$. We also use the symbol $\mc O$ for the Landau big O.
We denote $\N_0=\N\cup\{0\}$. We use the following notations for further (space) derivatives of the solutions 
$\phi' := \partial_x \phi$, $\phi'' := \partial_x^2 \phi$, $\phi^{(k)} := \partial_x^k \phi$, $k \geq 3$.

\subsection*{Acknowledgements} This work was partially developed during the visit of the first two authors to the CIRM, Trento (through the program Research in Pairs) and to the Yau Mathematical Sciences Center of Tsinghua University, Beijing. Both these institutions are thankfully acknowledged. Furthermore, we are grateful to Lorenzo Carvelli for many stimulating discussions. 
The work of GG is supported through NCCR SwissMAP.
The work of RL is supported through the ERC grant 676675 FLIRT.
The work of DV was supported 
through the NSFC  ``Research Fund for International Young Scientists'' grant and through a Tshinghua University startup research grant when working in the Yau Mathematical Sciences Center.


\section{Gauge Transformations}\label{sec:gauge}

The DNLS equation has interesting transformation properties with 
respect to a group of gauge maps 
(introduced in the periodic setting in \cite{Herr}) which will be now discussed.

For $\a \in \R$ let $\Ga_\a \, : \, L^2(\mathbb{T}) \to L^2(\mathbb{T})$ be defined by
\be\label{gauge_change}
(\Ga_\a f )(x):= e^{ i \alpha \mathcal{I}[f(x)] }  f(x)\,.
\ee
where
\begin{equation}\label{DefMathcali}
 \mathcal{I}[f(x)] := \frac{1}{2\pi }\int_0^{2\pi }d\theta\int_{\theta}^x\left(|f(y)|^2-\frac{\|f\|_{L^2(\T)}^2}{2\pi }\right)dy \,.
\end{equation}
One can easily check that the (real) function $\mathcal{I}[f(x)]$ is the unique zero average ($2 \pi$-periodic) primitive of $|f(x)|^2- (2 \pi)^{-1} \|f\|_{L^2}^2$.
Note that~$|f|=|\Ga_\a f|$ 
and $\Ga_\a f $ is $2\pi$-periodic.
Hence, $\mc G_{\alpha}$ maps $L^2(\T)$ into $L^2(\T)$ preserving the norm (namely $\| \Ga_\a  f  \|_{L^2} = \|f\|_{L^2}$). Using 
that $\mc I[f]=\mc I[\mc G_\alpha(f)]$
one can easily show that the map $\alpha \to \Ga_\a$ is a one parameter group of transformations on $(\R, +)$, namely
\begin{equation}\label{eq:gauge_properties}
\Ga_0=\mathbb{I}
\qquad
\text{and}
\qquad
\Ga_{\alpha_1}\circ\Ga_{\alpha_2}=\Ga_{\alpha_1+\alpha_2}\,,
\quad
\text{for any }\alpha_1,\alpha_2\in\R
\,. 
\end{equation}
For any $s \geq 0$ the gauge transformation $\Ga_\a$ is also 
a homeomorphism of  
$H^{s}(\T)$ into itself. This is an immediate consequence of the following useful inequality
\begin{equation}\label{lemma:gauge-bound-Hs}
\left\| \left( e^{i \alpha \mathcal{I}[f]}  - e^{ i \alpha \mathcal{I}[g]}  \right) h \right\|_{H^{s}} \leq C e^{|\alpha| C ( \|f\|_{H^{s}}^{2} + \|g\|_{H^{s}}^{2} )}
(\|f\|_{H^{s}} + \|g\|_{H^{s}} ) \|f-g\|_{H^{s}} \| h \|_{H^{s}} \, ,
\end{equation}
where $C$ only depends on $s$, proved in \cite{Herr} in the case $\alpha=-1$ (the adaptation of the proof to the general case $\alpha \in \R$ is straightforward). 
Let indeed assume $f, g \in B^{s}(R)$, 
where $B^{s}(R)$ is the ball of center zero and radius $R >0$ in the topology induced by $\| \cdot \|_{H^{s}}$. Using~\eqref{lemma:gauge-bound-Hs} one easily deduces
\begin{align}\label{lemma:gauge-bound-HsBis}
\| \Ga_\a f - \Ga_\a g \|_{H^{s}} 
& \leq \left\| \left( e^{i \alpha \mathcal{I}[f]}  - e^{i \alpha \mathcal{I}[g]}  \right) f \right\|_{H^{s}} 
+  \left\| \left( e^{i \alpha \mathcal{I}[g]}  - 1  \right) (f-g) \right\|_{H^{s}} + \| f-g \|_{H^{s}}
\\ \nonumber
&  \leq \big( 2 \tilde C R^{2} e^{ |\alpha| 2 \tilde C R^{2}}  \! + \! \tilde C R^{2} e^{ |\alpha| \tilde C R^{2}} \! + \! 1 \big) \| f-g \|_{H^{s}}
\\ \nonumber
&
\leq   C R^{2} e^{ |\alpha| C R^{2}}  \| f-g \|_{H^{s}} \, ,  
\end{align}
where the constants $C \geq\tilde C> 1$ here still only depends on $s$.

\subsection{Gauged DNLS equation}

Let $\psi$ be a solution of the DNLS equation \eqref{eq:DNLS}.
For any $\alpha\in\mb R$, we set for brevity 
\begin{equation}\label{GRelation1}
\phi:=\Ga_{\a}\psi
\,.
\end{equation}
From 
\eqref{eq:gauge_properties} we clearly have
$\psi=\ms G_{-\alpha}(\phi)$.
We also 
have 
\begin{equation}\label{eq:gauge_der}
\psi^{(k)} = \partial_x^k(\ms G_{-\alpha}(\phi))=e^{-i\alpha \mathcal{I}(\psi)} \left( \partial_{x} - i \alpha ( |\phi|^{2} - \mu [\phi])  \right)^k \phi \, ,
\end{equation}
where $\mu[f] := \frac{1}{2\pi} \| f \|_{L^2}^2$.
Note that, since $|\phi| = |\psi|$, we have $\mu[\phi(x, t)] = \mu[\psi(x, t)]$ and, since this 
quantity is conserved by the flow of the DNLS equation (see~\eqref{integrals_DNLS}, 
namely $\mu[\psi(x, t)]=\mu[\psi(x, 0)]$ for all $t\in\mb R$) we will often simply denote it by $\mu$. 

The following proposition specifies the gauged form of the DNLS equation. For brevity we denote throughout by $\phi$ the solution
of this one parameter family of equations, even if $\phi$ depends on the choice of the parameter $\alpha$, see \eqref{GRelation1}. 
\begin{proposition}\label{proposition:GDNLS}
Let $\psi$ be a solution of the DNLS equation \eqref{eq:DNLS}. Then, for any $\alpha\in\mb R$,
the function $\phi=\ms G_{\alpha}(\psi)$ satisfies the equation
\begin{equation}\label{eq:GDNLS}
i \partial_{t} \phi  
+ \phi'' 
+2i\alpha\mu\phi'=
  i c_1  |\phi|^{2} \phi'
+ i c_2 \phi^{2} \bar{\phi}'
+ c_3  |\phi|^{4} \phi
+ c_4 \mu |\phi|^{2}\phi 
+ \Gamma [\phi] \phi \, ,
\end{equation} 
where 
\begin{equation}\label{eq:GDNLSconst}
c_1 = 2 (\alpha +  \beta)\,,
\qquad 
c_2 = 2 \alpha + \beta\,,
\qquad
c_3 =  -\alpha^{2} - \frac{\alpha \beta}{2}\,,
\qquad
c_4 = - \alpha \beta 
\end{equation}
and  
\begin{equation}\label{eq:GDNLSgamma}
\Gamma [f] 
= 
\left(\frac{3\alpha \beta }{4 \pi}  +\frac{\alpha^2}{\pi}\right)\| f \|_{L^{4}}^{4} - \alpha^{2} \mu[f]^{2} +  \frac{i \alpha}{\pi} \int_{\mathbb{T}} f' \bar{f}  \, .
\end{equation}
\end{proposition}
\begin{proof}
Using \eqref{eq:gauge_der} we get
\begin{equation}
\begin{split}\label{20170407:eq1}
\psi''
&
=e^{-i\alpha \mathcal{I}(\psi)} \left( \phi''-3i\alpha|\phi|^2\phi'-i\alpha\phi^2\bar\phi'
+ 2i\alpha\mu\phi'
\right.
\\
&\left.
-\alpha^2|\phi|^4\phi+2\alpha^2\mu|\phi|^2\phi-\alpha^2\mu^2\phi\right) 
\,,
\\
(| \psi |^{2} \psi)'
&
=e^{-i\alpha \mathcal{I}(\psi)} 
 \left( 
2|\phi|^{2} \phi'+\phi^2\bar\phi'
  -i\alpha|\phi|^{4} \phi +i \alpha \mu| \phi|^{2}\phi \right) \, .
\end{split}
\end{equation}
Note that
\begin{equation}\label{20170407:eq2}
\partial_t\psi=e^{-i\alpha\mc I[\psi]}\left(\partial_t\phi-i\alpha\phi\partial_t\mc I[\psi]\right)
\,.
\end{equation}
Using the DNLS equation \eqref{eq:DNLS}, integration by parts, and equation \eqref{eq:gauge_der} it is straightforward to get
\begin{equation}\label{20170407:eq4}
\partial_{t}\mc I [\psi] 
=
i \phi'\bar{\phi}-i\phi\bar\phi'
+ \big(2\alpha+\frac{3\beta}{2}\big)  |\phi|^{4}
-2\alpha\mu|\phi|^2
+2\alpha\mu^2
-\big(\frac{3\beta}{4\pi}+\frac{\alpha}{\pi}\big) \| \phi \|_{L^{4}}^{4}
-\frac{i}{\pi}\int_{\mbb T}\phi'\bar\phi
\,.
\end{equation}
Substituting \eqref{20170407:eq1} and \eqref{20170407:eq2} in the DNLS equation \eqref{eq:DNLS} and using
equation \eqref{20170407:eq4} we get the statement.
\end{proof}
\begin{remark}\label{RealCoeff}
Note that $c_1,c_2,c_3$ and $c_4$ in \eqref{eq:GDNLSconst}, and $\Gamma [f]$ in \eqref{eq:GDNLSgamma} are real (indeed, using integration by parts, one can check that
$\int f' \bar f$ is purely imaginary). 
\end{remark}

We call the equation \eqref{eq:GDNLS} with coefficients given
by \eqref{eq:GDNLSconst} and \eqref{eq:GDNLSgamma} the \emph{gauged derivative nonlinear Schr\"odinger} (GDNLS) equation.
We recall that $\Phi_t$ denotes the flow-map of the DNLS equation \eqref{eq:DNLS}. Then, the flow defined by the GDNLS equation
\eqref{eq:GDNLS} is given by
\be\label{eq:G-Flusso}
\Phi_{t,\a}:=\Ga_\a \, \Phi_t \, \Ga_{-\a}\,,
\qquad \alpha\in\mb R\,. 
\ee
\begin{remark}
For the choice of $\alpha=-\beta=-1$ the GDNLS equation \eqref{eq:GDNLS} already appeared in \cite{Herr} and \cite{NOR-BS12}.
\end{remark}

\subsection{Integrals of motion}\label{sec:2.2}
Recall from \cite{KN78} (see also \cite{DSK13,GLV16}) that there exists an infinite sequence of integrals of motion
$\{E_{k}[\psi]\}_{k\in\frac12\N_0}$
for the DNLS equation \eqref{eq:DNLS}.
The first few of them are listed below:
\begin{align}
\begin{split}\label{integrals_DNLS}
E_0 [\psi]
& =
\frac12\|\psi\|_{L^2}^2
\, ,
\\
E_{\frac12}[\psi]
& =
\frac i2\int\psi'\bar\psi
+\frac\beta4\|\psi\|_{L^4}^4  \, ,
\\
E_1[\psi]
& =
\frac12 \|\psi\|_{\dot{H}^1}^2
+\frac{3i}4\beta\int|\psi|^2\psi'\bar\psi
+\frac{\beta^2}4\|\psi\|_{L^6}^6 \, ,
\\
E_{\frac32}[\psi] 
& = \frac i2\int\psi''\bar\psi'
+\frac\beta4\int\left((\psi')^2\bar\psi^2
+8|\psi|^2\psi'\bar\psi'+\psi^2(\bar\psi')^2
\right)
\\
&+\frac{5 i}{4}\beta^2\int|\psi|^4\psi'\bar\psi
+\frac 5{16}\beta^3\|\psi\|_{L^8}^8 \, ,
\\
E_2[\psi]
& = 
\frac12\|\psi\|_{\dot{H}^2}^2
+\frac{5i}4\beta\int|\psi|^2\left(\psi''\bar\psi'-\psi'\bar\psi''\right)
+
\frac{5}{4}\beta^2\int|\psi|^2\left(
(\psi')^2\bar\psi^2+\psi^2(\bar\psi')^2
\right)
\\
&
+\frac{25}{4}\beta^2\int|\psi|^4\psi'\bar\psi'
+\frac{35i}{16}\beta^3\int|\psi|^6\psi'\bar\psi
+\frac{7}{16}\beta^4\|\psi\|_{L^{10}}^{10}  \, .
\end{split}
\end{align}
\begin{remark}
These integrals of motion are slightly different from those appearing in the introduction of
\cite{GLV16}, where there is a typo in the coefficient of $\beta$ of $E_{1}[\psi]$.
\end{remark}
In this section we study the way the sequence of integrals of motion of the DNLS equation $\{E_{k}[\psi]\}_{k\in\frac12\N_0}$
changes under the gauge transformation $\Ga_\a$.

Recall that we have denoted by $\psi$ a solution of the DNLS equation \eqref{eq:DNLS} and we proved in 
Proposition \ref{proposition:GDNLS} that $\phi=\ms G_{\alpha}(\psi)$ 
is a solution of the GDNLS equation \eqref{eq:GDNLS}, for every $\alpha\in\mb R$.
We can rewrite the integrals of motion $E_{k}[\psi]$, $k\in\frac12\N_0$,
in terms of the new variables $\phi^{(k)} := \partial_{x}^k \phi$ in 
the following way
\be\label{eq:def-energie-gauged}
\mc E_{k}[\phi] : = E_{k}[\Ga_{-\alpha} \phi ]  
= E_{k}[\psi]\,,
\qquad
k\in\frac12\N_0
\, .
\ee
Again, we will omit the dependence on $\alpha$ of $\mc E_{k}$ in order to simplify the notations.
Clearly, when $\a = 0$ we have $\mc E_{k}=E_{k}$.
By a direct calculation one can check that the first few integrals of motion \eqref{integrals_DNLS} of the DNLS equation rewrite 
in the new variables $\phi^{(k)}$ as follows:
\begin{align}
\mc E_{0} [\phi] 
& =
\frac{1}{2}\|\phi\|^2_{L^2}\,,
\notag\\
\mc E_{\frac12}[\phi] 
& =
\frac i2\int\phi'\bar\phi
+\frac{1}{4}(2\alpha+\beta)\|\phi\|_{L^4}^4-\pi\alpha\mu^2\,,
\label{legge12}
\\
\mc E_{1}[\phi] 
& = 
\frac12 \| \phi \|_{\dot{H}^{1}}^{2}
+ i \alpha \mu \int \phi \bar\phi'
+ \frac i4 \left( 4\alpha+3\beta\right)\int |\phi|^{2} \phi' \bar{\phi}
\notag\\
&+
\pi\alpha^2\mu^3- \frac\alpha4\left( 4\alpha + 3\beta\right)\mu  \| \phi \|_{L^{4}}^{4} 
+ \frac14( \alpha+\beta)(2\alpha+\beta) \| \phi \|_{L^{6}}^{6} \, ,
\notag
\end{align}
\begin{align}
\mc E_{\frac32}[\phi]
& =
\frac i2\int\phi''\bar\phi'
-\frac32 \alpha\mu\int \phi'\bar\phi'
+\frac14(2\alpha+\beta)\int\left((\phi')^2\bar\phi^2
+\phi^2(\bar\phi')^2
\right)
\notag\\
&
+\frac12(5\alpha+4\beta)\int|\phi|^2\phi'\bar\phi'
+\frac32 i \alpha^2\mu^2\int \phi'\bar\phi
-3i\alpha(\alpha+\beta)\mu\int |\phi|^2\phi'\bar\phi
\notag\\
&
+\frac i4(6\alpha^2+12\alpha\beta+5\beta^2)\int|\phi|^4\phi'\bar\phi
-\pi\alpha^3\mu^4
+\frac32\alpha^2(\alpha+\beta)\mu^2\|\phi\|_{L^4}^4
\notag\\
&
-\frac14\alpha(6\alpha^2+12\alpha\beta+5\beta^2)\mu\|\phi\|_{L^6}^6
+\frac1{16}(2\alpha+\beta)(4\alpha^2+10\alpha\beta+5\beta^2)\|\phi\|_{L^8}^8
\notag\,,
\end{align}
\begin{align}
\mc E_{2}[\phi] &=
\frac{1}{2}\|\phi\|_{\dot{H}^2}^2
-2 i \alpha  \mu  \int \phi''\bar\phi'
+\frac{1}{4} i (6 \alpha +5 \beta )
\int|\phi|^2\left(\phi''\bar\phi'
-\phi'\bar\phi''\right)
\notag\\
&
-\frac{1}{2} i \alpha \int \left((\phi')^2\bar\phi\bar\phi'
-\phi\phi'(\bar\phi')^2\right)
+3 \alpha ^2 \mu ^2 \int \phi '\bar\phi'
-10 \alpha (\alpha +\beta )\mu\int |\phi|^2\phi'\bar\phi'
\notag\\
&
-\frac{1}{4} \alpha (8 \alpha +5 \beta )  \mu 
\int\left((\phi')^2\bar\phi^2+\phi^2(\bar\phi')^2\right)
+\frac{1}{4} (4 \alpha +5 \beta ) (8 \alpha +5 \beta )\int |\phi|^4\phi'\bar\phi'
\notag\\
&
+\frac{5}{4} (\alpha +\beta ) (2 \alpha +\beta )
\int|\phi|^2\left(\phi^2(\bar\phi')^2+\bar\phi^2 (\phi')^2\right)
-2 i \alpha ^3 \mu ^3 \int \phi'\bar\phi
\notag\\
&
+\frac{3}{4} i \alpha ^2 \mu ^2 (4 \alpha +5 \beta )
\int|\phi|^2\left(\phi'\bar\phi
-\phi\bar\phi'\right)
\notag\\
&
-\frac{3}{4} i \alpha \mu\left(4 \alpha ^2+10 \alpha  \beta +5 \beta ^2\right) 
\int|\phi|^4\left(\phi'\bar\phi
-\phi\bar\phi'\right)
\notag\\
&
+\frac{1}{16} i \left(32 \alpha ^3+120 \alpha ^2 \beta +120 \alpha  \beta ^2+35 \beta ^3\right)
\!\int\!\! |\phi|^6\phi'\bar\phi
+\pi \alpha ^4 \mu ^5
\notag\\
&-\frac{1}{2} \alpha ^3 (4 \alpha +5 \beta )\mu ^3\|\phi\|_{L^4}^4
+\frac{3}{4} \alpha ^2 \left(4 \alpha ^2+10 \alpha  \beta +5 \beta ^2\right)\mu ^2 
\|\phi\|_{L^6}^6
\notag\\
&-\frac{1}{16} \alpha \left(32 \alpha ^3+120 \alpha ^2 \beta +120 \alpha  \beta ^2+35 \beta ^3\right)\mu
\|\phi\|_{L^8}^8
\notag\\
&
+\frac{1}{16} (\alpha +\beta ) (2 \alpha +\beta ) \left(4 \alpha ^2+14 \alpha  \beta +7 \beta ^2\right)
\|\phi\|_{L^{10}}^{10}
\notag\,.
\end{align}
We have $\psi(x, t)=\Phi_t(\psi(x, 0))$, as $\psi$ solves \eqref{eq:DNLS}. Hence, from \eqref{eq:def-energie-gauged}
and the conservation of $E_{k}[\psi]$ for the DNLS flow,
it is $\mc E_k[\phi(x, t)]=E_k[\psi(x, 0)]$. In other words, $\mc E_k[\phi]$ is an integral of motion
for the GDNLS equation \eqref{eq:GDNLS} for every $k\in\frac12\N_0$.

We want to give a more detailed description of the integrals of motion $\{\mc E_{k}[\phi]\}_{k\in\N_0}$ defined in \eqref{eq:def-energie-gauged}. We start by reviewing some results from \cite{GLV16} about the structure of the integrals of motion
$E_{k}[\psi]$ of the DNLS equation.

Let $\mc V=[\psi^{(n)},\bar\psi^{(n)}\mid n\in\N_0]$ be the algebra of differential polynomials in the variables $\psi$ and $\bar\psi$.
On the differential algebra $\mc V$ we have the usual polynomial degree, which we denoted by $\deg$, defined by setting $\deg(\psi^{(n)})=\deg(\bar\psi^{(n)})=1$, for every
$n\in\N_0$,
and the usual differential degree, which we denoted by $\dd$, defined by setting $\dd(\psi^{(n)})=\dd(\bar\psi^{(n)})=n$, for every $n\in\N_0$.
For $n\in\N_0$, we also let $\mc V_n=\{f\in\mc V\mid \frac{\partial f}{\partial u^{(m)}}=0,
\text{ for every }m>n, u=\psi\text{ or }\bar\psi\}$.

It is shown in \cite{GLV16} that there exists an infinite sequence $\{h_k\}_{k\in\N_0}\subset\mc V$ such that
the local functionals $\tint h_k$, $k\in\N_0$, are integrals of motion for the DNLS equation \eqref{eq:DNLS}. Moreover
\begin{equation}\label{eq:hk}
h_{k}=\sum_{m=0}^{k}\beta^mh_{k,m}
\,,
\end{equation}
with $\deg(h_{k,m})=2m+2$ and $\dd(h_{k,m})=k-m$, for every $m=0,\dots,k$.
The integrals of motion $\{E_k[\psi]\}_{k\in\N_0}$, introduced at the beginning of this section are defined by
\begin{equation}\label{20170524:eq1}
E_{k}[\psi]=\tint h_{2k}\,,
\qquad k\in\N_0
\,.
\end{equation}
\begin{remark}
We recall from \cite{GLV16} that the Gibbs measures for the DNLS equation are associated to the integrals $\tint h_{2k}$, $k\in\N_0$.
\end{remark}
Let us introduce an integral grading on $\mc V$, which we denote by $\degb$, by setting
$$
\degb(\psi^{(n)})=-\degb(\bar\psi^{(n)})=1\,,
\quad
n\in\N_0
\,.
$$
We also write $\mc V=\bigoplus_{m\in\mb Z}\mc V[m]$, where $\mc V[m]=\{f\in\mc V\mid \degb(f)=m\}$ denotes the space of
homogeneous elements of degree $m\in\mb Z$. 
\begin{lemma}\label{20170522:lem1}
 Let $m\in\mb Z$, and let $f\not\in\mb C$ be such that $f\in\mc V[m]$.
\begin{enumerate}[a)]
\item For every $n\in\N_0$, we have that $\degb(\frac{\partial f}{\partial\psi^{(n)}})=\degb(f)-1$ and 
$\degb(\frac{\partial f}{\partial\bar\psi^{(n)}})=\degb(f)+1$.
\item
We have that $\degb(\partial f)=\degb(f)$. Hence, $\degb(\partial^n f)=\degb(f)$, for every $n\in\N_0$.
\item
If $f\not\in\partial\mc V$, then $\degb(\frac{\delta f}{\delta\psi})=\degb(f)-1$ and $\degb(\frac{\delta f}{\delta\bar\psi})=\degb(f)+1$.
\end{enumerate}
\end{lemma}
\begin{proof}
Part  a) is clear. Part b) follows from part a) and the fact that
$\partial=\sum_{n\in\N_0}(\psi^{(n+1)}\frac{\partial}{\partial\psi^{(n)}}+\bar\psi^{(n+1)}\frac{\partial}{\partial\bar\psi^{(n)}})$.
Part c) follows from parts a) and b), the definition of the variational derivatives
$\frac{\delta}{\delta\psi}=\sum_{n\in\N_0}(-\partial)^n\frac{\partial}{\partial\psi^{(n)}}$
and $\frac{\delta}{\delta\bar\psi}=\sum_{n\in\N_0}(-\partial)^n\frac{\partial}{\partial\bar\psi^{(n)}}$,
and the fact that $\ker\frac{\delta}{\delta \psi}=\ker\frac{\delta}{\delta\psi}=\mb C+\partial\mc V$ (see \cite{BDSK09}).
\end{proof}
\begin{proposition}\label{20170522:prop1}
$h_{k}\in\mc V[0]$ for every $k\in\N_0$.
\end{proposition}
\begin{proof}
The differential polynomials $h_{k}\in\mc V$ are inductively defined (up to total derivatives) by the recurrence relation (2.10) in \cite{GLV16}. In terms
of the variables $\psi=a-ib$ and $\bar\psi=a+ib$, it becomes ($k\in\N_0$)
\begin{equation}\label{20141102:eq3NP}
\left\{
\begin{array}{l}
\displaystyle{
\partial((\psi + \bar\psi) g)
= -2 \psi \partial\left( \frac{\delta h_{k}}{\delta \psi} \right) -2 \bar\psi \partial\left( \frac{\delta h_{k}}{\delta \bar\psi} \right)
}
\\
\displaystyle{
 \frac{\delta h_{k+1}}{\delta \psi} 
=  -i \partial\left( \frac{\delta h_k}{\delta \psi} \right)  - \frac{\beta}{2} \bar\psi(\psi + \bar\psi)  g
}
\\
\displaystyle{
 \frac{\delta h_{k+1}}{\delta \bar\psi} 
=  i \partial\left( \frac{\delta h_k}{\delta \bar\psi} \right)  - \frac{\beta}{2} \psi(\psi + \bar\psi)  g \, ,
}
\end{array}
\right.
\end{equation}
where $g$, $\frac{\delta h_{k+1}}{\delta\psi}$ and $\frac{\delta h_{k+1}}{\delta \bar\psi}$ are uniquely determined by this recurrence
if we know $\frac{\delta h_{k}}{\delta\psi}$ and $\frac{\delta h_k}{\delta\bar\psi}$. We can directly check from equation \eqref{integrals_DNLS}
that $\degb(h_0)=0$. Moreover, since the integrals of motion $\tint h_k$ are non-trivial, we have that $h_k\not\in\partial\mc V$, for every $k\in\N_0$.
Let us assume that $\degb(h_k)=0$ and let us show that $\degb(h_{k+1})=0$. By Lemma \ref{20170522:lem1} b) and c), from the first identity
in \eqref{20141102:eq3NP} we get that $\degb((\psi+\bar\psi)g)=0$. Hence, using Lemma \ref{20170522:lem1} b) the RHS of the second identity
in \eqref{20141102:eq3NP} is homogeneous of degree $\degb(h_k)-1=-1$ by inductive assumption. This force $\degb(h_{k+1})=0$
using Lemma \ref{20170522:lem1} c).
\end{proof}

Let us denote by $\tilde{\mc V}=\mb C[\phi^{(n)},\bar\phi^{(n)}\mid n\in\N_0]$ the algebra of differential polynomials
in the variables $\phi$ and $\bar\phi$. By an abuse of notation we denote with the same symbols
the polynomials and differential gradings of $\mc V$ and $\tilde{\mc V}$. Clearly, on $\tilde{\mc V}$, they are defined by
$\deg(\phi^{(n)})=\deg(\bar\phi^{(n)})=1$, and $\dd(\phi^{(n)})=\dd(\bar\phi^{(n)})=n$, for $n\in\N_0$.
For $n\in\N_0$, we also let $\tilde{\mc V}_n=\{f\in\tilde{\mc V}\mid \frac{\partial f}{\partial u^{(m)}}=0,
\text{ for every }m>n, u=\phi\text{ or }\bar\phi\}$.
Using equation \eqref{eq:gauge_der} we get a linear map
\begin{equation}\label{eq:gauge_map}
\mc V[0]\to\widetilde{\mc V}[\alpha,\mu]
\end{equation}
from the space $\mc V[0]$ to the algebra of polynomials in the variables $\alpha$ and $\mu$ with coefficients in $\widetilde{\mc V}$.
\begin{lemma}\label{lemma:energie-gauge-preciso}
For every $k\in\N_0$, the integrals of motion $\mc E_k[\phi]$ of the GDNLS equation have the form
\begin{equation}\label{MainReprForm}
\mc E_k[\phi]
=\sum_{m=0}^{2k}\sum_{p=0}^{2k-m}\sum_{q=0}^p
\beta^m\alpha^p\mu^q\tint\widetilde{h}_{2k,m;p,q}
\,,
\end{equation}
where $\widetilde h_{2k,m;p,q}\in\tilde{\mc V}$ are such that $\deg(\widetilde{h}_{2k,m;p,q})=2(m+p+1-q)$ and $\dd(\widetilde{h}_{2k,m;p,q})=2k-m-p$.
\end{lemma}
\begin{proof}
By expanding the RHS of equation \eqref{eq:gauge_der} we have that ($n\in\N_0$)
\begin{equation}\label{20170522:eq1}
\psi^{(n)}=e^{-i\alpha\mc I[\psi]}\sum_{p=0}^n\sum_{q=0}^p\alpha^p\mu^q P_{p,q}
\,,
\end{equation}
where $P_{p,q}\in\tilde{\mc V}$ are such that $\deg P_{p,q}=2(p-q)+1$ and $\dd(P_{p,q})=n-p$.
By Proposition \ref{20170522:prop1}, and equations \eqref{eq:def-energie-gauged} and \eqref{20170524:eq1} it follows that $\mc E_k[\phi]=\tint\tilde h_{2k}$, where $\tilde h_{2k}$ is the image of $h_{2k}$ in $\tilde{\mc V}[\alpha,\mu]$.
The result thus follows by substituting \eqref{20170522:eq1} in \eqref{eq:hk}.
\end{proof}
By expanding the RHS of equation \eqref{eq:gauge_der}, we can write ($k\geq2$)
\begin{equation}\label{eq:gauge_der_expanded}
\psi^{(k)}=e^{-i\alpha\mc I(\psi)}
(
\phi^{(k)} - i k \alpha |\phi|^2\phi^{(k-1)} + i k \alpha\mu\phi^{(k-1)}
- i \alpha(|\phi|^2)^{(k-1)}\phi+g_{k}
)
\,,
\end{equation}
where $g_{k}\in\tilde{\mc V}_{k-2}$. 
\begin{lemma}\label{lem:gauge1}
Under the gauge transformation \eqref{gauge_change} we have ($k\geq2$)
$$
\| \psi \|^{2}_{\dot{H}^{k}}=
\|\phi\|^2_{\dot{H}^k}
+2ik\alpha\mu\int\phi^{(k-1)}\bar\phi^{(k)}
+i(k+1)\alpha\int|\phi^2|
\left(\phi^{(k)}\bar\phi^{(k-1)}-\phi^{(k-1)}\bar\phi^{(k)}\right)
+\int r_{k}
\,,
$$
where $r_{k}\in\mc V_{k-1}$.
\end{lemma}
\begin{proof}
Using equation \eqref{eq:gauge_der_expanded} and integration by parts we get
\begin{align}
\begin{split}\label{20161117:eq1}
\|\psi\|_{\dot{H}^k}^2&=
\|\phi\|_{\dot{H}^k}^2
+2ik\alpha\mu\int\phi^{(k-1)}\bar\phi^{(k)}
+ik\alpha\int|\phi|^2\left(\phi^{(k)}\bar\phi^{(k-1)}-\phi^{(k-1)}\bar\phi^{(k)}\right)
\\
&
+i\alpha\int (|\phi|^2)^{(k-1)}\left(\phi^{(k)}\bar\phi-\phi\bar\phi^{(k)}\right)
+\int\phi^{(k)}\bar g_k+\int g_k\bar\phi^{(k)}+\int r_k\,,
\end{split}
\end{align}
where $r_k\in\tilde{\mc V}_{k-1}$. Note that, using integration by part we have
\begin{equation}\label{20161117:eq2}
\int \phi^{(k)}\bar g_k=-\int \phi^{(k-1)}\partial \bar g_k=\int f
\,,
\end{equation}
for some $f\in\tilde{\mc V}_{k-1}$, since $\partial \bar g_k\in\tilde{\mc V}_{k-1}$. Similarly we have
\begin{equation}\label{20161117:eq3}
\int g_k\bar\phi^{(k)}=\int h
\,,
\end{equation}
for some $h\in\tilde{\mc V}_{k-1}$. Using again integration by parts we have
\begin{align}
\begin{split}\label{20161117:eq4}
\int (|\phi|^2)^{(k-1)}\phi^{(k)}\bar\phi
&=
-\int (|\phi|^2)^{(k)}\phi^{(k-1)}\bar\phi
+\int q_k
\\
&
=-\int \phi^{(k)}\phi^{(k-1)}\bar\phi^2
-\int |\phi|^2\phi^{(k-1)}\bar\phi^{(k)}
+\int \tilde q_k\,,
\end{split}
\end{align}
where $q_k,\tilde q_k\in\tilde{\mc V}_{k-1}$. In the second identity we have used 
$$
\partial^k(ab)=\sum_{i=0}^k\binom{k}{i}a^{(i)}b^{(k-i)}
\,,
$$
which holds for any $a,b\in\mc V$. Recall from \cite[Proof of Corollary 2.9]{GLV16} that
\begin{equation}\label{20161117:eq5}
\phi^{(k)}\phi^{(k-1)}\bar\phi^2=\int l
\,,
\end{equation}
for some $l\in\tilde{\mc V}_{k-1}$. Using equation \eqref{20161117:eq4} (and its conjugate version) and \eqref{20161117:eq5}
we get
\begin{equation}\label{20161117:eq6}
\int (|\phi|^2)^{(k-1)}\left(\phi^{(k)}\bar\phi-\phi\bar\phi^{(k)}\right)
=
\int |\phi|^2\left(\phi^{(k)}\bar\phi^{(k-1)}-\phi^{(k-1)}\bar\phi^{(k)}\right)
+\int \tilde{\tilde q}_k
\,,
\end{equation}
where $\tilde{\tilde q}_k\in\mc V_{k-1}$.
The proof is concluded by combining equations \eqref{20161117:eq1}, \eqref{20161117:eq2}, \eqref{20161117:eq3} and
\eqref{20161117:eq6}.
\end{proof}
\begin{lemma}\label{lem:gauge2}
Under the gauge transformation \eqref{gauge_change} we have ($k\geq2$)
$$
\int| \psi |^{2}\psi^{(k)}\bar\psi^{(k-1)}=
\int| \phi |^{2}\phi^{(k)}\bar\phi^{(k-1)}
+\int\tilde r_k
\,,
$$
where $\tilde r_{k}\in\tilde{\mc V}_{k-1}$.
\end{lemma}
\begin{proof}
It is immediate using equation \eqref{eq:gauge_der_expanded}, integration by parts and the fact that
$\partial \tilde{\mc V}_n\subset\tilde{\mc V}_{n+1}$, for every $n\in\N_0$.
\end{proof}
Recall from \cite{GLV16} that, for $k\geq2$, we have 
\begin{equation}\label{eq:E_k}
E_{k}[\psi] = \frac12 \| \psi \|^{2}_{\dot{H}^{k}}
+ i \frac{2k+1}{4} \beta \int |\psi|^2\left(\psi^{(k)}\bar\psi^{(k-1)}-\psi^{(k-1)}\bar\psi^{(k)}\right)
+\int R_{k}
\,,
\end{equation}
where $R_k\in\mc V_{k-1}$.
Hence, we get the following result about the structure of the integrals of motion of the GDNLS equation.
\begin{corollary}\label{DanCor2}
For every $k\geq2$, the integrals of motion
of the GDNLS equation \eqref{eq:E_k} may be written as
\begin{align}\label{FBAK}
\begin{split}
\mc E_k[\phi] &=\frac12\|\phi\|^2_{\dot{H}^k}
+ik\alpha\mu\int\phi^{(k-1)}\bar\phi^{(k)}
\\
&+\frac i4\left((2k+2)\alpha+(2k+1)\beta\right)
\int|\phi|^2\left(\phi^{(k)}\bar\phi^{(k-1)}-\phi^{(k-1)}\bar\phi^{(k)}\right)
+\int R_{k}
\,,
\end{split}
\end{align}
where $R_{k}\in\tilde{\mc V}_{k-1}$.
\end{corollary}
\begin{proof}
It follows by the definition of $\mc E_{k}[\phi]$ given in \eqref{eq:def-energie-gauged},
equation \eqref{eq:E_k} and Lemmas \ref{lem:gauge1} and \ref{lem:gauge2}.
\end{proof}
In Section \ref{sect:Wick}, for every $k\geq2$, we will make a choice of the parameter $\alpha$ in order to simplify the expression of the integrals of motion $\mc E_k[\phi]$ given by \eqref{FBAK} as stated in the next result.
\begin{corollary}\label{Nontiene}
Let $k \geq 2$ and let us fix $\alpha = -\frac{2k+1}{2k+2}\beta$.
Then, the integrals of motion of the GDNLS equation have the form ($\ell\in\N_0$):
\begin{align}\label{eq:energie-dopo-gaugeFacili}
\mc E_\ell[\phi] 
= 
\frac12\|\phi\|^2_{\dot{H}^\ell}
& 
+ i \frac{2k+1 }{2k+2} \ell \beta \mu \int\phi^{(\ell)}\bar\phi^{(\ell-1)}
\\ \nonumber
&
+ \frac{i}{4} \frac{\ell - k}{k+1} \beta
\int|\phi|^2\left(\phi^{(\ell)}\bar\phi^{(\ell-1)}-\phi^{(\ell-1)}\bar\phi^{(\ell)}\right)
+\int R_{\ell}
\, ,
\end{align}
where $R_{\ell}\in\tilde{\mc V}_{\ell-1}$. For $\ell \neq k$ any monomial $h$ giving contribution to $R_{\ell}$ satisfies $\dd(h)\leq 2\ell-1$.
For $\ell = k$ the term $R_k$ can be decomposed as $R_{k,W} + R_{k,P}$, where $R_{k,W}$
is a linear combination of monomials of the following form (and of 
their conjugates)
\begin{equation}\label{Zoo}
\begin{array}{lllll}
1) &  \varphi^{(k-1)}\bar\varphi^{(k-1)}\varphi'\bar\varphi \,,&\quad&
2) & \varphi^{(k-1)}\varphi^{(k-1)}\bar\varphi'\bar\varphi  \,,
\\
3) & \phi^{(k-1)}\bar\phi^{(k-1)}\,,&\quad&
4) &  \varphi^{(k-1)}\varphi^{(k-1)}\bar\varphi^2 \, ,
\\
5)&
\varphi^{(k-1)}\varphi^{(k-1)}\varphi\bar\varphi^3 \,,&\quad&
6) &  \varphi^{(k-1)}\bar\varphi^{(k-1)}|\varphi|^{2m}\,,\quad m=1,2\,, 
\end{array}
\end{equation}
while $R_{k,P}$ is a linear combination of monomials of the form ($m\geq2$)
\begin{equation}\label{OTEASYFINALE}
u^{(\alpha_1)} \dots u^{(\alpha_m)}\,,
\qquad 
k-1\geq\alpha_1 \geq \dots \geq \alpha_m\geq0\,,
\quad \alpha_2\leq k-2 \, ,
\end{equation}
where $u$ may be either $\varphi$ or $\bar{\varphi}$. 
\end{corollary}
\begin{proof}
It follows from Corollary \ref{DanCor2}, letting $\alpha = - \frac{2k+1}{2k+2} \beta$ into \eqref{FBAK}, and Lemma \ref{lemma:energie-gauge-preciso}.
\end{proof}


\section{Truncated GDNLS equation}\label{sect:flows}

We recall that given a function $f:\T \to \C$ we denote by $f(n)$ its $n$-th Fourier coefficient, and that the canonical projections (see \eqref{eq:proj})
$P_N$, for $N\in\N_0$, are defined by
$$
P_N f := \sum_{|n| \leq N}  e^{inx} f(n) \, ,\quad P_{>N}:=\mathbb{I}-P_N\,.
$$
For every $N\in\N_0$, we define the \emph{truncated GDNLS equation} as the following equation
\begin{align}\label{20140509:DNLSapprox}
    i \partial_{t} \phi_N  
     + \phi_N''  +  2 i \alpha \mu[\phi_N] \phi_N'
 &  =
 i c_1 P_{N} ( | \phi_N|^{2}  \phi_N' )
   + i c_2 P_{N} ( \phi_N^{2}  \bar \phi_N' )
\\ \nonumber
& +  c_3 P_{N} ( | \phi_N|^{4}  \phi_N )
   + c_4 \mu[\phi_N] P_{N} ( | \phi_N|^{2}  \phi_N ) 
    + \Gamma [\phi_N] \phi_N \, ,
\end{align}
with initial datum
\begin{equation}\label{20140509:DNLSapproxDatum}
 \phi_N (x, 0) := P_{N} \phi(x,0) \, ,
 \end{equation}
where the constants $c_{j}$ and the functional $\Gamma$ have been defined in \eqref{eq:GDNLSconst}
and  \eqref{eq:GDNLSgamma}. 
Again, for brevity, we simply denote the solutions of \eqref{20140509:DNLSapprox} by $\phi_{N}$ even if they also depend on the parameter~$\alpha$.
Note that any solution of the truncated GDNLS equation \eqref{20140509:DNLSapprox}-\eqref{20140509:DNLSapproxDatum} satisfies 
$\phi_{N} = P_{N} \phi_{N}$. We denote with  
$ \Phi^N_{t,\alpha} $ 
the flow associated to equation~\eqref{20140509:DNLSapprox}.
It is immediate to show that the flow map is locally (in time) well defined looking at the Fourier transform of equation \eqref{20140509:DNLSapprox}.
and solving the associated ordinary differential equation. Then, since the local existence time only depends on the $L^{2}$ norm 
of the initial datum which is a conserved quantity, as shown in Proposition \ref{GaugedMassCons}, also the global well-posedness follows.   

When $\alpha=0$ the truncated GDNLS equation reduces to the truncated DNLS equation
\begin{equation}\label{20140509:UngDNLSapprox}
    i \partial_{t} \psi_{N}  
     + \psi_{N}''  
 =
    i \beta  P_{N} ( | \psi_{N}|^{2}  \psi_{N})' \,,
\end{equation}
with initial datum 
\begin{equation}\label{20140509:UngDNLSapproxDatum}
 \psi_N (x, 0) := P_{N} \psi(x,0) \, .
 \end{equation}
Indeed, for $\alpha=0$ we have $c_{1}= 2\beta$, $c_{2} = \beta$, $c_{3} = c_{4} =c_{5} = \Gamma = 0$. 
We shorten $\Phi^N_{t,\alpha=0} =\Phi^N_{t}$ for the flow associated to \eqref{20140509:UngDNLSapprox}.

Since $\psi_{N} = P_{N} \psi_{N}$, passing to the Fourier coefficients, equation 
\eqref{20140509:UngDNLSapprox} rewrites as a system of ordinary differential equations  
\begin{equation}\label{20140509:DNLSfourier}
\frac{d}{dt} \psi_N(n) = -in^2 \psi_N(n)
+i\beta n  \!\!\!\!\!\!\!  \sum_{\substack{|k|,|\ell|,|m|\leq N \\ k+m=\ell+n}} \!\!\!\!\!\!\! \psi_N(k) \bar\psi_N(\ell) \psi_N(m)
\,, \qquad  |n|\leq N\, ,
\end{equation}
which can be written in a Hamiltonian form and preserves the Euclidean norm. This well-known facts are stated
without proof in the next proposition.
\begin{proposition}\label{HamStruct}
Let $\{\cdot,\cdot\}$ be the Poisson bracket defined through
$$
\{\psi_N(n),\psi_N(n)\}=\{\bar\psi_N(n),\bar\psi_N(n)\}=0\,,
\qquad
\{\psi_N(n),\bar\psi_N(m)\}=-2in\delta_{n,m}\,.
$$
and define $h=E_{\frac12}[P_N\psi]$ (see \eqref{integrals_DNLS}), namely
$$
h
=-\frac12\sum_{|m|\leq N}m\psi_N(m)\bar\psi_N(m)
+\frac\beta4\sum_{\substack{|p|,|q|,|k|,|\ell|\leq N\\p+k=q+\ell}}
\psi_N(p)\bar\psi_N(q)\psi_N(k)\bar\psi_N(\ell)\,.
$$
Then, the system (\ref{20140509:DNLSfourier}) can be written as 
$$
\frac{d}{dt} \psi_N(n)
= \{h,\psi_N(n)\}\,,\quad |n|\leq N
\,.
$$
Moreover, we have
\be\label{eq:cons-L^2-flusso-tronc}
\frac{d}{dt} \sum_{|n|\leq N}|\psi_N(n)|^2=0\,.
\ee
\end{proposition}

Equation \eqref{eq:cons-L^2-flusso-tronc} means that the mass $\mu$ is an integral of motion of the truncated GDNLS equation
when $\a =0$. This actually holds for  
any value of $\a \in \R$, as shown in Proposition \ref{GaugedMassCons}.
As we already observed, as consequence we have that
the truncated flow map
$\Phi_{t,\a}^{N}$, $t \in \R$, is globally well defined for initial data in~$L^2$.

\begin{proposition}\label{GaugedMassCons}
Let $\phi_N$ be a solution of the truncated GDNLS equation \eqref{20140509:DNLSapprox}. Then, we have
$$
\frac{d}{dt} \sum_{|n|\leq N}|\phi_N(n)|^2=0
\,. 
$$
\end{proposition}
\begin{proof}
We want to show that $\frac{d}{dt} \int |\phi_N|^2 = 2\Re \int \bar\phi_N \partial_t\phi_N=0$. 
By the truncated GDNLS equation \eqref{20140509:DNLSapprox} we get
\begin{align}\nonumber
     \partial_{t} \phi_N  
 &    = i \phi_N''    \!-\! 2  \alpha \mu[\phi_N] \phi_N'
 +
  c_1 | \phi_N|^{2}  \phi_N' 
   +  c_2  \phi_N^{2}  \bar \phi_N' 
    \!-\! i  c_3  | \phi_N|^{4}  \phi_N
    \!-\! i c_4 \mu[\phi_N] | \phi_N|^{2}  \phi_N  
    \!-\! i \Gamma [\phi_N] \phi_N
\\\nonumber
& - c_1 P_{>N} ( | \phi_N|^{2}  \phi_N' )
   -  c_2 P_{>N} ( \phi_N^{2}  \bar \phi_N' )
    +i  c_3 P_{>N} ( | \phi_N|^{4}  \phi_N )
    +i c_4 \mu[\phi_N] P_{>N} ( | \phi_N|^{2}  \phi_N ) 
\,.
\end{align}
Using the above formula, Remark \ref{RealCoeff}, and the 
orthogonality relation $\int P_{>N}(f) P_N(g)= 0$, which holds for any $f,g\in L^2(\mb T)$, we are left to show that
\begin{equation}\label{20171023:eq1}
\Re\int \bar\phi_N\partial\phi_N
=\Re\int
i\bar \phi_N\phi_N''
+ c_1 | \phi_N|^{2}  \bar\phi_N \phi_N' 
+  c_2  |\phi_N|^2 \phi_N  \bar \phi_N' 
=0
\,.
\end{equation}
Using integration by parts, it is straightforward to check that $\int f''\bar f$ is real, while $\int |f|^2f'\bar f$ is purely imaginary, for every
$f\in L^2(\mb T)$. Hence, equation \eqref{20171023:eq1} follows by the above considerations and the fact that $c_1$ and $c_2$ are real numbers, see Remark \ref{RealCoeff}.
\end{proof}

Another consequence of Proposition \ref{HamStruct} is that, when $\alpha =0$, the truncated flow relative to \eqref{20140509:DNLSapprox} preserves the 
Lebesgue measure on $M_N(\mathbb{C}^{2N+1})$ (see \eqref{NewLebesgue}). 
Again, this is indeed the case for any $\alpha \in \R$ (see \cite{NOR-BS12} for $\alpha=-1$).
Recall that 
$M_N(A)$ is the cylindrical set \eqref{CylSets} with base $A \in \mathscr{B}(\C^{2N+1})$ and $\mathcal T_{N}$ is the $\sigma$-algebra generated
by these sets.

\begin{proposition}\label{LebMeasPres}
The flow $\Phi^N_{t,\alpha}$ preserves the Lebesgue measure \eqref{NewLebesgue} on~$(M_{N}(\mathbb{C}^{2N+1}), \mathcal T_{N})$.  
\end{proposition}

\begin{proof}
Passing to the Fourier coefficients, the truncated GDNLS equation \eqref{20140509:DNLSapprox}
rewrites as
\begin{equation}
     \frac{d}{dt} \phi_N (n) 
     = F_{n},
     \qquad |n|\leq N\,,
 \end{equation}
where
\begin{equation}\label{decF}
F_{n}:= F^{0}_{n} +  F^{1}_{n} + F^{2}_{n} + F^{3}_{n} + F^{4}_{n}\, ,
\end{equation}
with
\begin{align}
\begin{split}\label{20170611:eq1}
F^{0}_{n} :&= 
i \phi_N''(n)  -  2  \alpha \mu[\phi_N] \phi_N' (n)
\, ,
\\
F^{1}_{n} :&= 
 c_1 ( P_{N} ( | \phi_N|^{2}  \phi_N' ))(n)
+
c_2 ( P_{N} (  \phi_N^{2} \bar \phi_N' ))(n)
\\
&
= \sum_{\substack{|k|, |\ell|, |m| \leq N\\k +\ell= m+n }}
i(c_{1}  \ell-c_2 m)\phi_N (k) \phi_N (\ell)\bar \phi_N (m) \, ,
\\
F^{2}_{n} :&= - i  c_3 (P_{N} ( | \phi_N|^{4}  \phi_N ))(n)
\\
&
=-i c_3
\sum_{\substack{|k|, |\ell|, |m|,|p|,|q| \leq N\\k +\ell+p= m+n+q }}
\phi_N (k) \phi_N (\ell)\bar \phi_N (m)\phi_N (p)\bar \phi_N (q) \, ,
\\
F^{3}_{n} :&= - i c_4 \mu[\phi_N] (P_{N} ( | \phi_N|^{2}  \phi_N ))(n)
=-i c_4\mu[\phi_N]
\sum_{\substack{|k|, |\ell|, |m|\leq N\\k +\ell= m+n}}
\phi_N (k) \phi_N (\ell)\bar \phi_N (m)\, ,
\\
F^{4}_{n} :&=     - i ( \Gamma [\phi_N] \phi_N)(n)
\\
&=-i(\frac{3\alpha\beta}{4\pi}+\frac{\alpha^2}{\pi})\mu[\phi_N^2]\phi_N(n)+i\alpha^2\mu[\phi_N]^2\phi_N(n)
+\frac{i\alpha}{\pi}\phi_N(n)\sum_{|m|\leq N}m|\phi_N(m)|^2
\, .
\end{split}
\end{align}
We will show that 
\begin{equation}\label{20170611:toprove}
\dive F = 0 \, ,
\end{equation}
where the divergence operator is defined as
\begin{equation}\label{Def:dive}
\dive F = \sum_{|n| \leq N} \left( 
\frac{ \partial F_{n} }{ \partial \phi_{N} (n) } + \frac{ \partial \bar{F}_{n} }{ \partial \bar\phi_{N} (n) }  \right) \, .
\end{equation}
analyzing separately all these contributions. Recalling that the Lebesgue measure on 
$M_{N}(\C^{2N+1})$ has density proportional to $\prod_{|n| \leq N} d\phi_{N}(n)d \bar \phi_{N} (n)$, this proves the statement.  

Note that
$$
\mu[\phi_N]=\sum_{|m|\leq N}|\phi_N(m)|^2=\overline{\mu[\phi_N]}
\,,
\qquad
\frac{\partial \mu[\phi_N]}{\partial\phi_N(n)}=\bar \phi_N(n)
\,,
\qquad
\frac{\partial\mu[\phi_N]}{\partial\bar\phi_N(n)}=\phi_N(n)
\,.
$$
Then, from equations \eqref{20170611:eq1} it is straightforward to get
\begin{align}
\begin{split}\label{20170611:eq2}
\frac{\partial F^{0}_{n}}{\partial \phi_N (n)} &= - i n^{2} - 2 i \alpha n \mu[\phi_N] - 2i \alpha n |\phi_N(n)|^{2}
=-\frac{\partial \bar F^{0}_{n}}{\partial \bar \phi_N (n)}\,,
\\
\frac{\partial F^{1}_{n}}{\partial \phi_N (n)}
&=  \sum_{|m| \leq N} i\left((m+n)c_1 - 2m c_2 \right) |\phi_N(m)|^{2}
=-\frac{\partial \bar F^{1}_{n}}{\partial \bar \phi_N (n)} 
\,,
\\
\frac{\partial F^{2}_{n}}{\partial \phi_N (n)}& =-ic_3 \mu[\phi_N^2]
=-\frac{\partial \bar F^{2}_{n}}{\partial \bar \phi_N (n)}
\,,
\\
\frac{\partial F^{3}_{n}}{\partial \phi_N (n)}
&=-i c_4\bar\phi_N(n)
\sum_{\substack{|k|, |\ell|, |m|\leq N\\k +\ell= m+n}}
\phi_N (k) \phi_N (\ell)\bar \phi_N (m)
-2ic_4\mu[\phi_N]^2
\,,
\\
\frac{\partial \bar F^{3}_{n}}{\partial \bar \phi_N (n)}
&=i c_4\phi_N(n)
\sum_{\substack{|k|, |\ell|, |m|\leq N\\k +\ell= m+n}}
\bar\phi_N (k) \bar\phi_N (\ell) \phi_N (m)
+2ic_4\mu[\phi_N]^2
\,,
\\
\frac{\partial F^{4}_{n}}{\partial \phi_N (n)}
&=-i(\frac{3\alpha\beta}{2\pi}+\frac{2\alpha^2}{\pi})\phi_N(n)
\sum_{\substack{|k|, |\ell|, |m|\leq N\\k +\ell= m+n}}
\bar\phi_N (k) \bar\phi_N (\ell)\phi_N (m)
+\frac{i\alpha}{\pi}\sum_{|m|\leq N}m|\phi_N(m)|^2\\
&-i(\frac{3\alpha\beta}{4\pi}+\frac{\alpha^2}{\pi})\mu[\phi_N^2]
+i\alpha(\frac{m}{\pi}+2\alpha\mu[\phi_N])|\phi_N(n)|^2+i\alpha^2\mu[\phi_N]^2
\\
\frac{\partial \bar F^{4}_{n}}{\partial \bar\phi_N (n)}
&=i(\frac{3\alpha\beta}{2\pi}+\frac{2\alpha^2}{\pi})\bar\phi_N(n)
\sum_{\substack{|k|, |\ell|, |m|\leq N\\k +\ell= m+n}}
\phi_N (k) \phi_N (\ell)\bar\phi_N (m)
-\frac{i\alpha}{\pi}\sum_{|m|\leq N}m|\phi_N(m)|^2\\
&+i(\frac{3\alpha\beta}{4\pi}+\frac{\alpha^2}{\pi})\mu[\phi_N^2]
-i\alpha(\frac{m}{\pi}+2\alpha\mu[\phi_N])|\phi_N(n)|^2-i\alpha^2\mu[\phi_N]^2
\,.
\end{split}
\end{align}
Equation \eqref{20170611:toprove} follows immediately from the decomposition of $F$ given in \eqref{decF} and taking the sum
over all $|n|\leq N$ of the terms in \eqref{20170611:eq2}.
\end{proof}

Now we establish a nearness property of the gauged flow to the truncated one which will be used in the sequel. Let us recall that $B^{r}(R)$
is the ball in $H^{r}(\T)$ of radius $R$, centered at zero. We write $\Phi_{t,\alpha} ( A ) := \{ \Phi_{t,\alpha} ( f ) : f \in A  \}$ and 
$\Phi^N_{t,\alpha} ( A ) := \{ \Phi^N_{t,\alpha} ( f ) : f \in A  \}$.
The following proposition is the main achievement of this section.  

\begin{proposition}\label{Prop:nearness}
Let $0 \leq s < r$ with $r > 5/4$ and $R >1$. 
For every~$\varepsilon > 0$, there exists $N^{*}=N^{*}(\e)\in \N$, depending also on $s,r, |\alpha|, R, |\beta|$, such that
\begin{equation}\label{FlowControl}
\Phi^N_{t,\alpha} ( A ) \subseteq  \Phi_{t,\alpha} ( A ) + B^s( \varepsilon )\,, 
\quad  |t| \leq t_{R} \,,
\quad  N > N^{*}\, ,\quad A \subset B^{r}(R) \, ,
\end{equation}
where $0 < t_{R} < 1$ is a sufficiently small threshold which depends on $R,r,|\alpha|, |\beta|$. 
\end{proposition}

We need two accessory lemmas. The first one is a (local in time) integral estimate for solutions of the truncated GDNLS. Similar estimates have been proved for the Benjamin-Ono equation in~\cite{TV14} and the argument easily adapts to GDNLS.

\begin{lemma}\label{TVLemma}
Let $N \in \N \cup \{ \infty \}$.
For all $r > 5/4$ and for any $T \in [0,1]$ we have
\begin{equation}\label{MainEst2Bis}
\int_{0}^{T} \| \phi_{N}' (x, s)\|_{L^{\infty}(\T)} ds 
\leq 
C T^{3/4} 
\left( \sup_{s\in [0,T]} \| \phi_{N}(x, s)\|_{H^{r}(\T)}  + |\beta| T \big( 1+ \sup_{s\in [0,T]} \| \phi_{N}(x, s)\|_{H^{r}(\T)}^{5} \big) \right) \, ,
\end{equation}
where $\phi_{\infty} := \phi$ is a solution to the GDNLS equation \eqref{eq:GDNLS}, $\phi_{N}$ is a solution to the truncated GDNLS equation
\eqref{20140509:DNLSapprox},
and $C$ is a constant depending on $r, |\alpha|, |\beta|$ but uniform over $N$.
\end{lemma}
\begin{proof}

Let $N \in \N \cup \{ \infty \}$ and denote
$$
\Delta_{0}:= P_{1}, \qquad  
\Delta_{j} := P_{2^{j}} - P_{2^{j-1}}, \qquad j\geq1\, ,
$$
Since $\phi_{N}$ is a solution of the truncated GDNLS equation \eqref{20140509:DNLSapprox} or of the non truncated equation \eqref{eq:GDNLS}, 
when we pass to 
the integral formulation and we apply the operator $\Delta_{j} \partial_{x}$, since $[P_{N}, \partial_{x}]= 0$, we arrive to

\begin{equation}\label{FIntEq}
     \partial_{t} \Delta_{j} \phi_{N}'  
 =
     e^{it\partial_{x}^{2}} \Delta_{j} \phi_{N}'
 + 
     \beta \int_{0}^{t} e^{i(t-s)\partial_{x}^{2}}  \Delta_{j} P_N Z(\phi_N) '  \,,
\end{equation}
where we have denoted
\begin{align} \nonumber
Z(\phi_N) 
& 
:= 
c_1   | \phi_N|^{2}  \phi_N' 
   +  c_2     \phi_N^{2}  \bar \phi_N'   
\\ \nonumber   
& 
+   2  \alpha \mu[\phi_N] \phi_N' -i  c_3 P_{N} ( | \phi_N|^{4}  \phi_N )
-i c_4 \mu[\phi_N] P_{N} ( | \phi_N|^{2}  \phi_N ) 
-i \Gamma [\phi_N] \phi_N \, .
\end{align}
Using the algebra property of $H^{s}$, we easily get 
\begin{equation}\label{HugeNonlin}
\| Z(\phi_N) \|_{H^{s}} \leq C ( 1 + \| \phi_{N} \|^{5}_{H^{s+1}} ) \, , 
\end{equation}
for some $C$ that only depends on $s, |\alpha|, |\beta|$.
Now, let $\{t_\ell\}_{\ell=0, \ldots, 2^j}\subset[0,T]$ be
such that 
$$
t_{0}=0\,,\qquad t_{2^j}=T\qquad\text{and}\qquad
t_{\ell}-t_{\ell-1} = \frac{T}{2^{j}} \, .
$$
Looking at the integral equation \eqref{FIntEq}, 
using the Strichartz-type estimates (see \cite{BGT})
\begin{equation}\label{Strich1}
\left( \int_{0}^{T2^{-j}} \| e^{it \partial_{x}^{2}} \Delta_{j}  f \|^{4}_{L^{\infty}(\T)} dt \right)^{\frac{1}{4}} \leq C  \| \Delta_{j} f \|_{L^{2}(\T)} \, ,
\end{equation}
\begin{equation}\label{Strich2}
\left( \int_{0}^{T2^{-j}} \left\| \int_{0}^{t} e^{i(t-s) \partial_{x}^{2}} \Delta_{j}  F(x, s) \right\|^{4}_{L^{\infty}(\T)} dt \right)^{\frac{1}{4}}  
\leq C  \int_{0}^{T2^{-j}} \| \Delta_{j} F(x, s) \|_{L^{2}(\T)} ds \, ,
\end{equation}
valid for all $T \in [0,1]$ and $j\geq1$, for some absolute constant $C$, we can bound 
\begin{align}
\int_{t_{\ell}}^{t_{\ell+1}} & \|  \Delta_{j} \phi_{N}'(x, t) \|_{L^{\infty}(\T)} dt
\leq T^{\frac{3}{4}} 2^{-\frac{3}{4}j} 
\left( \int_{t_{\ell}}^{t_{\ell+1}} \| \Delta_{j} \phi_{N}' (x, t) \|^{4}_{L^{\infty}(\T)} dt \right)^{\frac{1}{4}}
\\ \nonumber
&
\lesssim T^{\frac{3}{4}} 2^{-\frac{3}{4}j} 
\left( \| \Delta_{j} \phi_{N}'(x, t_{\ell}) \|_{L^{2}(\T)} 
+  |\beta|  \int_{t_{\ell}}^{t_{\ell+1}} \!\!\!\!\! \| \Delta_{j}  Z(\phi_N)'(x, s) \|_{L^{2}(\T)} ds \right) 
\\ \nonumber
&
\lesssim T^{\frac{3}{4}} 2^{-j(1+\varepsilon)} 
 \| \Delta_{j} \phi_{N}'(x, t_{\ell}) \|_{H^{\frac{1}{4}+\varepsilon}(\T)} 
 + |\beta| T^{\frac{3}{4}} 2^{-j\varepsilon}  \int_{t_{\ell}}^{t_{\ell+1}} 
 \!\!\!\!\! \| \Delta_{j}  Z(\phi_N)(x, s) \|_{H^{\frac{1}{4}+ \varepsilon}(\T)} ds  \, ,
\end{align} 
for all $\varepsilon >0$. More precisely, 
in the first bound we used H\"older inequality to dominate the~$L^{1}([t_{\ell},t_{\ell+1}])$ norm with the $L^{4}([t_{\ell},t_{\ell+1}])$ norm, 
the second bound is an immediate consequence of~\eqref{FIntEq}-\eqref{Strich2} and the last one is the Bernstein inequality, since we are localising the frequencies 
over the annulus $2^{\ell} <|n|\leq 2^{\ell+1}$. Thus, summing over $\ell=0, \ldots, 2^{j}-1$, we obtain
\begin{align}\label{TVFIN}
\int_{0}^{T} 
& 
\|  \Delta_{j} \phi_{N}'(x, t) \|_{L^{\infty}(\T)} dt \lesssim
T^{\frac{3}{4}}  
\sup_{t \in [0,T]} \| \Delta_{j} \phi_{N}'(x, t) \|_{H^{\frac{1}{4}+\varepsilon}(\T)} \sum_{\ell=0}^{2^{j}-1} 2^{-j(1+\varepsilon)}
\\ \nonumber
& 
+ |\beta| T^{\frac{3}{4}} 2^{-j\varepsilon}  
\sum_{\ell=0}^{2^{j}-1} \int_{t_{\ell}}^{t_{\ell+1}} \| \Delta_{j}  Z(\phi_N)(x, s) \|_{H^{\frac{1}{4}+ \varepsilon}(\T)} ds 
\\ \nonumber
& 
\leq
 |\beta| T^{\frac{3}{4}} 2^{-j\varepsilon} 
\left( 
\sup_{t \in [0,T]} \| \Delta_{j} \phi_{N}'(x, t) \|_{H^{\frac{1}{4}+\varepsilon}(\T)} 
+ |\beta| \int_{0}^{T} \| \Delta_{j} Z(\phi_N)(x, s) \|_{H^{\frac{1}{4}+ \varepsilon}(\T)} ds
\right)\,.
\end{align}  
Since $\phi_{N}' = \sum_{j \in \Z_{+}} \Delta_{j} \phi_{N}'$, the~\eqref{MainEst2Bis} 
follows, for any $r = \frac{5}{4} + \varepsilon$, by~\eqref{TVFIN} and \eqref{HugeNonlin}, via triangle inequality.
\end{proof}


Next we show that the Sobolev norms stay bounded for short time under the evolution of the GDNLS flow. 
We remark that this is interesting 
for solutions of the truncated GDNLS equation~\eqref{20140509:DNLSapprox}, for 
which no further integrals of motion are available.

\begin{lemma}\label{lemma:limitatezza-loc-Hs-nongauge}
Let $N \in \N \cup \{ \infty \}$, $r > 5/4$, $R > 1$ and $\phi(x,0) \in B^{r}(R)$. 
There exists $t_{R} > 0$ such that  
\be\label{eq:limitatezza-loc-Hs-nongauge}
\sup_{|t| \leq t_{R} } \| \Phi^{N}_{t, \alpha} (\phi(x,0))\|_{H^r} \leq  5R    \, ,
\ee
where $\Phi_{t, \alpha}^{\infty} := \Phi_{t, \alpha}$ is the flow associated to \eqref{eq:GDNLS}
and $\Phi^{N}_{t, \alpha}$ is the flow associated to
\eqref{20140509:DNLSapprox}. 
The threshold~$t_{R}$ depends on $R$ and on $|\alpha|, |\beta|$ and $r$ but is independent on $N$.
\end{lemma}
\begin{proof}
Let $N \in \N \cup \{\infty\}$.  
Recall that, when $\alpha = 0$, we have denoted by $\phi_{N} = \Phi_{t}^{N} (\phi(x,0))$ a solution to the truncated DNLS
equation \eqref{20140509:UngDNLSapprox} with initial datum~$P_{N} \phi(x,0)$, while~$\phi_{\infty}$ is a solution to the
DNLS equation \eqref{eq:DNLS} with initial datum 
$P_{\infty} \phi(x,0) = \phi(x,0)$.
We will show that
\begin{equation}\label{ochoBis}
\sup_{|t| \leq  c R^{-\frac{32}{3}} } \, \sup_{ \phi(x,0) \in B^{r}(R)}   \| \phi_{N} \|_{H^r} \leq 5R \, ,
\end{equation}  
for a sufficiently small constant $c > 0$ which depends on $|\alpha|, |\beta|$ and $r$.
First we prove this for $t \in [0, c R^{ -\frac{32}{3} }]$.
We apply the Bessel kernel $J^{r}$, namely the operator with symbol
$(1 + n^2 )^{r/2}$, to the truncated or non truncated GDNLS equation. 
Since $J^{r}$ commutes with 
$\partial_{x}$ and $P_N$ and since $\partial_{x} J^{r} = J^{r+1}$, 
we get
\begin{align}\nonumber
\partial_t J^{r} \phi_N - i J^{r+2} \phi_N  
& =  
c_1  P_N J^{r}( | \phi_N|^{2}  \phi_N' )
   +  c_2 P_N J^{r} (  \phi_N^{2}  \bar \phi_N' )  
+   2  \alpha \mu[\phi_N] J^{r+1} \phi_N 
\\ \nonumber
&
-i  c_3  P_N J^{r} ( | \phi_N|^{4}  \phi_N )
-i c_4 \mu[\phi_N]  P_N J^{r} ( | \phi_N|^{2}  \phi_N ) 
-i \Gamma [\phi_N] P_N J^{r} \phi_N \, .
\end{align}
We take the $L^2$ inner product of this equation against $J^{r} \phi_N = P_{N} J^{r} \phi_N$, and then the real part, so that, using integration by parts
and recalling that $\| f \|_{H^{r}} \simeq \|J^{r} f \|_{L^{2}}$,  
we obtain
\begin{equation}\label{STTB}
\partial_t \|  \phi_N\|^2_{H^{r}} 
\simeq 
\Re (Z_1(\phi_N) + Z_2(\phi_N) + Z_3(\phi_N)) \, ,
\end{equation}
where
$$
Z_1(\phi_N) := c_1 \int   ( J^{r} ( |\phi_N|^2  \phi'_N) ) J^{r}\bar\phi_N\,,
\qquad
Z_2(\phi_N) := c_2 \int   (J^{r} ( \phi_N^2  \bar\phi'_N) ) J^{r}\bar\phi_N \, .
$$
$$
Z_3(\phi_N) := -i  c_3  P_N J^{r} ( | \phi_N|^{4}  \phi_N )
-i c_4 \mu[\phi_N]  P_N J^{r} ( | \phi_N|^{2}  \phi_N ) 
-i \Gamma [\phi_N] P_N J^{r} \phi_N \, ,
$$
and we have used that $\Re \int (J^{r+1} \phi_N) J^{r} \bar \phi_N =0$.
Integrating by parts we get
\begin{equation}\nonumber
Z_1(\phi_N)  = c_1 \widetilde{Z_1}(\phi_N) + c_1 \int ([J^{r}, |\phi_N|^2 ]  \phi'_N)  J^{r}\bar\phi_N  \, ,
\end{equation}
where
$$
\widetilde{Z_1}(\phi_N) := \int  |\phi_N|^2 (J^{r} \phi'_N)  J^{r}\bar\phi_N.
$$
But, integrating by parts, we notice
\begin{equation}
\widetilde{Z_1}(\phi_N)  = - \overline{ \widetilde{Z_1} (\phi_N)} -  \int  (|\phi_N|^2)' |J^{r} \phi_N|^{2}   \, ,
\end{equation}
so that, 
$$
\Re \widetilde{Z_1} (\phi_N) = - \frac{1}{2}  \int  (|\phi_N|^2)' |J^{r} \phi_N|^{2}   \, .
$$
Thus 
$$
\Re Z_1 (\phi_N) = - \frac{c_1}{2}  \int  (|\phi_N|^2)' |J^{r} \phi_N|^{2} + c_1 \int ([J^{r}, |\phi_N|^2 ]  \phi'_N)  J^{r}\bar\phi_N
$$
Using the Kato--Ponce commutator estimate \cite{KatoPonce}:
$$
\| [J^{r}, f] g\|_{L^{2}} \lesssim \|  f' \|_{L^{\infty}} \| g \|_{H^{r-1}} + \| f \|_{H^{r}} \| g \|_{L^{\infty}} \, ,
$$
with $f=|\phi_N|^2$ and $g=\phi'_N$, the Sobolev embedding $H^{r} \hookrightarrow L^{\infty}$, 
and the algebra property of~$H^{r}$, we can bound   
\begin{align}\nonumber
| \Re Z_1 (\phi_N) |   
& 
\leq
C   \Big( \| (| \phi_{N}|^{2})' \|_{L^{\infty}}  
\| J^{r} \phi_{N} \|_{L^{2}}
 +
\|  (|\phi_N|^{2})' \|_{L^{\infty}} \|  \phi_{N}' \|_{H^{r-1}}  
 + 
\| |\phi_N|^{2} \|_{H^{r}} \|  \phi'_N  \|_{L^{\infty}} \Big) \| J^{r} \bar\phi_{N} \|_{L^{2}}
\\ \label{PTI1}
& 
\leq C \| \phi'_{N} \|_{L^{\infty}} \| \phi_{N} \|_{L^{\infty}}  \|  \phi_{N} \|^{2}_{H^{r}}
+
\| \phi'_{N} \|_{L^{\infty}}  \|  \phi_{N} \|^{3}_{H^{r}} 
\lesssim 
\| \phi'_{N} \|_{L^{\infty}}  \|  \phi_{N} \|^{3}_{H^{r}} \, ,
\end{align}
were $C$ are possibly increasing constants which only depend on $r, |\alpha|, |\beta|$. Similarly
\begin{equation}\nonumber
Z_2 (\phi_N)  = c_2   \widetilde{Z_2}(\phi_N) + c_2 \int ([J^{r}, \phi_N^2 ]  \bar\phi'_N)  J^{r}\bar\phi_N 
\end{equation}
where
$$
\widetilde{Z_2}(\phi_N)
:=
\int  \phi_N^2 (J^{r} \bar\phi'_N)  J^{r}\bar\phi_N
$$
and, integrating by parts, we notice
\begin{equation}
\widetilde{Z_2} (\phi_N) = - \widetilde{Z_2} (\phi_N) - \int  (\phi_N^2)' |J^{r} \phi_N|^{2}  \, ,
\end{equation}
namely
$$
\widetilde{Z_2} (\phi_N) = - \frac12 \int  (\phi_N^2)' |J^{r} \phi_N|^{2}
$$
so that
$$
 Z_2 (\phi_N) = - \frac{c_2}{2}  \int  (\phi_N^2)' |J^{r} \phi_N|^{2}  + c_2  \int ([J^{r}, \phi_N^2 ]  \bar\phi'_N)  J^{r}\bar\phi_N \, ,
$$
and (here we let $f=\phi_N^2$ and $g=\bar\phi'_N$) 
\begin{align}\nonumber
| Z_2 (\phi_N) | 
& \lesssim  \Big( \| ( \phi_{N}^{2})' \|_{L^{\infty}}  
\| J^{r} \phi_{N} \|_{L^{2}}
 +
\|  (\phi_N^{2})' \|_{L^{\infty}} \|  \bar\phi'_{N} \|_{H^{r-1}}  
+ 
\| \phi_N^{2} \|_{H^{r}} \|  \bar\phi'_N  \|_{L^{\infty}} \Big) \| J^{r} \bar\phi_{N} \|_{L^{2}}
\\ \label{PTI2}
& \lesssim  \| \phi'_{N} \|_{L^{\infty}}  \| \phi_{N} \|_{L^{\infty}} \|  \phi_{N} \|^{2}_{H^{r}} 
+ \| \phi'_{N} \|_{L^{\infty}}   \|  \phi_{N} \|^{3}_{H^{r}} 
\lesssim 
\| \phi'_{N} \|_{L^{\infty}}  \|  \phi_{N} \|^{3}_{H^{r}}  \, . 
\end{align}
On the other hand, using the algebra property of $H^{r}$, we easily get 
\begin{equation}\label{HugeNonlinBis}
\| Z_3(\phi_N) \|_{H^{r}} \leq C ( 1 + \| \phi_{N} \|^{5}_{H^{r}} ) \, ,
\end{equation}
for some $C$ that only depends on $s, |\alpha|, |\beta|$.
Plugging \eqref{PTI1}, \eqref{PTI2} and \eqref{HugeNonlinBis} into \eqref{STTB}, we arrive to
\begin{align}\label{MainEstPreq}
\partial_t \|  \phi_N \|^2_{H^{r}} 
\leq C^*
(1+ \| \phi'_{N} \|_{L^{\infty}} ) (1+ \|  \phi_{N} \|^{5}_{H^{r}} ) \, ,
\end{align}
for some larger $C^*$ that only depends on $s, |\alpha|, |\beta|$. We take $C^* >2$.
Now we use that $\partial_{t} \eta = \frac{2}{3} a \eta^{5/2}$ for $\eta = (1-a)^{-2/3}$ with the choice $\eta(t) = \|  \phi_N (x, t)\|^2_{H^{r}}$.
Distinguishing the time regimes where $0 \leq \eta(t) \leq 1$ and where $\eta(t) >1$, 
the estimate \eqref{MainEstPreq} implies, via comparison principle,
the following a priori bound
\begin{align}\nonumber 
 \|   \phi_N (x, t) \|^{2}_{H^{r}}  
& \leq
\| \phi_{N}(x, 0) \|^{2}_{H^{r}} 
 \Bigg(
 2 C^* 
\int_{0}^{t}(1+ \| \phi'_{N}(x,s) \|_{L^{\infty}} )ds \, +
\\ \label{questaqua}
& 
\!\!\!\!\!\!\!\!\!\!\!\!\!\!\!\!\!
+
\left(  1 - 3C^*  \| \phi_{N}(x, 0) \|^{3}_{H^{r}}   \int_{0}^{t} \left(1 + \| \phi'_{N}(x, s) \|_{L^{\infty}} \right) ds   \right)^{-\frac{2}{3}}
\Bigg) < 1 +4 \| \phi_{N}(x, 0) \|_{H^{r}} \leq 5R     \, ,
\end{align}
as long as
\begin{equation}\label{TLI}
3 C^*(  \| \phi_{N}(x, 0) \|^{2}_{H^{r}}+  \| \phi_{N}(x, 0) \|^{3}_{H^{r}} )  \int_{0}^{t} (1 + \| \phi'_{N}(x, s) \|_{L^{\infty}} ds ) < \frac12 \, .
\end{equation}
Now we set
$$
X(t, \phi_{N}(x, 0)) := \sup_{s \in [0,t]} \| \phi_{N}(x, t) \|_{H^{r}} \, 
$$
and
we assume that $\left\{ t > 0 : X(t) > 5R \right\}$ is not empty, otherwise the statement follows in the larger time 
regime $t \in [0,\infty]$.
Then we set
$$
T(\phi_{N}(x, 0)) = \inf \left\{ t > 0 : X(t) > 5R \right\} \, .
$$
Our goal is to show that $T(\phi_{N}(x, 0)) > c R^{ -\frac{32}{3} }$ for any $P_{N} \phi(x,0) \in B^{r}(R)$, that would 
imply the statement in the case $t \in [0, c R^{ -\frac{32}{3} }]$. 
By the continuity of $t \to \| \phi_{N}(x, t) \|_{H^{r}}$ it is clear that 
\be\label{eq:contradd-XT}
X(T(\phi_{N}(x, 0)), \phi_{N}(x, 0)) = 5R \, .\ee
 With a little abuse 
of notation we will simply denote these quantities by $T$ and $X(T)$, namely we will omit the dependence on $\phi_{N}(x, 0)$.  
Let now assume that $T \leq c R^{ -\frac{32}{3} }$.   
We deduce a contradiction by \eqref{questaqua}, using the dispersive estimate (\ref{MainEst2Bis}), proven in Lemma \ref{TVLemma}. 
Since $\phi_{N}(x, 0) \in B^{r}(R)$ and $R >1$, 
the \eqref{TLI} would be true once
\begin{equation}\label{TLIBIS}
6 C^* R^{3} \int_{0}^{t} (1 + \| \phi'_{N}(x, s) \|_{L^{\infty}} ds ) < \frac12
\end{equation}
Since $T \in[0, c R^{-\frac{32}{3}}]$, $X(T)= 5R$, the \eqref{MainEst2Bis} gives 
(assume $c < 1$ so that $T < 1$)
\begin{align}\nonumber
6 C^* R^{3}   \int_{0}^{T}
&
( 1 + \| \phi'_{N}(x, s) \|_{L^{\infty}} ) ds  
\leq 
6 C^* R^{3}  (T + C T^{\frac{3}{4}} (X(T) + |\beta| T X^{5}))
\\ \nonumber
&
\leq 6C^* 2^{5}R^{8}  T^{\frac{3}{4}} (1 + C(1 + |\beta|))
= 
6C^* 2^{5} c^{\frac{3}{4}} (1 + C(1 + |\beta|)) < \frac{1}{2}
 \, , 
\end{align}
where in the last inequality we have chosen $c$ sufficiently small. 
Thus \eqref{TLIBIS} and so \eqref{TLI} are satisfied, so that we can apply the \eqref{questaqua} 
with $t \in[0,T]$ and we get $X(T) < 5R$ in contradiction with (\ref{eq:contradd-XT}).

This completes the proof of~\eqref{ochoBis} in the case $t \in [0, cR^{ -\frac{32}{3} }]$.
As for $t \in [- c R^{ -\frac{32}{3} }, 0]$, one can consider the 
equations \eqref{20140509:DNLSapprox}, \eqref{eq:GDNLS} with $\phi_N$ replaced by $\bar\phi_N$, that is satisfied by $\phi_N(-x,-t)$, for which the argument above applies with 
obvious modifications, to show that  
\eqref{ochoBis} holds also when we restrict to negative times thus concluding the proof.
\end{proof}


Now we can prove Proposition \ref{Prop:nearness}. 

\begin{proof}[Proof of Proposition \ref{Prop:nearness}]
Let $N \in \N \cup \{\infty\}$. Assume that we have shown  
\begin{equation}\label{L2small}
\sup_{ |t| \leq t_{R} }  \, \sup_{\phi(x,0) \in B^{r}(R)}
\|  \Phi_{t,\alpha}(\phi(x,0)) - \Phi^N_{t,\alpha}(\phi(x,0))  \|_{L^2} \to 0
\quad \mbox{as} \quad N \to \infty \, ,
\end{equation}
for some $t_{R} \in (0,1)$ that we conveniently choose to be the threshold quantity in 
Lemma~\eqref{lemma:limitatezza-loc-Hs-nongauge}. This allows us to
use the ($N$-uniform) bound \eqref{eq:limitatezza-loc-Hs-nongauge} to get, for $0 \leq s <r$, the following
\begin{align}\nonumber
\|  & \Phi_{t,\alpha}(\phi(x,0))    - \Phi^N_{t,\alpha}(\phi(x,0))  \|_{H^{s}} 
\\ \nonumber
&\leq \|  \Phi_{t,\alpha}(\phi(x,0)) - \Phi^N_{t,\alpha}(\phi(x,0))  \|^{(r-s)/r}_{L^{2}}
\|  \Phi_{t,\alpha}(\phi(x,0)) - \Phi^N_{t,\alpha}(\phi(x,0))  \|_{H^{r}}^{s/r} 
\\ \nonumber
&
\leq (10 R)^{s/r}
 \|  \Phi_{t,\alpha}(\phi(x,0)) - \Phi^N_{t,\alpha}(\phi(x,0))  \|^{(r-s)/r}_{L^{2}} \, ,
\end{align}
for all $\phi(x,0) \in B^{r}(R)$. Thus,
using~\eqref{L2small}, we would get 
\begin{equation}
\sup_{ |t| \leq t_{R} } \,  \sup_{\phi(x,0) \in B^{r}(R)}
\| \Phi_{t,\alpha}(\phi(x,0)) - \Phi^N_{t,\alpha}(\phi(x,0))  \|_{H^{s}} \to 0
\quad 
\mbox{as}
\quad N \to \infty \, ,
\end{equation}
which implies~\eqref{FlowControl}.
It remains to prove the~\eqref{L2small}. 
We consider the difference 
\begin{equation}\label{Def:Delta}
\delta_N(x, t) := \phi(x, t) - \phi_{N}(x, t)  :=  \Phi_{t}(\phi(x,0)) - \Phi^N_{t}(\phi(x,0) )
\end{equation}
that, recalling \eqref{eq:GDNLS} and \eqref{20140509:DNLSapprox}, 
solves the equation
\be\label{DNLSDelta}
\left \{
\begin{array}{rcl}
     \partial_{t} \delta_N - i \delta_N'' + 2 \alpha \mu \delta_N'  & = & 
     - 2 \alpha (\mu[\phi] - \mu[\phi_N]) \phi_N'+
     \sum_{\ell=1}^{7} Z_{\ell}[\phi, \phi_N]+
       P_{>N} \sum_{\ell = 1}^{4} R_{\ell} [\phi_N] \, ,
   \\ 
    \delta_N(x, 0)& = & P_{> N} \phi(x,0) \, ,
\end{array}
\right.
\ee
where, with $c_5 := \frac{3\alpha \beta }{4 \pi} +\frac{\alpha^2}{\pi}$, 
$c_6:= -\alpha^2$ and $c_7 := \frac{i \alpha}{\pi}$, we have denoted  
\begin{equation}\nonumber
Z_1 [\phi,  \phi_{N} ] := 
 c_1 (|\phi|^{2} \phi' - |\phi_{N}|^{2} \phi_{N}')    ,
\qquad 
Z_2 [\phi,  \phi_{N} ] := 
c_2 (\phi^{2} \bar\phi' - \phi_N^{2} \bar\phi_N')   \, ,
\end{equation}
\begin{equation}\nonumber
Z_3 [\phi,  \phi_{N} ] := 
-i c_3 (|\phi|^{4} \phi - |\phi_N|^{4} \phi_N )   ,
\qquad 
Z_4 [\phi,  \phi_{N} ] := 
-i c_4 (\mu[\phi] |\phi|^{2} \phi - \mu[\phi_N] |\phi_N|^{2} \phi_N )   \, ,
\end{equation}
\begin{equation}\nonumber
Z_5 [\phi,  \phi_{N} ] := 
-i c_5 (  \| \phi \|_{L^{4}}^{4} - \| \phi_N \|_{L^{4}}^{4}), 
\qquad
Z_6 [\phi,  \phi_{N} ] := 
-i c_6 ( \mu[\phi]^{2} - \mu[\phi_N]^{2}  ), 
\qquad
\end{equation}
\begin{equation}\nonumber
Z_7 [\phi,  \phi_{N} ] := 
-i c_7 (  \int_{\mathbb{T}} \phi' \bar{\phi} - \int_{\mathbb{T}} \phi_N' \bar{\phi_N}   )
  \, ,
\end{equation}
and
\begin{equation}\nonumber
R_{1}[\phi_N] := (|\phi_{N}|^{2} \phi_{N})'  ,
\qquad 
R_2 [\phi_{N} ] := 
c_2 \phi_N^{2} \bar\phi_N'    \, ,
\end{equation}
\begin{equation}\nonumber
R_3 [\phi_{N} ] := 
-i c_3 |\phi_N|^{4} \phi_N     ,
R_4 [\phi_{N} ] := 
-i c_4 \mu[\phi_N] |\phi_N|^{2} \phi_N    ,
\end{equation}
Now we take the $L^{2}$ inner product of equation~\eqref{DNLSDelta} with~$\delta_N$, and then 
the real part, so that, after integration by parts, we arrive to 
\begin{equation}\label{DiffEnEst}
 \partial_{t} \int |\delta_N|^2 
 =  
     - 2 \alpha  \Re \int (\mu[\phi] - \mu[\phi_N]) \phi_N' \bar \delta_N+ 
     \sum_{\ell=1}^{7} \Re \int  Z_{\ell}[\phi, \phi_N] \bar \delta_N+
       P_{>N}  \Re \int \sum_{\ell = 1}^{4} R_{\ell} [\phi_N] \bar \delta_N \, .
\end{equation}
Now we need to bound the terms on the right hand side of \eqref{DiffEnEst}. 
Notice that, since $A \subset B^r(R)$, by Lemma~\ref{lemma:limitatezza-loc-Hs-nongauge} we have 
that $\phi(x,t), \phi_N(x,t) \subset B^r(5R)$ for $|t| \leq t_R$. Thus
\begin{align}\nonumber
\left| \Re  \int (\mu[\phi] - \mu[\phi_N]) \phi_N' \bar \delta_N  \right| 
& 
= 
\left| \Re  \left( \int ( |\phi| - |\phi_N|) ( |\phi| + |\phi_N|)  ) \right) \left( \int  \phi_N' \bar \delta_N  \right) \right| 
\\ \label{R0}
&
\leq
\| \delta_N \|_{L^{2}}^{2} ( \| \phi \|_{L^{2}} + \| \phi_{N} \|_{L^{2}} ) \| \phi_N' \|_{L^{2}} \lesssim R^2  \| \delta_N \|_{L^{2}}^{2} \, .
\end{align}
Then
\begin{align}\nonumber
\left| \Re \int Z_1[\phi,  \phi_{N} ] \bar \delta_N \right|& 
= |c_1| \left| \Re \int
  ( |\phi|^{2} \phi' - |\phi_N|^{2} \phi_N'  ) \bar \delta_N \right|
\\ \nonumber
& 
= |c_1| \left| \Re \int  |\phi|^{2} \delta_N' \bar \delta_N +  (|\phi| + |\phi_{N}|) (|\phi| - |\phi_{N}|) \phi'_N \bar\delta_N  \right|
\end{align}
and since
$$
\Re \int  |\phi|^{2} \delta_N' \bar \delta_N =  \frac12 \int (|\delta_N|^{2})' |\phi|^{2} = - \frac12 \int |\delta_N|^{2} (|\phi|^{2})' ,
$$
we arrive to
\begin{align}\nonumber 
\left| \Re \int Z_1[\phi,  \phi_{N} ] \bar \delta_N \right|
& \lesssim 
|c_1| \| \delta_N \|_{L^{2}}^{2} \|\phi \|_{L^{\infty}} \| \phi' \|_{L^{\infty}} 
+ |c_1| \| \delta_N \|_{L^{2}}^{2} ( \|\phi \|_{L^{\infty}} + \|\phi_N \|_{L^{\infty}} ) \| \phi'_N \|_{L^{\infty}}
  \\ \label{R1}
& \leq C R (  \| \phi' \|_{L^{\infty}} +  \| \phi'_N \|_{L^{\infty}})\| \delta_N \|_{L^{2}}^{2}  \, ,
\end{align}
where hereafter $C$ denotes several constants, possibly increasing from line to line, 
that only depend on $|\alpha|, |\beta|, r$.
Similarly
\begin{align}\nonumber
\left| \Re \int Z_2 [\phi,  \phi_{N} ]  \bar \delta_N  \right| 
& 
= |c_2| \left| 
\Re \int
( \phi^{2} \bar\phi' - \phi_N^{2} \bar\phi_N' )  \bar \delta_N \right|
\\ \nonumber
&
= |c_2| \left| 
\Re \int  \phi^2 \bar\delta_N'  \bar\delta_N +  (\phi + \phi_N) \bar\phi'_N |\delta_N|^{2}  \right| 
\end{align}
and since
$$
\int  \phi^2 \bar\delta_N'  \bar\delta_N = \frac12 \int \phi^{2} (\bar\delta_N^{2})' = -\frac12 \int (\phi^{2})' \bar\delta_N^{2} \, ,
$$
we arrive to
\begin{align}\nonumber
\left| \Re \int Z_2 [\phi,  \phi_{N} ]  \bar \delta_N  \right|
& \lesssim 
|c_2| \| \delta_N \|_{L^{2}}^{2} \|\phi \|_{L^{\infty}} \| \phi' \|_{L^{\infty}} 
+ |c_2| \| \delta_N \|_{L^{2}}^{2} ( \|\phi \|_{L^{\infty}} + \|\phi_N \|_{L^{\infty}} ) \| \phi'_N \|_{L^{\infty}}
  \\ \label{R2}
& \leq C R (  \| \phi' \|_{L^{\infty}} +  \| \phi'_N \|_{L^{\infty}})\| \delta_N \|_{L^{2}}^{2}  \, .
\end{align}
Moreover
\begin{align}\nonumber
\left| \Re \int Z_3[\phi,  \phi_{N} ] \bar \delta_N \right|& 
= |c_3| \left| \Re \int
  |\phi|^4   |\delta_N|^2  + |\phi_N| ( |\phi|^4 - |\phi_N|^4 )  \bar \delta_N \right|
\\ \nonumber
&
= |c_3| \left| \Re \int
  |\phi|^4   |\delta_N|^2  + |\phi_N| ( |\phi| - |\phi_N| ) ( |\phi| + |\phi_N| ) ( |\phi|^2 + |\phi_N|^2 )  \bar \delta_N \right|
\\ \label{R3}
& 
\lesssim |c_3| \| \delta_N\|_{L^2}^{2} ( \| \phi \|_{L^{\infty}}^4 + \| \phi_N \|_{L^{\infty}}^4   ) 
\leq R^4 \| \delta_N\|_{L^2}^{2} \, ,
\end{align}
and similarly
\begin{align}\nonumber
& \left|  \Re \int Z_4 [\phi,  \phi_{N} ] \bar \delta_N \right|  
= |c_4| \left| \Re \int (\mu[\phi] |\phi|^{2} \phi - \mu[\phi_N] |\phi_N|^{2} \phi_N )  \bar \delta_N \right| 
\\ \nonumber
&
 = |c_4| \left| \Re \int \mu[\phi] |\phi|^2 |\delta_N|^2 
 + 
 \mu[\phi] (|\phi| - |\phi_N|) (|\phi| + |\phi_N|) \phi_N   \bar \delta_N
\right.
\\ \nonumber
&
\left.
 +
\left( \int ( |\phi| - |\phi_N| ) (  |\phi| + |\phi_N|  )\right) 
\left( \Re \int |\phi_N|^2 \phi_N \bar \delta_N \right)
  \right| 
 \\ \nonumber
 & 
 \lesssim 
 |c_4| \|  \delta_N \|_{L^{2}}^{2} \Big( \| \phi_N \|_{L^{2}}^{2} \| \phi_N \|_{L^{\infty}}^{2}  
 + \| \phi_N \|^2_{L^{2}}  (\| \phi\|^{2}_{L^{\infty}} + \| \phi_{N}\|^{2}_{L^{\infty}} ) 
 +   \| \phi\|^{4}_{L^{\infty}} + \| \phi_{N}\|^{4}_{L^{\infty}} \Big) 
\\ \label{R4}
& \leq R^4   \|  \delta_N \|_{L^{2}}^{2} \, .
 \end{align}
We finally estimate 
\begin{align}\nonumber
& \left|  \Re \int Z_5[\phi, \phi_N] \bar \delta_N  \right| 
\\ \nonumber
& = |c_5| \left|   \left( \int |\phi|^4 \right)    \left( \int   |\delta_N|^2 \right)  
+ \left( \int |\phi|^4 - |\phi_N|^4 \right) \left( \int  \bar \phi_N \bar \delta_N \right)
 \right|
 \\ \nonumber
 &
 \leq 
|c_5| \| \phi \|_{L^4}^4 \| \delta \|_{L^{2}}^{2} + |c_5|  \| \phi_N \|_{L^2} \| \delta_N \|_{L^2}
 \int  (|\phi| - |\phi_N|) (|\phi| + |\phi_N|) (|\phi|^2 + |\phi_N|^2) 
\\ \label{R5}
&
\leq C  R^4 \| \delta_N \|_{L^{2}}^{2} + ( \| \phi\|_{L^2}^4 + \| \phi_N \|_{L^2}^4 ) \| \delta_N \|^2_{L^2}  
\leq C R^4 \| \delta_N \|_{L^{2}}^{2}  \, ,
\end{align}
and
\begin{align}\nonumber
& \left|  \Re \int Z_6 [\phi, \phi_N] \bar \delta_N  \right| 
= 
|c_6| \left| \left( \int |\phi|^2 \right)^2 \left( \int   |\delta_N|^2 \right) +  \left( \int |\phi|^2 - |\phi_N|^2  \right)^2 \left( \int \bar \phi_N \bar \delta_N \right)  \right|
\\ \nonumber
&
\lesssim |c_6| R^4  \| \delta_N \|_{L^2}^2 +  R \left( \int ( |\phi| - |\phi_N| ) ( |\phi| + |\phi_N| )  \right) \left( \int \bar \phi_N \bar \delta_N \right) 
\\ \label{R6}
&
\lesssim |c_6| R^4  \| \delta_N \|_{L^2}^2 
+  R^2  ( \| \phi \|_{L^2} + \| \phi_N \|_{L^2} )  \| \phi_N \|_{L^2}  \| \delta_N \|_{L^2}^2  \leq C R^4  \| \delta_N \|_{L^2}^2 \, ,
\end{align}
and 
\begin{align}\nonumber
& \left|  \Re \int Z_7 [\phi, \phi_N] \bar \delta_N  \right| 
= 
|c_7|  
\left| 
\left( \int \phi' \bar \phi \right) \left( \int   |\delta_N|^2 \right) 
+
\left(\int \phi' \bar \phi - \phi'_N \bar \phi_N  \right)
\left( \int \bar \phi_N \bar \delta_N \right)
\right|
\\ \nonumber
&
\lesssim  
|c_7|   \| \phi' \|_{L^2} \| \phi \|_{L^2} \| \delta_N \|^2_{L^2}
+ \left( \int \phi' \bar \delta_N - \delta'_N \bar \phi_N   \right) 
\left( \int \bar \phi_N \bar \delta_N \right)
\\ \nonumber
& 
=
|c_7|   \| \phi' \|_{L^2} \| \phi \|_{L^2} \| \delta_N \|^2_{L^2}
+ \left( \int \phi' \bar \delta_N + \delta_N  \bar \phi_N'    \right) 
\left( \int \bar \phi_N \bar \delta_N \right)
\\ \label{R7}
&  \leq C R^{2} \| \delta_N \|^2_{L^2} + (\| \phi' \|_{L^2} + \| \phi'_N \|_{L^2} ) \| \phi_N \|_{L^2} \| \delta_N \|^2_{L^2}
\leq    C R^{2} \| \delta_N \|^2_{L^2} \, . 
\end{align}
Letting 
$$
\eta(t)=  \sup_{x \in \T} \left ( |\phi'(x, t)|^2 + |\phi'_N(x, t)|^2  \right) \, ,
$$
and
plugging \eqref{R0}, \eqref{R1}, \eqref{R2}. \eqref{R3}, \eqref{R4}, \eqref{R5}, \eqref{R6}, \eqref{R7} 
into \eqref{DiffEnEst} we arrive to
\begin{equation}\label{DiffEnEstFinal}
\partial_{t} \int |\delta(x, t)|^{2} 
\leq  C ( R^4 + R \eta(t)  ) \int |\delta(x, t)|^{2}  
+ P_{>N}  \Re \int \sum_{\ell = 1}^{4} R_{\ell} [\phi_N] \bar \delta_N 
\qquad |t| \leq t_R  \, .
\end{equation}
Using the algebra property of $H^{r-1}$ and Lemma \ref{lemma:limitatezza-loc-Hs-nongauge} 
it is immediate to show that $R_{\ell} [\phi_N(x,t)] \bar \delta_N(x,t)$ belong to $B^{r-1}(C R^5)$ for all $\ell=1, \ldots, 4$ and for $|t| \leq t_R$, 
so that 
$$
\left| P_{>N}  \Re \int \sum_{\ell = 1}^{4} R_{\ell} [\phi_N] \bar \delta_N \right| \leq C R^6 N^{-2(r-1)} \qquad |t| \leq t_R  \, .
$$
so that, using the Gr\"onwall inequality, the estimate~\eqref{DiffEnEstFinal} gives
\begin{align}\nonumber
\int |\delta(x, t)|^{2} 
& 
\leq e^{C R^5 t + C R \int_{0}^{t} \eta(\tau) d\tau  }  \int |\delta(x, 0)|^{2}  + C R^6 N^{-2(r-1)}  \int_0^t e^{C R^5 (t-s) + C R \int_{s}^{t} \eta(\tau) d\tau } ds
\\ \nonumber
&
\leq C e^{C R^6 t  }  R^2 N^{-2r}  + C R^6 N^{-2(r-1)} t e^{C R^6} , \qquad |t| \leq t_R \, ,
\end{align}
that implies the \eqref{L2small}, so that the proof is concluded.
\end{proof}


\section{Asymptotic Stationarity of the integrals of motion}\label{sect:Wick}

Let $k \geq 2$. In this section we only work with the truncated GDNLS under the choice
\be\label{eq:scelta-alpha}\alpha = \alpha_k := -\frac{2k+1}{2k+2}\beta \, ,\ee
namely
\begin{align}\label{GeFInal}
    i \partial_{t} \phi_{N}  
     + \phi_{N}''  - i \beta  \mu \frac{2k+1}{k+1}    \phi_{N}'
 &  =
 i c_1 P_{N} ( | \phi_{N}|^{2}  \phi_{N}' )
   + i c_2  P_{N} ( (  \phi_{N})^{2}  \bar{\phi}_{N}' )
\\ \nonumber
& +  c_3 P_{N} ( | \phi_{N}|^{4}  \phi_{N} )
   + c_4 \mu P_{N} ( | \phi_{N}|^{2}  \phi_{N} ) 
   + \Gamma [\phi_{N}] \phi_{N} \, ,
\end{align}
with initial datum
\begin{equation}
 \phi_{N} (x, 0) := P_{N} \phi(x,0) \, ,
 \end{equation}
 where
\begin{equation}
c_1 = \frac{\beta}{k+1} \,,
\qquad 
c_2 = - \frac{k}{k+1} \beta\,,
\qquad
c_3 =  -\frac{k(2k+1)}{4(k+1)^2}\beta^2 \,,
\qquad
c_4 = \frac{  \beta }{2} \frac{2k+1}{k+1} \, ,
\end{equation}
and  
\begin{equation}
\Gamma [\phi_{N}] 
= \frac{(2k+1)(k-1)}{8\pi(k+1)^2}\beta^2 \| \phi_N \|_{L^{4}}^{4} - \frac{(2k+1)^{2}}{4(k+1)^{2}} \mu^{2} - \frac{i (2k+1)}{2(k+1)} \int_{\mathbb{T}} \phi_N' \bar \phi_N  \, .
\end{equation}
We have denoted the associated flow by 
 $\Phi^N_{t,\a_k} $. We write $\Phi_{t,\a_k} = \Phi^\infty_{t,\a_k}$ for the flow 
 associated to the non truncated equation.

This choice of $\alpha$ simplifies the
(higher order) integrals of motion $\mc E_\ell[\phi_N]$, $\ell\in\mb N_0$, which take the following form for $\ell=k$ (see~\eqref{eq:energie-dopo-gaugeFacili}):
\begin{equation}\nonumber
\mc E_k[\varphi_{N}] =
\frac12\|\varphi_{N} \|^2_{\dot{H}^k} + i \frac{2k+1 }{2k+2} k \beta \mu \int \varphi_{N}^{(k)}  \bar \varphi_{N}^{(k-1)}
+\int R_{k}\, .
\end{equation}

The main goal of this section is to prove the following estimate.
\begin{proposition}\label{Prop: energie} 
Let $k \geq 2$, $0 \leq \ell \leq k$ and (\ref{eq:scelta-alpha}). 
We have
\be\label{eq:goal}
\lim_{N \to \infty} \left\| \frac{d}{dt} \mc E_{\ell}[ \Phi^N_{t,\a_k} (\varphi(x,0)) ] \Big |_{t=0}\right\|_{L^2(\g_k)} \!\!\!\!\!\!\!\!\!\! = \  0 \, .
\ee
\end{proposition}
For $\ell = 0$, equation \eqref{eq:goal} is a consequence of the (stronger) result provided by Proposition \ref{GaugedMassCons}.
\begin{proof}
 Following \cite{Zh01} we define the linear operator $D_{N}$ which acts on the multilinear form 
$\int u^{(\alpha_{1})} \dots u^{(\alpha_{m})}$ according to the Leibniz rule   
\begin{align}\label{TimeDerivativeOp}
D_{N} & \int u^{(\alpha_{1})}  \dots u^{(\alpha_{m})} 
\\ \nonumber
& := i
\sum_{j =1}^{m} 
  \int u^{(\alpha_{1})} \!\!\! \dots P_{> N} \left(
i c_{1}  |u|^{2} u'
+ i c_{2} u^{2} \bar{u}' 
+
c_{3}  |u|^{4} u
+ c_{4} |u|^{2} u \right)^{(\alpha_{j})} \!\!\! \dots u^{(\alpha_{m})} \, . 
\end{align}
Notice that equation \eqref{GeFInal} can be rewritten as
\begin{align}\nonumber
    \partial_{t} \phi_{N}  &
    \! = \!
     i  \phi_{N}''  
      +   \beta  \mu \frac{2k+1}{k+1}    \phi_{N}'
   +
  c_1  | \phi_{N}|^{2}  \phi_{N}' 
  \! +   c_2 (  \phi_{N})^{2}  \bar{\phi}_{N}' 
 \! -i   c_3  | \phi_{N}|^{4}  \phi_{N} 
  \!  -i  c_4 \mu | \phi_{N}|^{2}  \phi_{N}  
   \! -i  \Gamma [\phi_{N}] \phi_{N} 
 \\ \nonumber 
 &  
 - c_1 P_{> N} ( | \phi_{N}|^{2}  \phi_{N}' )
  \!  -    c_2  P_{> N} ( (  \phi_{N})^{2}  \bar{\phi}_{N}' )
 \! + i   c_3 P_{> N} ( | \phi_{N}|^{4}  \phi_{N} )
   \! +i   c_4 \mu P_{> N} ( | \phi_{N}|^{2}  \phi_{N} ) 
      \, ,
\end{align}
where the first line is the GDNLS equation \eqref{eq:GDNLS} for $\alpha = -\frac{2k+1}{2k+2}\beta$, whose flow preserves the
integrals of motion $\mc E_{\ell}$. More precisely
$$
\frac{d}{dt} \mc E_{\ell}[ \Phi_{t,\a_k} (P_N \varphi(x,0)) ] = 0 \, .
$$
Using this, the fact that $\mc E_{\ell}$ are linear combinations of multilinear forms, and the fact that  
$\Phi^N_{0,\a_k} = P_N = \Phi_{0,\a_k} P_N$, one can easily check that
\begin{equation}\label{DN=ddt}
\frac{d}{dt} \mc E_{\ell}[\Phi^N_{t,\a_k}( \phi(x,0))] \Big|_{t=0} = D_{N}  \mc E_{\ell}[P_N  \phi(x,0)] \, ,
\end{equation}
Actually, the \eqref{DN=ddt} holds at any time $t \in \R$. However, we will only use this identity in the case $t=0$.

In order to simplify the notation, until the end of the section we will write $\phi$ in place of $\phi(x,0)$. Nevertheless, 
all the functions we will consider are always calculated at time $t=0$. 
Notice that, by orthogonality
\begin{equation} \label{BilinearGuys1}
D_{N}   \| P_N \phi \|^2_{\dot{H}^\ell} = 0, 
\quad 
\mbox{for all $\ell \in \N_{0}$,}
\qquad
D_{N} \int (P_N \phi)^{(\ell)}( P_N \bar \phi)^{(\ell-1)} =0 ,
\quad \mbox{for all $\ell \in \N$}\, .
\end{equation}
We will show that 
\begin{equation}\label{eq:goalBis}
\lim_{N \to \infty} \left\| D_N \mc E_{\ell}[ P_N \varphi ] \right\|_{L^2(\g_k)}  = \  0 \, ,
\end{equation}
for all $0 \leq \ell  \leq k$, that imples \eqref{eq:goal} by \eqref{DN=ddt}.
\subsubsection*{Case $0 \leq \ell  \leq k-1$}
In this case the gauged integrals of motion have the form (see~\eqref{eq:energie-dopo-gaugeFacili}):
\begin{align}\nonumber
 \mc E_{\ell}  [P_N \phi ] 
&  = \frac12\| P_N \phi \|^2_{\dot{H}^\ell} + i \frac{2k+1 }{2k+2} \ell \beta \mu \int (P_N \phi)^{(\ell)} (P_N \bar \phi)^{(\ell-1)}
+
\\ \nonumber
&
 \frac{i}{4} \frac{\ell - k}{k+1} \beta
\int | P_N \varphi  |^2\left((P_N \phi)^{(\ell)} (P_N \bar \phi)^{(\ell-1)}-(P_N \phi)^{(\ell-1)} (P_N\bar \phi_N)^{(\ell)}\right)
+\int R_{\ell}\,.
\end{align} 
When we apply the operator $D_{N}$, we are
left to consider only its action on the last two terms in the RHS, since the first two terms give no contribution by \eqref{BilinearGuys1}.
Recalling that $R_{\ell}\in\mc V_{\ell-1}$ and using the fact that $\ell -1 \leq k-2$, we have that $R_\ell$ is a a linear combination of monomials of the form 
\begin{equation}\label{20170614:eq1}
D_{N} \int u_{N}^{(\alpha_{1})} \dots u_{N}^{(\alpha_{j})} \dots u_{N}^{(\alpha_{m})}, 
\qquad \alpha_{j} \leq k-2, \quad j=1, \ldots, m \, ,
\end{equation}
where $u_{N}$ is either $P_N \phi $ or $P_N \bar \phi$. 
Thus, using \eqref{TimeDerivativeOp} and reordering the indexes, we have that \eqref{20170614:eq1}
is a  linear combination of
terms of the form
\begin{equation}\label{TrivialForm}
\int u_{N}^{(\beta_{1})} \dots u_{N}^{(\beta_{m-1})} \ P_{> N} ( u_{N}^{(\beta_{m})} \ldots u_{N}^{(\beta_{m+r})} ),
\quad r = 2 \ \mbox{or} \ 4 \, ,
\end{equation}
with $\beta_{j} \leq k-2$, $j=1, \ldots, m+r-1$ and $\beta_{m+r} \leq k-1$. Hence the contribution of these terms
to \eqref{eq:goalBis} is 
$$
\int u_{N}^{(\beta_{1})} \dots u_{N}^{(\beta_{m-1})} \ P_{> N} ( u_{N}^{(\beta_{m})} \ldots u_{N}^{(\beta_{m+r})} )  \, ,
$$
which goes to zero, as $N \to \infty$, 
due to the forthcoming Lemma~\ref{LemmaObv}. 

Finally, we apply $D_N$ to the term
$\int|(P_N \phi) |^2((P_N \phi)^{(\ell)} (P_N \bar \phi)^{(\ell-1)}-(P_N \phi)^{(\ell-1)} (P_N \bar \phi)^{(\ell)})$.
By a simple integration by parts argument, we obtain a linear combination of terms
like~\eqref{TrivialForm} (for $m=3$), whose contribution to \eqref{eq:goalBis} vanishes as $N \to \infty$, and of terms 
of the form 
\begin{equation}\label{TrivialFormBis}
\int u_{N}^{(\beta_{1})}u_{N}^{(\beta_{2})} \ P_{> N} ( u_{N}^{(\beta_{3})} \ldots u_{N}^{(\beta_{3+r})} ),
\quad r = 2 \ \mbox{or} \ 4 \, ,
\end{equation}
with $\beta_{j} \leq k-2$ for all $j=2,\dots,r+2$ and $\beta_{1}, \beta_{3+r} \leq k-1$,
whose contribution to \eqref{eq:goalBis} again vanishes by Lemma \ref{LemmaObv}.

\subsubsection*{Case $\ell = k$}
The gauged integral of motion has the form (see~\eqref{eq:energie-dopo-gaugeFacili}):
\begin{equation}
\mc E_k[P_N \phi] =
\frac12\|P_N \phi \|^2_{\dot{H}^k} + i \frac{2k+1 }{2k+2} k \beta \mu \int (P_N \phi)^{(k)}   (P_N \bar \phi)^{(k-1)}
+\int R_{k}\, ,
\end{equation}
and again, because of \eqref{DN=ddt} and \eqref{BilinearGuys1}, we only need to evaluate $D_{N} \int R_{k}$.
Let us split the terms appearing in $R_{k}$ into two classes according to Corollary \ref{Nontiene} . The first class $R_{k,P}$ contains terms  
of the form
\begin{equation}\label{OTEASY}
u_{N}^{(\alpha_1)}  \dots u_{N}^{(\alpha_m)}, \qquad 
\alpha_j \leq k-2, \quad \text{for }j = 1, \ldots, m-1\,, \qquad \alpha_m \leq k-1 \, ,
\end{equation}
where again $u_{N}$ is either $ P_N \varphi_{N}$ or $P_N \bar \varphi$.
The second class $R_{k,W}$ contains the terms of the form~\eqref{Zoo}. The 
contribution of the terms of $R_{k,P}$ to \eqref{eq:goalBis} is zero in the limit $N \to \infty$, as shown in the forthcoming Lemma~\ref{UsLem}. 
The harder terms of class $R_{k,W}$ are treated in the following way.  First of all, recalling \eqref{BilinearGuys1}, we do not need to consider the 
terms of the kind 3) in \eqref{Zoo}. Then we notice that $D_{N}$ maps a generic element of $R_{k,W}$ into a linear combination of monomials that, 
after integation over $\T$, are, modulo conjugation, of the following two types:
\begin{equation}\label{AllOfThem1}
\int \P \left( (P_N \phi)^{(\beta_{1})}  (P_N \bar \phi)^{(\beta_{5})} (P_N \phi)^{(\beta_{2})} \right)
\P \left(   (P_N \bar \phi)^{(\beta_{6})} (P_N \phi)^{(\beta_{3})}  (P_N \bar \phi)^{(\beta_{7})} (P_N \phi)^{(\beta_{4})}  (P_N \bar \phi)^{(\beta_{8})} \right) \, ,
\ee
with $\b_1+\ldots+\b_8=2k-1$, $\b_1 \leq k$ and $\beta_{j} \leq k-1$, $j = 1, \ldots, 8$, or 
\begin{equation}\label{AllOfThem2}
\int \P \left( (P_N \phi)^{(\beta_{1})}  (P_N \bar \phi)^{(\beta_{4})} (P_N \phi)^{(\beta_{2})} \right)
\P \left(   (P_N \bar \phi)^{(\beta_{5})} (P_N \phi)^{(\beta_{3})}  (P_N \bar \phi)^{(\beta_{6})} \right)  \,,
\end{equation}
with $\b_1+\ldots+\b_6\in\{2k,2k-1\}$, $\b_1 \leq k$ and $\beta_{j} \leq k-1$, $j = 1, \ldots, 6$. 
These terms give a null contribution to \eqref{eq:goalBis}, in the limit $N \to \infty$, as consequence of Lemmas \ref{Lemma:Wick6}
and  \ref{Lemma:Wick4}. 
Thus the proof is concluded.  
\end{proof}
\begin{remark}
Since each integral of motion $\mc E_{\ell}[\phi]$ is a linear combination of multilinear forms,
using the hypercontractivity of the measure $\g_k$, 
the estimate \eqref{eq:goal}  can be promoted to any $L^p(\g_k)$ norm, for $p < \infty$.
\end{remark}

We again recall that in the rest of the section all functions will be evaluated at time $t=0$, even though not explicitly indicated.

\begin{lemma}\label{LemmaObv}
Let $k\geq2$, $m, r \in \N_{0}$ with $m \geq 2$ and $\beta_{j} \leq k-2$ for all $j=1, \ldots m-2$, $j= m, \ldots, m+ r-1$ and $\beta_{m-1}, \beta_{m+r} \leq k-1$. 
Then, letting 
$u_{N}$ denote either $P_N \varphi$ or $P_N \bar\varphi$, we have
\be\label{PAS}
\int u_{N}^{(\beta_{1})}\dots u_{N}^{(\beta_{m-1})} \ P_{> N} ( u_{N}^{(\beta_{m})} \ldots u_{N}^{(\beta_{m+r})} )  \to 0 \quad \mbox{as} \quad N\to\infty \, ,
\ee
$\g_k$-a.s. and in $L^{2}(\g_k)$ mean.
The same is true if we rather assume $\beta_{j} \leq k-2$ for all $j=1, \ldots m + r -2$ and $\beta_{m + r -1}, \beta_{m+r} \leq k-1$.  
\end{lemma}
\begin{proof}
In both cases H\"older inequality and of Sobolev embedding yield
\begin{align}\label{DomConv}
\big| \tint u_{N}^{(\beta_{1})}\dots u_{N}^{(\beta_{m-1})} 
& \!\!
\ P_{> N} ( u_{N}^{(\beta_{m})} \ldots u_{N}^{(\beta_{m+r})} )  \big|
\\ \nonumber
& \lesssim  \| \phi \|^{m-1}_{ H^{k-1}} 
\| P_{>N}(u_{N}^{(\beta_m)}\dots u_{N}^{(\beta_{m+r})}) \|_{L^{2}}  \, .
\end{align}
We decompose
\begin{align}\label{AEdecomp}
\| & P_{>N} u_{N}^{(\beta_m)}\dots u_{N}^{(\beta_{m+r})} \|_{L^{2}}
\\ \nonumber
&
\leq 
\|  P_{>N}(u^{(\beta_m)}\dots u^{(\beta_{m+r})}) \|_{L^{2}} 
+ 
\|  ( u_{N}^{(\beta_m)}\dots u_{N}^{(\beta_{m+r})} - u^{(\beta_m)}\dots u^{(\beta_{m+r})}) \|_{L^{2}} \, .
\end{align}
where $u$ is either $\phi$ or $\bar \phi$, in such a way that $u_N \to u$ in the $H^{r-1}$ topology.
Since
$$
\| u^{(\beta_m)}\dots u^{(\beta_{m+r})} \|_{L^{2}} \leq \| \phi \|^{r+1}_{H^{k-1}} \, , 
$$
recalling that $\| \phi \|_{ H^{k-1}} \leq C$ for $\g_k$-almost any $\phi$,
we see that $ \|  P_{>N}(u^{(\beta_m)}\dots u^{(\beta_{m+r})}) \|_{L^{2}}$ converges to zero 
$\g_k$-a.s. as $N \to \infty$.
Then, taking advantage of the
multilinearity of the monomials involved, also the second term on the right hand side of \eqref{AEdecomp} similarly vanishes.
Indeed, this is clear when $r=0$ and, assuming this true for any integer smaller than $r$, we can show it decomposing
\begin{align}\nonumber
 \| & P_{>N} ( u_{N}^{(\beta_m)}\dots u_{N}^{(\beta_{m+r})} - u^{(\beta_m)}\dots u^{(\beta_{m+r})}) \|_{L^{2}}
\leq
 \|  u_{N}^{(\beta_m)}\dots u_{N}^{(\beta_{m+r-1})} (u_{N}^{\beta_{m+r}} - u^{\beta_{m+r}} )  \|_{L^{2}}
\\ \nonumber
& +
\| ( u_{N}^{(\beta_m)}\dots u_{N}^{(\beta_{m+r-1})} - u^{(\beta_m)}\dots u^{(\beta_{m+r-1})} ) u^{\beta_{m+r}}   \|_{L^{2}} 
\\ \nonumber
& \leq \| \phi  \|^{r}_{H^{k-1}} \| u_{N}^{\beta_{m+r}} \! - \!  u^{\beta_{m+r}}  \|_{L^{2}} + 
\|  u_{N}^{(\beta_m)}\dots u_{N}^{(\beta_{m+r-1})} \! - \! u^{(\beta_m)}\dots u^{(\beta_{m+r-1})} \|_{L^{2}} \| \phi \|_{H^{k-1}} \, ,
\end{align}
that, as a consequence of the induction assumption, clearly goes to zero $\g_k$-a.s., as $N \to \infty$. 

In conclusion, we have shown that $\|  P_{>N}(u_N^{(\beta_m)}\dots u_N^{(\beta_{m+r})}) \|_{L^{2}}$ 
converges to zero $\g_{k}$-a.s., as $N\to\infty$. 
Looking at the \eqref{DomConv} and recalling that $\| \phi \|_{H^{k-1}}$ is $\g_k$-a.s. finite, we have shown that
$ \tint u_N^{(\beta_{1})}\dots u_N^{(\beta_{m-1})}  P_{> N} ( u_N^{(\beta_{m})} \ldots u_N^{(\beta_{m+r})}) $
converges to zero 
$\g_k$-almost surely. Since 
$$
\big| \tint u_{N}^{(\beta_{1})}\dots u_{N}^{(\beta_{m-1})}  P_{> N} ( u_{N}^{(\beta_{m})} \ldots u_{N}^{(\beta_{m+r})} )  \big|^2
\lesssim \| \phi \|^{2(m+r)}_{ H^{k-1}} \, ,
$$
and $\| \phi \|^{2(m+r)}_{ H^{k-1}}$ is integrable with respect to $\g_k$,
also the~$L^{2}(\g_k)$ convergence follows and the proof is concluded. 
The integrability of $\| \phi \|^{2(m+r)}_{ H^{k-1}}$ is a consequence of the Fernique Theorem; see \cite{kuo} Chapter 3 Theorem 3.1.
\end{proof}


\begin{lemma}\label{UsLem}
Let $k\geq2$, $2 \leq m \in \N$ and $\alpha_1 \leq \dots \leq \alpha_{j} \ldots \leq \alpha_m$, with $\alpha_m\leq k-1$ and $\alpha_{m-1}\leq k-2$. Then, letting 
$u_N$ denote either $P_N \varphi $ or $P_N \bar\varphi$, we have
\be\label{PASBis}
D_{N} \int u_N^{(\alpha_1)}\dots u_N^{(\alpha_m)}\longrightarrow0 \quad \mbox{as }N\to\infty \, ,
\ee
$\g_k$-a.s. and in $L^{2}(\g_k)$ mean.
\end{lemma}
\begin{proof}
Recalling \eqref{TimeDerivativeOp}, after reordering the indexes, the expression \eqref{PASBis} is   
a linear combination of
terms of the form
\begin{equation}\label{GenForm}
\int u_N^{(\beta_{1})} \dots u_N^{(\beta_{m-1})} \ P_{> N} ( u_N^{(\beta_{m})} \ldots u_N^{(\beta_{m+r})} ) \, ,
\end{equation}
where $3 \leq r \in \N$ and either 
\begin{equation}\label{FinalmCasoDiff}
\beta_{j} \leq k-2 \quad \mbox{for $j=1, \ldots, m+r-1$}, \qquad \beta_{m+r} \leq k \, , 
\end{equation}
or 
\begin{equation}\label{FinalmCasoFacil}
\beta_{j} \leq k-2 \quad \mbox{for $j=1, \ldots m-2$, $j= m, \ldots, m+ r-1$}, \qquad \beta_{m-1}, \beta_{m+r} \leq k-1 \, .
\end{equation}
If we are in the case of \eqref{FinalmCasoFacil}, we can simply apply Lemma \ref{LemmaObv} to deduce the statement.
In the case \eqref{FinalmCasoDiff}, after integration by parts we reduce to a linear combination of terms of the
form \eqref{GenForm}, but with $\beta_{j}$ which satisfies \eqref{FinalmCasoFacil} or such that
$$
\beta_{j} \leq k-2 \quad \mbox{for all $j=1, \ldots m + r -2$}, \qquad  \beta_{m + r -1}, \beta_{m+r} \leq k-1 \, .
$$
Since we are still under the assumptions of Lemma \ref{LemmaObv}, we can use it to control these terms too, so that the proof is concluded.
\end{proof}

To evaluate the contribution of \eqref{AllOfThem1}, \eqref{AllOfThem2} to \eqref{eq:goalBis} we need a different approach, based on the Wick theorem,
that we shall use in the following form. Let $\ell \in \N$ and $S_{\ell}$ be the 
symmetric group on $\{1,\dots,\ell\}$, whose elements are denoted by $\s$. Then
\be\label{eq:Wick}
\E \left[ \prod_{j=1}^{\ell}  \phi(n_j)  \bar \phi(m_j)  \right]=
\sum_{\sigma \in S_{\ell}}\prod_{j=1}^{\ell} \frac{\d_{m_j,n_{\sigma(j)}}}{1+n_{\sigma(j)}^{2k}} \, ,
\ee
where we recall that $\E$ is the expectation with respect to the Gaussian measure $\g_k$, so that
\begin{equation}\label{ExpVsL2}
\E \left[ \, f \bar f  \, \right] = \| f \|^{2}_{L^{2}(\g_k)} \, .
\end{equation} 
We say that $\sigma$ contracts the pairs of indexes $(m_j,n_{\sigma(j)})$.  

\begin{lemma}\label{Lemma:Wick6}
Let $N \in \N$, $0 \leq \beta_1 \leq k$ and $0 \leq \beta_{j} \leq k-1$ for $j = 2, \ldots, 8$. 
We have that
\begin{align}\label{Eq:Wick1Giuseppe}
\bigg\| \int \P  &\left( (P_N \phi)^{(\beta_{1})}  (P_N \bar \phi)^{(\beta_{5})} (P_N \phi)^{(\beta_{2})} \right)
\\ \nonumber
&
\P \left(  (P_N \bar \phi)^{(\beta_{6})} (P_N \phi)^{(\beta_{3})}  (P_N \bar \phi)^{(\beta_{7})} (P_N \phi)^{(\beta_{4})}  (P_N \bar \phi)^{(\beta_{8})} \right) \bigg\|_{L^{2}(\gamma_{k})}\lesssim \frac{\ln N}{\sqrt N} \, .
\end{align}
\end{lemma}

\begin{proof}
Writing
\begin{align}\nonumber
\P & \left( (P_N \phi)^{(\beta_{1})}  (P_N \bar \phi)^{(\beta_{5})} (P_N \phi)^{(\beta_{2})} \right)
\P \left(  (P_N \bar \phi)^{(\beta_{6})} (P_N \phi)^{(\beta_{3})}  (P_N \bar \phi)^{(\beta_{7})} (P_N \phi)^{(\beta_{4})}  (P_N \bar \phi)^{(\beta_{8})} \right) 
\\ \nonumber
& =    \!\!\!\!\!\!\!\!\!\!\!\!\!\!\!\!
\sum_{\substack{ | n_{1} - m_{1} + n_{2} | > N \\ | m_{2} - n_{3} + m_{3} - n_{4} + m_{4}| > N  \\ |n_j|, |m_j|  \leq N, \, j = 1, \ldots, 4} }
\!\!\!\!\ 
\left(\prod_{j=1}^{4} (i n_j)^{\beta_j}  (-i m_j )^{\beta_{j+4}} \phi (n_j)   \bar \phi(m_j) \! \right)
\! e^{i(n_1 - m_1 + n_2 -m_2 + n_3 - m_3 + n_4 - m_4)x}
\end{align}
and its conjugate as
\begin{align}\nonumber
\P & \left(  (P_N \bar \phi)^{(\beta_{1})}  (P_N \phi)^{(\beta_{5})}  (P_N \bar \phi)^{(\beta_{2})} \right)
\P \left( (P_N \phi)^{(\beta_{6})}  (P_N \bar \phi)^{(\beta_{3})}  (P_N \phi)^{(\beta_{7})}  (P_N \bar \phi)^{(\beta_{4})}  (P_N \phi)^{(\beta_{8})} \right) 
\\ \nonumber
& = \!\!\!\!\!\!\!\!\!\!\!\!\!\!\!\!
\sum_{\substack{ | m_{5} - n_{5} + m_{6} | > N \\ | n_{6} - m_{7} + n_{7} - m_{8} + n_{8}| > N  \\ |n_j|, |m_j| \leq N, \, j = 5, \ldots, 8} } 
\!\!
\left(\prod_{j=5}^{8} (i n_j )^{\beta_j}  (-i m_j )^{\beta_{j-4}} \phi(n_j)   \bar \phi(m_j) \! \right)
\! e^{-i(m_5 - n_5 + m_6 - n_6 + m_7 - n_7 + m_8 - n_8)x} \, ,
\end{align}
using \eqref{ExpVsL2}, 
we see that the square of the l.h.s. of  (\ref{Eq:Wick1Giuseppe}) can be written in the compact form
$$
\sum_{A_{N}} 
 \prod_{j=1}^8 \left[n_{j}^{\beta_{j}} m_{j}^{\beta_{[j+4]_8}}\right]\E\left[\prod_{j=1}^8  \phi(n_j)   \bar \phi(m_j)\right]\,,
$$
where, letting 
$ n = (n_1, \ldots, n_8)$, $m = (m_{1}, \ldots, m_{8})$ 
we have defined
\begin{equation}\nonumber
A_{N} := \left\{ (n, m) : 
\begin{array}{l}
\sum_{j=1}^4 n_{j}=\sum_{j=1}^4 m_{j}\,,\quad \sum_{j=5}^8n_{j}=\sum_{j=5}^8 m_{j}
\\
|n_{j}| \leq N, \ |m_{j}| \leq N, \quad j =1, \ldots, 8\,,
\\
| n_{1} - m_{1} + n_{2} | > N\,, \quad 
| m_{5} - n_{5} + m_{6}| > N\,,
\\
| m_{2} - n_{3} + m_{3} - n_{4} + m_{4}| > N \,,
\\
|  n_{6} - m_{7} + n_{7} - m_{8} + n_{8}| > N 
\end{array}
\right\} \, .
\end{equation}
Since
$n_{1} = \sum_{j=1}^{4} m_j - \sum_{j=2}^{4} n_{j}$ and $m_5 = \sum_{j=5}^{8} n_{j} - \sum_{j=6}^{8} m_1$, the 
condition $| n_{1} - m_{1} + n_{2} | > N$ reduces to $| m_{2} - n_{3} + m_{3} - n_{4} + m_{4}| > N$ 
and $| m_{5} - n_{5} + m_{6}| > N$ reduces to $|n_{6} - m_{7} + n_{7} - m_{8} + n_{8}| > N$. Thus we can rewrite $A_{N}$
as
\begin{equation}\nonumber
A_{N} = \left\{ (n, m) : 
\begin{array}{l}
\sum_{j=1}^4 n_{j}=\sum_{j=1}^4 m_{j}, \quad \sum_{j=5}^8n_{j}=\sum_{j=5}^8 m_{j}
\\
|n_{j}| \leq N, \ |m_{j}| \leq N, \quad j =1, \ldots, 8\,,
\\
| m_{2} - n_{3} + m_{3} - n_{4} + m_{4}| > N \,,
 \quad |  n_{6} - m_{7} + n_{7} - m_{8} + n_{8}| > N  
\end{array}
\right\} \, .
\end{equation}
Now we use the Wick formula (\ref{eq:Wick}), with $\ell=8$, and obtain
\begin{align}\nonumber
&
\sum_{A_{N}} 
\prod_{j=1}^8 \left[n_{j}^{\beta_{j}} m_{j}^{\beta_{[j+4]_8}}\right]\E\left[\prod_{j=1}^8  \phi(n_j)   \bar \phi(m_j)\right]
=
\sum_{A_{N}}  \prod_{j=1}^8 \left[n_{j}^{\beta_{j}} m_{j}^{\beta_{[j+4]_8}}\right]
\sum_{\sigma \in S_{8}}\prod_{j=1}^{8} \frac{\d_{m_j,n_{\sigma(j)}}}{1+n_{\sigma(j)}^{2k}}
\\ \label{TTSH}
& 
=
\sum_{\sigma \in S_{8}}
\sum_{A^{\sigma}_{N}}  \prod_{j=1}^8 \left[n_{j}^{\beta_{j}} n_{\sigma(j)}^{\beta_{[j+4]_8}}\right]
\prod_{j=1}^{8}\frac{1}{1+n_{\sigma(j)}^{2k}}
=
\sum_{\sigma \in S_{8}} 
\sum_{A^{\sigma}_{N}} 
\prod_{j=1}^8\frac{  n_j^{  \b_{j} + \b_{[\sigma^{-1}(j)+4]_8 } } } {1+n_j^{2k}}
\end{align}
where 
\begin{equation}\nonumber
A^{\sigma}_{N} := \left\{ n :
\begin{array}{l}
\sum_{j=1}^4n_{j} =\sum_{j=1}^4 n_{\s(j)}, \quad
\sum_{j=5}^8n_{j}=\sum_{j=5}^8 n_{\s(j)} ,
\\
|n_{j}| \leq N, \quad j =1, \ldots, 8\,,
\\
| n_{\s(2)} - n_{3} + n_{\s(3)} - n_{4} + n_{\s(4)}| > N,
\\
|  n_{6} - n_{\s(7)} + n_{7} - n_{\s(8)} + n_{8}| > N 
\end{array}
\right\}   \, .
\end{equation}
We will bound the sum over $\sigma$ in \eqref{TTSH}, term by term, distinguishing the permutations which 
satisfies $\sigma(5) = 1$ and the ones such that~$\sigma(5) \neq 1$.
\subsubsection*{Case $\sigma(5) = 1$}
Noting that $\b_{[\sigma^{-1}(1)+4]_8} = \b_{[5+4]_8} = \beta_1$, 
we
see that 
$\beta_{j} + \b_{[\sigma^{-1}(j)+4]_8 } \leq 2k-2$ for all $j = 2, \ldots, 8$. 
Using this, we can estimate  
\begin{equation}\label{PPAS}
\sum_{A^{\sigma}_{N}} 
\prod_{j=1}^{8}\frac{|n_j|^{ \b_{j}+ \b_{[\sigma^{-1}(j)+4]_8 } }  } {1+n_j^{2k}}
\lesssim
\sum_{ \mathcal{A}^{\sigma}_{N}} 
\prod_{j=2}^{8} \frac{ 1 } {1+ n_j^2}
\, ,
\end{equation}
where 
\begin{equation}\nonumber
\mathcal{A}^{\sigma}_{N} \! := \! \left\{ (n_2, \ldots, n_8) \! : \! 
| n_{\s(2)} - n_{3} + n_{\s(3)} - n_{4} + n_{\s(4)}| > N , 
\, |  n_{6} - n_{\s(7)} + n_{7} - n_{\s(8)} + n_{8}| > N  
\right\}  
\end{equation}
and we have removed $n_1$ by the summation thanks to the relation $n_{1} = \sum_{j=1}^{4} n_{\sigma(j)} - \sum_{j=2}^{4} n_{j}$
and to the fact that $\mathcal{A}^{\sigma}_{N}$ is independent on $n_1$.
Since $\sigma(5) =1$, it is clear that we can cover $\mathcal{A}^{\sigma}_{N}$ with the following sets
$$
\mathcal{A}^{\s, \ell}_{N} := \{ (n_2, \ldots, n_8) : |n_{\ell}| >  N/5 \}\, , 
$$
where $\ell \in\{3,4,\sigma(2),\sigma(3),\sigma(4)\}$,
and 
that the sum over $\mathcal{A}^{\sigma}_{N}$ in \eqref{PPAS} is bounded by the total contribution of the sums over
these sets. We will show that 
\begin{equation}\label{WWSOT}
\sum_{ \mathcal{A}^{\s, \s(2)}_{N}} 
\prod_{j=2}^{8} \frac{ 1 } {1+ n_j^2} \lesssim \frac{1}{N} \, ,
\end{equation}
all the other sums can be treated in the same way. We have
\begin{equation}\nonumber
\sum_{ \mathcal{A}^{\s, \s(2)}_{N}} 
\prod_{j=2}^{8} \frac{ 1 } {1+ n_j^2}  
 \lesssim
 \sum_{|n_{\s(2)}| > N/5} \frac{1}{n_{\s(2)}^{2}}  
 \prod_{\substack{ j = 2, \ldots, 8 \\ j \neq \sigma(2)} }  
 \sum_{ n_j  \in \Z} \frac{1}{ 1+ n_j^{2}} 
\lesssim \frac{1}{N} \, .
\end{equation}
\subsubsection*{Case $\sigma(5) \neq 1$}
We can write
\begin{equation}\label{DaRif}
\sum_{A^{\sigma}_{N}} 
\prod_{j=1}^{8}\frac{|n_j|^{ \b_{j}+ \b_{[\sigma^{-1}(j)+4]_8 } }  } {1+n_j^{2k}} = 
\sum_{A^{\sigma}_{N}} 
\left(
\frac{|n_1|^{ \b_{1} + \b_{[\sigma^{-1}(1)+4]_8 } }  } {1+n_1^{2k}} 
\frac{|n_{\sigma(5)}|^{ \b_{1} + \b_{\sigma(5)} }} {1+n_{\sigma(5)}^{2k}} 
\prod_{ \substack{j=2, \ldots 8, \\ j \neq \sigma(5)  } } \frac{|n_j|^{\b_{j} + \b_{[\sigma^{-1}(j)+4]_8 } }  } {1 + n_j^{2k}}  
\right)
\, ,
\end{equation} 
and we notice, $j=\ell, \ldots, 8$:
\begin{equation}\label{LogWick}
\b_{1} + \b_{[\sigma^{-1}(1)+4]_8 } \leq 2k - 1, 
\quad 
\b_{1} + \b_{\sigma(5)} \leq 2k-1 ,
\quad
\b_{\ell} + \b_{[\sigma^{-1}(\ell)+4]_8} \leq 2k -2, \ \ell \notin \{1, \s(5)\} 
\, .
\end{equation}
Then we cover 
$A^{\sigma}_{N}$ with the following sets, $\ell=1, \ldots, 4:$
$$
A^{\s, \ell}_{N} := \left\{ n : |n_{\ell}| >  N/5, \ |n_{1}|, |n_{\s(5)}| \leq N, \ \sum_{j=1}^4n_{j} =\sum_{j=1}^4 n_{\s(j)} \right\} , 
$$
and $\ell = 5, \ldots, 8$:
$$
A^{\s, \ell}_{N} := \left\{ n : |n_{\ell}| >  N/5, \ |n_{1}|, |n_{\s(5)}| \leq N, \ \sum_{j=5}^8 n_{j} =\sum_{j=5}^8 n_{\s(j)} \right\} .
$$
It is
clear that the sum over $A^{\sigma}_{N}$ in \eqref{DaRif} is bounded by the total contribution of the sums over
these sets.
Looking at \eqref{DaRif} and using the first condition in \eqref{LogWick} and $n_1 = \sum_{j=1}^4 n_{\s(j)} - \sum_{j=2}^4 n_{j}$ we can bound
\begin{equation}\nonumber
 \sum_{A^{\sigma, 1}_{N}} 
\prod_{j=2}^{8}\frac{|n_j|^{ \b_{j}+ \b_{[\sigma^{-1}(j)+4]_8 } }  } {1+n_j^{2k}} \lesssim 
\frac{1}{N}  
\sum_{|n_{\s(5)}| \leq N}  {\frac{1}{1 + |n_{\s(5)}|}} 
\prod_{ \substack{j=2, \ldots 8, \\ j \neq \sigma(5)  } } \sum_{n_j \in \Z}\frac{ 1 } {1 + n_j^{2}} \lesssim \frac{\ln N}{N}  
\, ,
\end{equation} 
Similarly, using the second condition in \eqref{LogWick} and $n_{\s(5)} = \sum_{j=5}^8 n_{j} - \sum_{j=6}^8 n_{\s(j)}$, we get
\begin{equation}\nonumber
 \sum_{A^{\sigma, \s(5)}_{N}} 
\prod_{j=2}^{8}\frac{|n_j|^{ \b_{j}+ \b_{[\sigma^{-1}(j)+4]_8 } }  } {1+n_j^{2k}} \lesssim 
\frac{1}{N}  
\sum_{|n_{1}| \leq N}  {\frac{1}{1 + |n_{1}|}} 
\prod_{ \substack{j=2, \ldots 8, \\ j \neq \sigma(5)  } } \sum_{n_j \in \Z}\frac{ 1 } {1 + n_j^{2}} \lesssim \frac{\ln N}{N}  
\, .
\end{equation} 
If $j \notin \{ 1, \s(5) \}$, we can use the last condition in \eqref{LogWick} to bound
\begin{align}\nonumber
 \sum_{A^{\sigma, \ell}_{N}} 
 & 
 \prod_{j=1}^{8}\frac{|n_j|^{ \b_{j}+ \b_{[\sigma^{-1}(j)+4]_8 } }  } {1+n_j^{2k}}
  \\ \nonumber
 & 
 \lesssim  
\sum_{|n_{1}| \leq N}  {\frac{1}{1 + |n_{1}|}} 
\sum_{|n_{\s(5)}| \leq N}  {\frac{1}{1 + |n_{\s(5)}|}}
\sum_{|n_{\ell}| > N/5} \frac{1}{ n_{\ell}^{2}} 
\prod_{ \substack{j=1, \ldots 8, \\ j \notin{1, \ell, \s(5)}  } } \sum_{n_j \in \Z}\frac{ 1 } {1 + n_j^{2}} 
\lesssim \frac{(\ln N)^2}{N}  
\, ,
\end{align} 
that concludes the proof.
\end{proof}


\begin{lemma}\label{Lemma:Wick4}
Let $N \in \N$, $0 \leq \beta_1 \leq k$ and $0 \leq \beta_{j} \leq k-1$ for $j = 2, \ldots, 6$.
We have that
\begin{equation}\nonumber
\left\| \int \P \left( (P_N \phi)^{(\beta_{1})}  (P_N \bar \phi)^{(\beta_{4})} (P_N \phi)^{(\beta_{2})} \right)
\P \left(  (P_N \bar \phi)^{(\beta_{5})} (P_N \phi)^{(\beta_{3})}  (P_N \bar \phi)^{(\beta_{6})}\right) \right\|_{L^{2}(\gamma_{k})}
\lesssim \frac{\ln N}{\sqrt N} \, .
\end{equation}
\end{lemma}

\begin{proof}
Adapted directly from the proof of Lemma \ref{Lemma:Wick6}.
\end{proof}

\begin{remark}
Here we gave explicit rates of convergence only for those terms in the integrals of motion that we have to treat by the Wick theorem, where we have only $L^2$ convergence, but not on all the other ones, where convergence is a.s. and in $L^2$. Indeed we aimed more to emphasise the different nature of the terms than to provide explicit rates of convergence for all of them. Of course one could use the Wick theorem to obtain an explicit rate of convergence for all the terms. It should be clear that the slowest possible decay is given by the terms dealt by the last two lemmas.
\end{remark}


\section{Invariant Measures}\label{Sect:Proof}

Let $k\geq2$ and $\alpha \in \R$. We define a sequence of measures which approximate the weighted 
Gaussian measure relative to $\g_k$, with weight (we omit the dependance on $\a$ which is irrelevant at this stage)
$$
\prod_{m=0}^{k-1}
\chi_{R_m}\left(\mc E_m[f]\right)
\exp(-\widetilde Q_{k}[f]) \, ,
$$ 
where
\begin{equation}\label{Q}
\widetilde Q_{k}[f] : =  \mc E_{k}[f] - \frac12\|f\|^2_{\dot{H}^k} \,.
\end{equation}
To do so we set
\begin{equation}\nn
\tilde\rho_{k,N}(A) 
 :=  \int_{\mc M_N (A)}
\left(
\prod_{m=0}^{k-1}
\chi_{R_m}\left(\mc E_{m} [P_N \psi] \right)
\right)
e^{- \widetilde Q_k [P_N \psi] }\gamma_{k}(d\psi), 
\,.
\end{equation}
where $\chi_{R_m}$ are smooth non negative cut-off functions like in definition \eqref{eq:GIBBS-measure}
and $\mc M_N (A)$ is defined in \eqref{SecCyl}.

%
In Proposition \ref{sec5prop1} we prove the existence of the weak 
limit of $\{\tilde\rho_{k,N}\}_{N\in\N}$ for any $k\geq2$ along with its integrability properties. We first need the following 
\begin{lemma}\label{Wick2}
The sequence $\{\int P_N f^{(k-1)} P_N \bar f^{(k)}\}_{N \in \N}$ is Cauchy in $L^{2}(\g_k)$, with
$$
\left\| \int P_N f^{(k-1)} P_N \bar f^{(k)} - \int  f^{(k-1)}  \bar f^{(k)}  \right\|_{L^{2}(\g_k)} \lesssim \frac{1}{N} \, .
$$
\end{lemma}
\begin{proof}
Let $N, M \in N$ with $N < M$. 
We will show that 
\begin{equation}\label{EasyW}
\left\| \int P_M f^{(k-1)} P_M \bar f^{(k)} -  P_N f^{(k-1)} P_N \bar f^{(k)}  \right\|_{L^{2}(\g_k)}
 \lesssim \frac{1}{N} \, ,
\end{equation}
that is enough to deduce the statement. 
We write
$$
P_M f^{(k-1)} P_M \bar f^{(k)} = \sum_{|n_1|, |m_1| \leq M} (i n_1)^{k-1} (-i m_1)^{k}  f(n_1)  \bar f(m_1) e^{i(n_1 - m_1)x} \, , 
$$
and its conjugate as
$$
P_M \bar f^{(k-1)} P_M  f^{(k)} = \sum_{|m_2|, |n_2| \leq M} (-i m_2)^{k-1} (i n_2)^{k}  \bar f(m_2)  f(n_2) e^{-i(m_2 - n_2)x}  \, , 
$$
so that 
$$
\left\| \int P_M f^{(k-1)} P_M \bar f^{(k)} -  P_N f^{(k-1)} P_N \bar f^{(k)}  \right\|^2_{L^{2}(\g_k)}
=
\sum_{A_{N,M}} 
n_1^{k-1} m_1^k n_2^k m_2^{k-1}
\E\left[\prod_{j=1}^2  f(n_j)   \bar f(m_j)\right] \, ,
$$
where letting
where, letting 
$n = (n_1, n_2)$, $m=(m_1, m_2)$, 
we have defined
\begin{equation}\nonumber
A_{N,M} := \left\{ (n, m) : 
n_1 = m_1 , \ n_2=m_2,
\quad
|n_{j}| \leq M, \ |m_{j}| \leq M, \ j =1, 2, 
\quad 
\max( | n_{1}|, |n_2|  > N )
\right\} \, .
\end{equation}
Thus, using the Wick formula \eqref{eq:Wick} with $\ell = 2$, we get
$$
\sum_{A_{N,M}} 
n_1^{k-1} m_1^k n_2^k m_2^{k-1}
\E\left[\prod_{j=1}^2  f(n_j)  \bar f(m_j)\right]
=
\sum_{A^1_{N,M}} 
 \frac{ n_1^{2k-1} n_2^{2k-1} }{ (1+n_1^{2k})(1+n_2^{2k}) }
+
\sum_{A^2_{N,M}}  
\frac{ n_1^{4k-4}  }{ (1+n_1^{2k})^2 }
$$
where
\begin{equation}\nonumber
A^1_{N,M} := \left\{ n : 
\
|n_{1}|, |n_2| \leq M, 
\quad 
\max( | n_{1}|, |n_2|  > N )
\right\} ,
\qquad
A^2_{N,M} := \left\{ n_1 : \ 
N< | n_{1}| \leq M 
\right\} \, .
\end{equation}
Since, using the symmetry $(n_1, n_2) \leftrightarrow (- n_1, n_2)$, the sum over~$A^1_{N,M}$ is zero and the 
sum over~$A^1_{N,M}$ is bounded by a multiple of $N^{-2}$, we have proved the \eqref{EasyW}.
\end{proof}
We are now ready to analyse the convergence of the sequence $\{\tilde\rho_{k,N}\}_{N\in\N}$.
\begin{proposition}\label{sec5prop1}
Given $\beta \in \mathbb{R}$, $k\geq 2$ and $R_1,\ldots,R_{k-1}>0$,  there exists a sufficiently small~$R_0$, which depends on $|\a|,|\beta|, k, R_1, \ldots, R_{k-1}$, so that $\r_{k,N}\rightharpoonup \tilde\r_k$, with
\be\label{eq:gauged-Gibbs}
\tilde\rho_{k}(A)=\int_A \prod_{m=0}^{k-1}
\chi_{R_m}\left(\mc E_m[f]\right)
\exp(-\widetilde Q_{k}[f])\g_k(df), \qquad\text{for every }A\in\B(L^2(\T)) \, ,
\ee
Moreover the Radon-Nykodim derivative $\frac{d\tilde\r_{k}}{d\g_k}$ belongs to $L^2(\g_k)$.
\end{proposition}

The proof goes along the same lines of the main theorem in \cite{GLV16}, with some additional considerations. We will skip some details, just referring to our previous work when the arguments are close enough. 

\begin{proof}
By equation \eqref{Q} and Corollary \ref{DanCor2} we have
\begin{align}\label{Q2}
\widetilde Q_{k}[f] 
& = - ik\alpha\mu\int f^{(k)}\bar f^{(k-1)}
\\ \nonumber
& 
+ \frac i4 \left((2k+2)\alpha+(2k+1)\beta\right)
\int|f|^2( f^{(k)}\bar f^{(k-1)} - \bar f^{(k)} f^{(k-1)} )
+\int R_{k}[f]\,,
\end{align}
where $R_k[f]\in\mathcal V_{k-1}$. The difference with the integrals of motion for DNLS is given by the first addendum on the r.h.s. of \eqref{Q2}. 
Indeed the second term and the remainders $R_k$ have the same structure 
as in~\cite{GLV16}. Thus it is straightforward to adapt that argument to show that 
one can 
choose~$R_0$ small enough in such a way that as $N \to \infty$
\begin{align}\label{GiaFatto}
 \prod_{m=0}^{k-1}
 &
\chi_{R_m}^{1/2}\left(\mc E_m[P_N f]\right)
\exp\left(-\widetilde Q_{k}[P_N f] - ik\alpha\mu\int P_N f^{(k)}P_N \bar f^{(k-1)} \right)
\\ \nonumber
&
\to 
\prod_{m=0}^{k-1}
\chi_{R_m}^{1/2}\left(\mc E_m[f]\right)
\exp\left(-\widetilde Q_{k}[f] - ik\alpha\mu\int f^{(k)} \bar f^{(k-1)} \right)
\quad 
\g_k-\mbox{a.s.} \, ,\nn
\end{align}
and the limit lies in $L^4(\g_k)$. Thus it remains to show the same for the remaining term in the exponential, namely that for $R_0$ sufficiently small
\begin{align}\label{TIWWRWTP}
 \prod_{m=0}^{k-1}
 &
\chi_{R_m}^{1/2}\left(\mc E_m[P_N f]\right)
\exp\left(ik\alpha\mu\int P_N f^{(k)}P_N \bar f^{(k-1)} \right)
\\ \nonumber
&
\to 
\prod_{m=0}^{k-1}
\chi_{R_m}^{1/2}\left(\mc E_m[f]\right)
\exp\left(ik\alpha\mu\int f^{(k)} \bar f^{(k-1)} \right) 
\quad
\g_k-\mbox{a.s.}
\end{align} 
and the limit lies in $L^4(\g_k)$. 

Combining Lemma \ref{Wick2}, Chebyshev's inequality and Borel-Cantelli lemma we 
prove that the sequence $\int P_N f^{(k-1)} P_N \bar f^{(k)}$ converges $\g_k$-a.s. .
This also implies the convergence in measure of $\exp\left(p_0\int P_N f^{(k-1)} P_N \bar f^{(k)}\right)$ for any $p_0\in\R$.  
Arguing as in \cite[Lemma 5.3]{GLV16} and using the Borel-Cantelli lemma, we can show that indeed 
$\exp\left(\int P_N f^{(k-1)} P_N \bar f^{(k)}\right)\to\exp\left(\int f^{(k-1)} \bar f^{(k)}\right)$ $\g_k$-a.s. .
We want 
to prove this limit to be in $L^4(\g_k)$ for sufficiently small $R_0$. 
From that, recalling~\eqref{GiaFatto}, we deduce that the measures $\tilde\rho_{k,N}$ converges weakly to a limit a.c. w.r.t. $\g_k$, whose density is the $L^2(\g_k)$-limit of the densities w.r.t. the finite dimensional Gaussian measures of $\tilde\rho_{k,N}$ (we refer to the proof of Theorem 1.1 in \cite{GLV16} for more details).

We write for $N\in\N$
$$
\left\|\exp\left(ik\alpha\mu\int P_Nf^{(k-1)} P_N\bar f^{(k)}\right)\right\|^4_{L^4(\g_k)}=\int \g_k(df) \exp\left(4ik\alpha\mu\int P_Nf^{(k-1)} P_N\bar f^{(k)}\right)\,.
$$
We then pass to the Fourier coefficients and change variables
$$
\int \g_k(df) \exp\left(4ik\alpha\mu\int P_Nf^{(k-1)} P_N\bar f^{(k)}\right)=\E[e^{4k\alpha\mu\sum_{|n|\leq N}|g_n|^2 q_n}]=\prod_{|n|\leq N}\E[e^{4k\alpha\mu|g_n|^2 q_n}]\,,
$$
where $\{\bar g_n,g_n\}_{|n|\leq N}$ are i.i.d. standard complex random variables and we shortened 
$$q_n:=\frac{n^{2k-1}}{1+n^{2k}} \, .$$ 
Let now assume $\alpha > 0$ (the case $\alpha <0$ is analogous). 
Squared Gaussian random variables are exponentially distributed, thus
$\E[e^{4k\alpha\mu|g_n|^2 q_n}]\leq 1$ for $n\leq0$
and for $n\geq1$ and $4k\alpha\mu<1$ there is an absolute constant $c>0$ such that  
$$
\E[e^{4k\alpha\mu|g_n|^2 q_n}]\leq e^{c(4k\alpha\mu q_n)^2}\,
$$
(see for instance \cite[Lemma 5.15]{ver}). We conclude that for $4k\alpha\mu<1$
$$
\left\|\exp\left(ik\alpha\mu\int P_Nf^{(k-1)} P_N\bar f^{(k)}\right)\right\|^4_{L^4(\g_k)}\leq e^{2(4k\alpha\mu)^2\sum_{n=1}^N q_n^2}\,,
$$
which is bounded uniformly in $N$. 
\end{proof}

\begin{remark}
The convergence in \eqref{GiaFatto}, \eqref{TIWWRWTP} and so that of 
$$
\prod_{m=0}^{k-1}
\chi_{R_m}\left(\mc E_m[P_N f]\right)
\exp(-\widetilde Q_{k}[P_N f])\g_k(df)
$$
can be promoted to convergence in $L^{p_0}$ for any $p_0 \in [1, \infty)$, as long as we choose $R_0$
sufficiently small. For details we refer to the proof of Theorem 1.1 in \cite{GLV16}. 
\end{remark}

\begin{remark}
Even though we preferred to give here a direct proof, the last statement can be also proven just invoking the Girsanov-Ramer theorem \cite{ramer}. The pull-back of $\g_k$ under the anticipative Hilbert-Schimdt map
$$
\phi\mapsto\phi+a\left(\int_0^x\phi-\frac{1}{2\pi}\int_0^{2\pi}\phi\right)\,,\quad a\in\R\,,
$$
is absolutely continuous w.r.t. $\g_k$ and the density is given by the standard Gaussian change of variables under a shift. 

\end{remark}

\begin{remark}\label{rmk:integrability}

With a glance to the proof of \cite[Proposition 5.4]{GLV16}, we realise that Proposition~\eqref{sec5prop1}
remains valid if one considers the modified densities
$$
\left(
\prod_{m=0}^{k-1}
(\chi_m)_{R_m}\left(\mc E_{m}[\phi]\right)
\right)
e^{-\widetilde Q_{k}(\phi)}\gamma_{k}(d\phi)
$$
where the cut-off functions $\chi_m$ are not necessarily all the same. Notice that we have already used this fact in the proof of Proposition, 
when we have replaced $\chi_m$ by $\chi_m^{1/2}$.  
\end{remark}

Until the end of the Section we work with the choice
\begin{equation}\label{AlphaChioce}
\alpha = \alpha_k := -\frac{2k+1}{2k+2}\beta \, ,
\end{equation}
except in Lemma \ref{Lemma:t0}, that holds for all $\a \in \R$.
The main goal will be to prove that $\tilde\r_k$ is invariant under $\Phi_{t,\a_k}$, that is the flow associated to the GDNLS equation
\eqref{eq:GDNLS} under the choice \eqref{AlphaChioce}.  
\begin{proposition}\label{Prop:invarianza-gauged}
Let $k\geq2$. For any $t \in \R$ we have that
$$
\tilde{\rho}_{k}(A) = \tilde{\rho}_{k}(\Phi_{t,\a_k}(A)) \, , 
$$
for every  $A \in \mathscr{B}(L^{2}(\T))$ such that $A \subseteq H^{r}(\T)$ with $5/4 < r < k-1/2$.
\end{proposition}

Then if we set for every $A\in\B(L^2(\T))$
\begin{equation}\label{ultimo}
\hat \r_k(A):= \tilde\rho_{k} (\Ga_{\a_k}(A)) \, ,
\end{equation}
where $\Ga_{\a_k}(A) := \{ \Ga_{\a_k} f : f \in A \}$,
we have
\begin{align}\nn
\hat\r_k(\Phi_t(A)) 
 = \tilde\r_k(\Ga_{\a_k}(\Phi_t(A)))
& 
=
\tilde\r_k(\Ga_{\a_k}(\Phi_t(\Ga_{-\a_k} (\Ga_{\a_k}(A)))))
\\ \nn
& 
= \tilde\r_k(\Phi_{t,\a_k}(\Ga_{\a_k}(A)))=\tilde\r_k(\Ga_{\a_k}(A))
= \hat\r_k(A)\,,
\end{align}
that is $\hat\r_k$ is invariant along the DNLS flow. In the first equality we used \eqref{ultimo}, in the second equality we used \eqref{eq:gauge_properties},
in the third equality we used the definition of the gauged flow given in~\eqref{eq:G-Flusso}, in the fourth equality we used 
Proposition \ref{Prop:invarianza-gauged}, and finally we used \eqref{ultimo} again.
This establishes the existence of the invariant measures $\hat\r_k$ stated in our main Theorem \ref{Th:Main}.

The proof of Proposition \ref{Prop:invarianza-gauged} needs two intermediate lemmas.

\begin{lemma}\label{Prop.invarianza-loc} 
Let $k\geq 2$. Then we have
\be\label{eq:lolcal}
\lim_{N \to \infty} \sup_{t \in \R}\sup_{ A  \in \B(L^{2}(\T))} \left|\frac{d}{dt}  \tilde \rho_{k,N}(\Phi^N_{t,\a_k}(A))  \right|=0 \, . 
\ee
\end{lemma}
Since an explicit representation of the measure $\tilde\rho_{k,N}$ is available only at $t=0$, the following observation by Tzvetkov and Visciglia is crucial to 
prove Lemma \ref{Prop.invarianza-loc}.  
\begin{lemma}\label{Lemma:t0}
Let $k\geq 2 $. Then for all $\a, t \in \R$ we have
\be\label{eq:TVt0}
\sup_{A \in \B(L^{2}(\T)) } \left|\frac{d}{dt} \tilde \rho_{k,N} (  \Phi^N_{t,\a}(A)  ) \right|
\leq \sup_{A \in \ms \B(L^{2}(\T)) }\left|\frac{d}{dt}  \tilde \rho_{k,N}(  \Phi^N_{t,\a}(A)  )\Big|_{t=0}\right|  \,.
\ee
\end{lemma}
We omit the proof, that can be done directly following \cite{TV13b} (Proposition 5.4, step 2).

\begin{proof}[Proof of Lemma \ref{Prop.invarianza-loc}]
Let $A \in \ms \B(L^{2}(\T))$.
The proof is based on the identity
\begin{align}\nonumber
  \tilde\rho_{k,N} & (\Phi^N_{t,\a_k}(A)) = 
   \int_{\mc M_N(\Phi^N_{t,\a_k}(A))} \left( \prod_{m=0}^{k-1}
\chi_{R_m}\left(\mc E_m [P_N f] \right) \right)
\exp(-\widetilde Q_{k}[P_N f]) \g_k (df)
  \\ \nonumber
 & = \frac{1}{Z_N}  \int_{P_N(A)} \left( \prod_{m=0}^{k-1}
\chi_{R_m}\left(\mc E_m [\Phi^N_{t,\a_k} g] \right) \right)
\\ \nonumber
& \ \ \ \ \ \ \cdot
\exp \left(-\widetilde Q_{k}[\Phi^N_{t,\a_k} g] - \frac12 ( \| \Phi^N_{t,\a_k} g \|_{\dot{H}^k}^2 \right) \exp(-\frac12 \| \Phi^N_{t,\a_k} g \|_{L^2}^2  )  
 \bigg( \prod_{|n| \leq N} d \Phi^N_{t,\a_k}f d \overline{ \Phi^N_{t,\a_k}f } \bigg)
 \\ \nonumber
 & = \frac{1}{Z_N}  \int_{P_N(A)} \left( \prod_{m=0}^{k-1}
\chi_{R_m}\left(\mc E_m [\Phi^N_{t,\a_k} g] \right) \right) 
\\ \nn
& \ \ \ \ \ \ \cdot 
\exp \left(- \mc E_k [\Phi^N_{t,\a_k} g]  \right) \exp(-\frac12 \| \Phi^N_{t,\a_k} g \|_{L^2}^2  )  
 \bigg( \prod_{|n| \leq N} d \Phi^N_{t,\a_k}f d \overline{ \Phi^N_{t,\a_k}f } \bigg)
\end{align}
where in the first identity, using that $(\Phi^N_{t,\a_k})^{-1} = \Phi^N_{t,-\a_k}$, we have changed variables $f = \Phi^N_{t,\a_k} g$ and used the explicit representation \eqref{Def:gammaK}
and the product structure of the Gaussian measure. Notice that we are identifying with a little abuse of notation the set $P_N(A) := \{P_N f : f \in A \}$ with 
the set of the relative Fourier coefficients $\{ (f(-N), \ldots, f(N))  : P_N f \in A  \}$. The second identity just follows by the definition~\eqref{Q} of $\widetilde Q_k[\cdot]$.  
Using 
$$
\frac{d}{dt} \exp\left( -\frac12 \| \Phi^N_{t,\a_k} g \|_{L^2}^2  \right)  
 \bigg( \prod_{|n| \leq N} d \Phi^N_{t,\a_k}f d \overline{ \Phi^N_{t,\a_k}f } \bigg) = 0
$$
that is a direct consequence of 
Propositions~\ref{GaugedMassCons} and~\ref{LebMeasPres} we can 
compute
\begin{align}\nonumber
\frac{d}{dt}    \tilde\rho_{k,N} & (\Phi^N_{t,\a_k}(A)) \Big|_{t=0}
\\ \nonumber
&= \frac{1}{Z_N} \sum_{\ell=0}^{k-1}\int_{P_N(A)} \left( \frac{d}{dt}\mc E_\ell [\Phi^N_{t,\a_k}f] \right)
\chi'_{R_\ell}(\mc E_\ell [\Phi^N_{t,\a_k}f]) \left( \prod_{\substack{m=0 \\ m \neq \ell}}^{k-1}\chi_{R_m}(\mc E_m [\Phi^N_{t,\a_k}f]) \right)
\\ \nonumber
& \ \ \ \ \ \ \cdot
\exp \left(- \mc E_k [\Phi^N_{t,\a_k} g]  \right) \exp \left(-\frac12 \| \Phi^N_{t,\a_k} g \|_{L^2}^2  \right)  
 \bigg( \prod_{|n| \leq N} d \Phi^N_{t,\a_k}f d \overline{ \Phi^N_{t,\a_k}f } \bigg) \Big|_{t=0} 
\\ \nonumber
&
- \frac{1}{2 Z_N}
\int_{P_N(A)} 
\left( \frac{d}{dt} \mc E_k [\Phi^N_{t,\a_k}f] \right) \left( \prod_{m=0}^{k-1}\chi_{R_m}(\mc E_m[\Phi^N_{t,\a_k}f]) \right) 
\\ \nonumber
& \ \ \ \ \ \ \cdot
\exp \left(- \mc E_k [\Phi^N_{t,\a_k} g]  \right) \exp\left(-\frac12 \| \Phi^N_{t,\a_k} g \|_{L^2}^2  \right)  
 \bigg( \prod_{|n| \leq N} d \Phi^N_{t,\a_k}f d \overline{ \Phi^N_{t,\a_k}f } \bigg) \Big|_{t=0} 
\\ \nonumber
&=  \sum_{\ell=0}^{k-1}\int_{\mc M_N(A)} \left( \frac{d}{dt}\mc E_{\ell}[\Phi^N_{t,\a_k} f] \Big|_{t=0} \right)
\chi'_{R_\ell}(\mc E_{\ell}[P_N f]) \left( \prod_{\substack{m=0 \\ m \neq \ell}}^{k-1}\chi_{R_m}(\mc E_m[P_N f]) \right) 
\\ \nonumber
& \ \ \ \ \ \ \cdot
\exp(\widetilde Q_{k}[P_N f])\g_k(d f )
\\ \nonumber
&
-
\int_{\mc M_N(A)} 
\left( \frac{d}{dt} \mc E_k[\Phi^N_{t,\a_k} f] \Big|_{t=0} \right) \left( \prod_{m=0}^{k-1}\chi_{R_m}(\mc E_m[P_N f]) \right) \exp(\widetilde Q_{k}[P_N f])
\g_k(d f )\,.
\end{align}
Since $\supp \chi'_{R}\subseteq \supp \chi_R$, Remark \ref{rmk:integrability} entails that the functionals ($\ell=1,\ldots,k-1$)
\begin{align}\nonumber
\chi'_{R_\ell}(\mc E_{\ell}[P_N f]) \Bigg( \prod_{\substack{m=0 \\ m \neq \ell}}^{k-1}\chi_{R_m}(\mc E_m[P_N f]) \Bigg) \exp(\widetilde Q_{k}[P_N f])
\, ,
\end{align}
\begin{equation}\nonumber
  \left(  \prod_{m=0}^{k-1}\chi_{R_m}(\mc E_m[P_Nf]) \right) \exp(\widetilde Q_{k}[P_N f])
\end{equation}
are bounded in $L^2(\g_k)$ uniformly in $N$ (namely their $L^2(\g_k)$ norm is bounded by a constant independent on $N \in \mathbb{N}$). 
Thus the Cauchy-Schwarz inequality yields
\be
\left|\frac{d}{dt} \tilde\rho_{k,N}(\Phi^N_{t,\a_k}(A)) \Big|_{t=0}
\right|\lesssim \sum_{m=1}^k\left\|\frac{d}{dt}\mc E_m[\Phi^N_{t,\a_k} f]\Big|_{t=0}\right\|_{L^2(\g_k)}\,,
\ee
and the r.h.s. vanishes as $N \to \infty$ thanks to Proposition \ref{Prop: energie}. Using Lemma \ref{Lemma:t0}, this concludes the proof.
\end{proof}

We are finally ready to prove Proposition \ref{Prop:invarianza-gauged}.

\begin{proof}[Proof of Proposition \ref{Prop:invarianza-gauged}]
We fix $R > 1$, $5/4 < r < k - 1/2$ and $0 \leq s < r$ (recall $k \geq 2$). We consider $C \subset B^r(R)$ that is compact in the $H^s(\T)$ 
topology. 
Given $T > 0$ we integrate \eqref{eq:lolcal} over the interval $[0,T]$, 
so that, letting $N \to \infty$, we obtain 
\begin{equation}\label{SempreRecall}
 \tilde \rho_{k}(C) =
 \lim_{N \to \infty} \tilde \rho_{k,N}(C)   =
\lim_{N \to \infty} \tilde \rho_{k,N}(\Phi^N_{t,\a_k}(C))     \,,
\end{equation}
where we used Proposition \ref{sec5prop1} in the first identity.
Now we take $T = t_R$ given by Proposition~\ref{Prop:nearness}, so that, given any
$\e >0$ we have   
\be\label{eq:inclusione-flussi.fine}
\Phi^N_{t,\a_k}(C)  \subseteq \Phi_{t,\a_k}( C) + B^{s}(\e), \qquad |t| \leq t_R \, .  
\ee
for all sufficiently large $N = N(\e)$. 
Thus
$$
 \tilde \rho_{k,N}(\Phi^N_{t,\a_k}(C))   \leq 
 \tilde \rho_{k,N} ( \Phi_{t,\a_k}( C) + B^{s}(\e) ),     \qquad |t| \leq t_R \, ,
$$
for all $N = N(\e)$ sufficiently large. Recalling Proposition \ref{sec5prop1} we can pass to the limit $N \to \infty$ so that, for all $\e>0$, we find 
\begin{equation}\nonumber
\lim_{N \to \infty} 
\tilde \rho_{k,N}(\Phi^N_{t,\a_k}(C))   
\leq
\tilde \rho_{k} (\Phi_{t,\a_k}( C) + B^{s}(\e) ), \qquad |t| \leq t_R  \, .
\end{equation}
By \eqref{SempreRecall} and the arbitrarity of $\e$ we have arrived to
\begin{equation}\label{SempreRecall2}
 \tilde \rho_{k}(C) \leq \inf_{\e >0} \tilde \rho_{k} (\Phi_{t,\a_k}( C) + B^{s}(\e) ), \qquad |t| \leq t_R \, .
\end{equation} 
Since $C$ is compact in $H^s$ and the flow map is continuous in the $H^s$ topology when we restrict the 
data to $B^r(R)$ (see \cite[Theorem 1.1, Corollary 1.2]{Herr}),
we have that $\Phi_{t,\a_k}(C)$ is compact too. Thus
$$
\bigcap_{\e >0} ( \Phi_{t,\a_k}( C) + B^{s}(\e) ) = \Phi_{t,\a_k}( C)   \, ,
$$
and
$$
\inf_{\e >0} \tilde \rho_{k} (\Phi_{t,\a_k}( C) + B^{s}(\e) ) = \tilde \rho_{k} \left( \bigcap_{\e >0} ( \Phi_{t,\a_k}( C) + B^{s}(\e) ) \right) 
= \tilde \rho_{k} (  \Phi_{t,\a_k}( C) ), \qquad |t| \leq t_R    \, .
$$
Recalling \eqref{SempreRecall2} we arrive to
$$
 \tilde \rho_{k}(C)  \leq \tilde \rho_{k} (  \Phi_{t,\a_k}( C) ), \qquad |t| \leq t_R  \, .
$$
Now, since $C \in B^r(R)$, by \cite[Theorem 1.1, Corollary 1.2]{Herr} and inequality \eqref{lemma:gauge-bound-HsBis}, we 
notice that $\Phi_{t,\a_k}(C)$ belongs to a (eventually larger $R < R'$) ball $B^r(R')$, so that 
a well-known continuation argument allows to take $t_{R} = \infty$,
namely
\begin{equation}\label{DesRes}
\tilde \rho_{k}(C)  \leq \tilde \rho_{k} (  \Phi_{t,\a_k}( C) ), \qquad  t \in \R  \, .
\end{equation}
Letting $R \to \infty$, this bound can be promoted to any compact set in $H^r(\T)$ an then to
to an identity 
\begin{equation}\label{DesResBis}
\tilde \rho_{k}(C) = \tilde \rho_{k} (  \Phi_{t,\a_k}(C) ), \qquad  \mbox{$C$ compact in $H^s(\T)$ and $C \subset H^r(\T)$, $t \in \R$}   \, .
\end{equation}
exploiting the time reversibility of the flow. Indeed, still as 
consequence of~\cite[Theorem 1.1, Corollary 1.2]{Herr} and inequality \eqref{lemma:gauge-bound-HsBis}
we have that  
$\Phi_{t,\a_k}$ is an $H^s$-diffeomorphism once we restrict the data to a ball of $H^r(\T)$. Thus, for any $C'$ compact
in $H^s(\T)$ and contained in a ball of $H^r$ the set $\Phi_{t,\a_k} ( C' )$ has the 
same property, and, letting $C' = \Phi_{-t,\a_k}(C)$, we get
$$
\tilde \rho_{k} ( \Phi_{-t,\a_k}(C))  \leq \tilde \rho_{k}( \Phi_{t,\a_k} ( \Phi_{-t, k}(C) )  
= \tilde \rho_{k} ( C ),
\qquad t \in \R \, ,
$$
that, along with the \eqref{DesRes}, leads to \eqref{DesResBis}.
Finally, using the regularity of the gaussian measure, we get that \eqref{DesResBis} holds for any Borel set of $H^s(\T)$ that are subsets of 
$H^r(\T)$ too. This completes the proof.
\end{proof}

\section{Finite dimensional approximations of the Gauge map}\label{sect:gauge1}

In this Section we show that the Gauge map $\Ga_\a$ can be approximated by a family of 
bijective maps $\Ga_\a^N$ acting on $E_N$. For any $N \in \mathbb{N}$ these maps are still a one-parameter 
group with respect to $(\R,+)$.

%
%
%

Given $f\in E_N$ we define $\Ga^N_a f$, $\a \in \mathbb{R}$ as the (unique) solution of the differential equation 
\begin{equation}\label{eq:gauge-troncataDiff}
\frac{d}{d\a} \Ga_\a^N f=i P_N\left(\mathcal I[\Ga^N_{\a}f]\Ga^N_{\a} f \right)\,,\quad \Ga_0^N f=f \,.
\end{equation}
The \eqref{eq:gauge-troncataDiff} is globally well posed for data in $L^2(\T)$ by the same argument explained in Section \ref{sect:flows} 
for the truncated GDNLS equation \ref{20140509:DNLSapprox}, taking advantage of the
fact that the $L^2$ norm is a conserved quantity, as shown in Lemma \ref{L2NormGauge}.
In the rest of the section we will prove some properties of the solutions of \eqref{eq:gauge-troncataDiff} and of the associated flow map $f \to \Ga^N_\a f$.
Hereafter we will use the notation $\Ga_\a^{\infty} f := \Ga_\a f$ and $E_{\infty} := L^{2}(\T)$.
Let us recall that $B^s(R) := \{ f : \| f \|_{H^s}  \leq R \}$.

\begin{lemma}\label{L2NormGauge}
Let $N \in \N \cup \{ \infty \}$. For all $f \in E_N$ we have
$$
\| \Ga_\a^N f \|_{L^{2}} = \| f \|_{L^2} 
$$
\end{lemma}
\begin{proof}
Using equation \eqref{eq:gauge-troncataDiff} and its conjugate we can compute
$$
\frac{d}{d\a}\| \Ga_\a^N f \|^2_{L^{2}} = 2 \Re i \int  \overline{\Ga^N_{\a} f} P_N \left(\mathcal I[\Ga^N_{\a}f]  \Ga^N_{\a} f \right)  
$$
Since $P_{>N} \overline{\Ga^N_{\a}f} = 0$, we get, by orthogonality,
\begin{align}
\frac{d}{d\a}\| \Ga_\a^N f \|^2_{L^{2}}
= 2 \Re  i \int   \mathcal I[\Ga^N_{\a}f]  |\Ga^N_{\a} f|^2   = 0 \, ,
\end{align}
where the last identity is immediate since $ \mathcal I [\mathscr G_\alpha^N f]$ is real.
\end{proof}

Then we control the growth of the $\dot H^1$ norm of $\Ga_\a f$.
A similar statement can be proved for general $H^s$, $s\geq 0$ norms, but will be not necessary for our purposes. 
As usual we use the simplified notation $\mu = \mu[f] := \frac{1}{2 \pi}\|f\|^2_{L^2}$.

\begin{lemma}\label{H1NormGauge}
Let $N \in \N \cup \{ \infty \}$. For all $f \in E_N$ we have
\begin{equation}\label{StrIneq}
\| \Ga_\a^N f \|^2_{\dot{H}^{1}} \lesssim  e^{\alpha \mu} ( \| f \|^2_{\dot{H}^{1}} + \mu ) 
\end{equation}
\end{lemma}
\begin{proof}
Proceeding as in the proof of Lemma \ref{L2NormGauge}, and using the fact that $\mc I[\ms G_\alpha^Nf]$ is a real function,
we obtain
$$
\frac{d}{d\a}\| \Ga_\a^N f \|^2_{\dot{H}^{1}} = 2 \Re i \int  \mathcal I[\Ga^N_{\a}f]'  \Ga^N_{\a} f \overline{\Ga^N_{\a} f}'
=
- 2 \Im \int  \mathcal I[\Ga^N_{\a}f]'  \Ga^N_{\a} f \overline{\Ga^N_{\a} f}'   \, .
$$
Then, recalling the definition \eqref{DefMathcali} of $\mathcal I [\cdot]$ and using Lemma \ref{L2NormGauge}, we have
$$ 
\mathcal I [\mathscr G_\alpha^N f]' = |\mathscr G_\alpha^N f|^2 - \frac{1}{2\pi} \| \mathscr G_\alpha^N f \|_{L^2}^2 
= 
|\mathscr G_\alpha^N f|^2 - \frac{1}{2\pi} \| f \|_{L^2}^2 
=: 
|\mathscr G_\alpha^N f|^2 - \mu \, ,
$$ 
so that 
$$
\frac{d}{d\a}\| \Ga_\a^N f \|^2_{\dot{H}^{1}} = - 2 \Im  \int |\mathscr G_\alpha^N f|^2  \Ga^N_{\a} f \overline{\Ga^N_{\a} f}'   
+ 2  \mu \Im \int   \Ga^N_{\a} f \overline{\Ga^N_{\a} f}' \, .
$$
Now using the H\"older, Cauchy--Scwartz and the following Gagliardo--Nirenberg inequality
$$
\| h \|_{L^{6}} \leq \| h \|_{\dot{H}^1}^{\frac13} \| h \|_{L^2}^{\frac23} \, ,
$$
and again recalling Lemma \ref{L2NormGauge}, we arrive to
\begin{align}
\frac{d}{d\a} 
\| \Ga_\a^N f \|^2_{\dot{H}^{1}} 
& 
\lesssim \| |\mathscr G_\alpha^N f|^3 \|_{L^{2}} \|\mathscr G_\alpha^N f \|_{\dot{H}^1} + 
\mu \| \Ga^N_{\a} f \|_{L^{2}} \|\mathscr G_\alpha^N f \|_{\dot{H}^1}
\\ \nn
& 
\simeq
\| \mathscr G_\alpha^N f \|_{L^{6}}^3 \|\mathscr G_\alpha^N f \|_{\dot{H}^1} + 
\mu^{3/2}  \|\mathscr G_\alpha^N f \|_{\dot{H}^1}
\\ \nn
&\lesssim
\| \mathscr G_\alpha^N f \|_{L^{2}}^2 \|\mathscr G_\alpha^N f \|_{\dot{H}^1}^2 + 
\mu^{2} +  \mu \|\mathscr G_\alpha^N f \|_{\dot{H}^1}^2
\simeq \mu^{2} +  \mu \|\mathscr G_\alpha^N f \|_{\dot{H}^1}^2 \, .
\end{align}
Thus Gr\"onwall's inequality yields
$$
\| \Ga_\a^N f \|^2_{\dot{H}^{1}} 
 \lesssim 
\| f \|^2_{\dot{H}^{1}} e^{\a \mu} + \mu e^{\a \mu} - \mu \, , 
$$
which implies \eqref{StrIneq}.
\end{proof}

In the following Lemma we show that the 
flow $\Ga^N_\a$ approximates $\Ga^\infty_\a := \Ga_\a$ for large $N$ in the $L^2$ topology. The approximation is uniform for 
initial data in 
a ball of $H^1$. By the proof it will be clear that one can obtain a similar approximation property 
w.r.t the $H^s$ topology and data in a ball of $H^r$, for all $0 \leq s < r$ as long as $r > 1/2$. However we do not need this stronger statement. 

\begin{lemma}\label{HsDecay}
Let $N \in \N $, $R >1$ and $\bar \alpha \in \R$ 
We have
\begin{equation}\label{StabL2Easy}
\lim_{N \to \infty}  \sup_{f \in E_N \cap B^1(R), \, |\a| \leq |\bar \a| } \| \Ga_\a f - \Ga_\a^N f \|_{L^2} = 0 \, . 
\end{equation}
\end{lemma}

\begin{proof}
%
%
We will need the immediate inequalities
\begin{equation}\label{Triv1}
\| \mathcal{I}(\Ga^N_\a f) \|_{L^{\infty}} \lesssim \| \Ga^N_\a f \|^2_{L^{2}} \simeq \mu ,
\end{equation}
\begin{equation}\label{Triv2}
\| \mathcal{I}(\Ga_\a f) - \mathcal{I}(\Ga^N_\a f)\|_{L^{\infty}} \lesssim \| \Ga_\a f  +  \Ga^N_\a f \|_{L^2} \| \Ga_\a f  -  \Ga^N_\a f \|_{L^2} \lesssim \sqrt{\mu} \| \Ga_\a f  -  \Ga^N_\a f \|_{L^2} \, ,
\end{equation}
valid for $N \in \N \cup \{ \infty \}$ and $f \in E_{N}$, (recall $\mu = \mu[f] := \frac{1}{2\pi}\|f\|^2_{L^2}$). These follow immediately recalling the form \eqref{DefMathcali} of $\mathcal{I}[\cdot]$ and Lemma \ref{L2NormGauge}.
Let
$$
\delta^N_\a f :=  \Ga_\a f - \Ga^N_\a f 
$$
Hereafter we restrict to $f \in E_N$. 
Notice that $\delta^N_\a f$
solves
\begin{equation}\nn
\frac{d}{d \a} \delta^N_\a f 
= 
i P_{>N} \left( \mathcal{I} [\Ga_\a f] \Ga_\a f  \right) 
+ i  P_N \left( \mathcal{I}[\Ga_\a f]\delta^N_\a f + \left( \mathcal{I} [\Ga_\a f] - \mathcal{I}[\Ga^N_\a f] \right)\Ga^N_\a f  \right) \, .
\end{equation}
Pairing this in $L^2$ with $\delta^N_\a f$ we arrive to
\begin{align}\nn
\frac{d}{d \a}
&\| \delta^N_\a f \|_{L^2} 
\\ \nn
& = 2 \Re i \left(  
- \int \left(  P_{>N} \left( \mathcal{I}[\Ga^N_\a f] \Ga^N_\a f  \right) \right) \overline{\delta^N_\a f}
+ \int   \mathcal{I}[\Ga_\a f] |\delta^N_\a f|^2  
+ 
\int \left(  \mathcal{I}[\Ga_\a f] - \mathcal{I}[\Ga^N_\a f] \right) \Ga^N_\a f   \overline{\delta^N_\a f} \right)   \, .
\end{align}
Using the H\"older and Cauchy--Schwartz inequalities and \eqref{Triv1}, \eqref{Triv2} we arrive to
\begin{align}\label{StabGAlpha}
\frac{d}{d \a} 	\| \delta^N_\a f \|_{L^2} 
& 
\lesssim
\| P_{>N} ( \mathcal{I}[\Ga^N_\a f] \Ga^N_\a f  ) \|^2_{L^2} +  \| \delta^N_\a (f) \|^2_{L^2}
\\ \nn
& 
+ \| \mathcal{I}[\Ga_\a f] \|_{L^{\infty}} \| \delta^N_\a f \|^2_{L^2} 
+ 
\| \mathcal{I}[\Ga_\a f] - \mathcal{I}[\Ga^N_\a f] \|_{L^{\infty}}  \| \Ga^N_\a f \|_{L^2} \| \delta^N_\a f \|_{L^2} 
\\ \nn
& 
\lesssim 
\| P_{>N} ( \mathcal{I}[\Ga^N_\a f] \Ga^N_\a f  ) \|^2_{L^2} 
+ (1 + \mu)  \| \delta^N_\a f \|^2_{L^2}    \, .
\end{align}

Then using the algebra property of $H^1$, the fact that $ \partial_x \mathcal{I}[\Ga_\a f] = \Ga_\a f $ and \eqref{Triv1} and Lemma \ref{H1NormGauge}, 
we can estimate
\begin{align}\nn
\sup_{f \in E_N \cap B^1(R), \, |\a| \leq |\bar \a| }
\| \mathcal{I}[\Ga_\a f] \Ga_\a f \|_{H^1}
& \lesssim 
\sup_{f \in E_N \cap B^1(R), \, |\a| \leq |\bar \a| } \mu \| \Ga_\a f \|_{H^1} + \| \Ga_\a f \|^2_{H^1}  \leq \mu^2  \, ,
\\ \nn
& 
\lesssim \mu e^{\frac{|\bar \a| \mu}{2}} (R + \sqrt \mu) + e^{|\bar \a| \mu} (R^2 + \mu) \lesssim R^2 e^{|\bar \a| R} \, ,
\end{align}
so that 
$$
\sup_{f \in E_N \cap B^1(R), \, |\a| \leq |\bar \a| } 
\| P_{>N} ( \mathcal{I}[\Ga_\a f] \Ga_\a f  ) \|^2_{L^2} \lesssim \frac{1}{N^2}  R^2 e^{|\bar \a| R}
$$
and \eqref{StabGAlpha} becomes
$$
\frac{d}{d \a} \| \delta^N_\a f \|_{L^2}  \lesssim \frac{1}{N^2} R^2 e^{|\bar \a| R} + (1 + \mu)  \| \delta^N_\a f \|^2_{L^2}   \, ,
$$
for $|\a| \leq |\bar \a|$ and $f \in B^1(R)$.
Thus using Gr\"onwall's lemma and the fact that $\delta^N_\a f |_{\a = 0} = 0$ we arrive to
$$
\sup_{f \in E_N \cap B^1(R), \, |\a| \leq |\bar \a| }  \| \delta^N_\a f \|^2_{L^2}  \lesssim \frac{1}{N^2} R^2 e^{|\bar \a| R} e^{(1 + \mu) |\bar \alpha|} \, , 
$$
that implies \eqref{StabL2Easy}. 
\end{proof}

The next result is a direct corollary of Lemmas \ref{H1NormGauge} and \ref{HsDecay}. For the proof we refer to \cite[Proposition 2.10]{sigma}

\begin{corollary}\label{StabCorImp}
Let $\varepsilon > 0$, $R >1$, $\bar \alpha \in \R$. Given $A \subset B^1(R)$ a compact set, there exists
$N^*$ such that 
\begin{equation}\nonumber
\Ga_\a (A) \subset \Ga^N_\a (A + B^s(R)) \, . 
\end{equation}
for all $|\a| \leq |\bar \a|$ and for all $N > N^*$.
\end{corollary}

The next technical lemma will be used to prove Proposition \ref{FlowMapBij} and then to show the 
absolute continuity of $\g_k$ under the gauge map. Let us recall that we have defined in \eqref{Def:dive} the operator $\dive$ applied to an $n-$th 
dimensional vectorial function of $f, \bar f \in E_N$  as 
$$
\dive F (f, \bar f) = \sum_{|n| \leq N} \left( 
\frac{ \partial F_{n} }{ \partial f (n) } + \frac{ \partial \bar{F}_{n} }{ \partial \bar f (n) }  \right) \, .
$$

\begin{lemma}\label{lemma:div}
We have
 \be\label{eq:div}
\left|  \dive i P_N \left(\mathcal I[P_N f] P_N f)\right) \right| \lesssim \| f \|^2_{H^{1}} \frac{\log N}{\sqrt N} \,.
\ee
 \end{lemma}
 \begin{proof}
A direct computation from (\ref{DefMathcali}) yields
\be\label{eq:FTdi-IBis}
(\mathcal{I}[P_N f])(0) = 0, \qquad
(\mathcal{I}[P_N f])(m) = 
- \frac{i}{m} \sum_{|\ell|, |\ell - m| \leq N }  f (\ell) \bar f (\ell - m) \quad \mbox{if}  \quad m \neq 0\,, 
\ee
thus
\begin{equation}\label{FTOTRHS1}
i \left( \mathcal{I}[P_N f] P_N f \right)(n) =   \sum_{m \, : \, m \neq 0, |n-m| \leq N  } \frac{1}{m}
 \sum_{\ell \, : \, |\ell|, |\ell - m|,  \leq N }   f(n-m) f (\ell) \bar f (\ell - m) \, ,
\end{equation}
and
$$
\dive i P_N \left(\mathcal I[P_N f ] P_N f)\right)=   2  \sum_{n \, : \, |n| \leq N} \quad \sum_{m \, : \,  m \neq 0, |n-m| \leq N }  \frac{1}{m} |f(n-m)|^{2}  \, .
$$
%
%
%
Since
$$
\sum_{n \, : \, |n| \leq N} \quad \sum_{m \, : \,  m \neq 0, |n-m| \leq N }  \frac{1}{m} |f(n-m)|^{2} = 
 \sum_{n=1}^{N} \left( |f(-n)|^{2} - |f(n)|^{2} \right) \sum_{m=N-n+1}^{N+n} \frac{1}{m} \, ,
$$
and we clearly have $|f(n)|^{2} \leq \frac{1}{n^2} \| f \|^{2}_{H^{1}}$, we can estimate 
\begin{align}\nn
 \sum_{n=1}^{N} 
 & 
 \left( |f(-n)|^{2} - |f(n)|^{2} \right) \sum_{m=N-n+1}^{N+n} \frac{1}{m}
 \\ \nn
& 
\lesssim \|f\|^{2}_{H^{1}} \sum_{n=1}^{N} \frac{1}{n^{2}} \ln \left( \frac{N+n}{N-n+1}\right) 
= \|f\|^{2}_{H^{1}} \sum_{n=1}^{N} \frac{1}{n^{2}} \ln \left( 1 + \frac{2n-1}{N-n+1} \right) \, .
\end{align}
Then, for $N$ sufficiently large, we have
\begin{align}
 \sum_{n=1}^{N} \frac{1}{n^{2}} \ln \left( 1 + \frac{2n-1}{N-n+1} \right)
& 
\lesssim 
\sum_{n=1}^{\lfloor \sqrt{N} \rfloor} \frac{1}{n^{2}} \ln \left( 1+ \frac{2}{\sqrt{N}}\right)
+ 
\sum_{n=\lfloor \sqrt{N} \rfloor}^{N} \frac{1}{n^{2}} \ln \left( 2N \right)
\\ \nonumber
& 
\lesssim \ln 
\left( 1+ \frac{2}{\sqrt{N}}\right) + \frac{ \ln \left( 2N \right) }{ \sqrt{N}}  \, ,
\end{align}
whence \eqref{eq:div}.
\end{proof}

Now we will prove that the flow maps $f \to \Ga^N_\a f$ are injective and that they preserve the Lebesgue measure in the limit 
$N \to \infty$, if we restrict to bounded subsets of $H^1$.  
We recall that we have defined the Lebesgue measure (see \eqref{NewLebesgue}) 
as proportional to 
$$
\prod_{|n| \leq N} d\phi_{N}(n)d \bar \phi_{N} (n) \, .
$$
When we transform this volume form under $\Ga^N_\a$ we have to take into account the 
Jacobian matrix $D\mathscr G_\alpha^N(f)$, where the differential operator $D$ acts on a map 
$$
T : f  \in E_N \to  Tf \in  E_N
$$ 
in the following way 
\begin{equation}\label{Def:Diff}
(D T) (f) =
\begin{pmatrix}
\left(\frac{\partial (T f )(m)}{\partial f(n)}
\right)_{|m|,|n|\leq N}
&
\left(\frac{\partial (\overline{Tf})(m)}{\partial f(n)}
\right)_{|m|,|n|\leq N}
\\
\left(\frac{\partial (T f)(m)}{\partial \overline{f}(n)}
\right)_{|m|,|n|\leq N}
&
\left(\frac{\partial (\overline{T f})(m)}{\partial \overline{f}(n)}
\right)_{|m|,|n|\leq N}
\end{pmatrix}
\,.
\end{equation}

\begin{proposition}\label{FlowMapBij}
For all $f \in E_N$we have 
\begin{equation}\label{DiffDive}
\det[(D\Ga^N_{\a})(f)] = \exp \left( \int_0^\a d\a'     \dive i P_N  \left(\mathcal I [\Ga^N_{\a'}(f) ] \Ga^N_{\a'}(f) \right) \right) \, .
\end{equation}
In particular, the flow map $f \to \Ga_\a f$ is injective. Moreover for all $\a \in \mb R$ and $f \in H^1$ we have
$\det[(D\Ga^N_{\a})(f)] \to 1$ as $N \to \infty$. In fact the convergence is uniform for $f$ in bounded subsets of~$H^1$ and $\a$ in bounded intervals. 
\end{proposition}

\begin{proof}
Since $\Ga_{\a}^N$ is a one parameter group of transformations the chain rule gives 
\be\label{eq:1-gruppo}
(D\Ga^N_{\a+\e})(f)=(D\Ga^N_{\e})(\Ga^N_{\a}(f))\cdot (D\Ga^N_{\a})(f)\,.
\ee 
Then 
\bea\nn
\frac{d}{d\a}\det[(D\Ga^N_{\a})(f)]
&
=
&
\lim_{\e\to0}\frac{\det[(D\Ga^N_{\e})(\Ga^N_{\a}(f)]\det[(D\Ga^N_{\a})(f)]-\det[(D\Ga^N_{\a})(f)]}{\e}\nn\\
&
=
&
\lim_{\e\to0}\frac{\det[(D\Ga^N_{\e})(\Ga^N_{\a}(f))]-1}{\e}\cdot \det[(D\Ga^N_{\a})(f)]\nn\\
&
=
&
\left.\frac{d}{d\e}\det[(D\Ga^N_{\e})(\Ga^N_\a(f))]\right|_{\e=0}\det[(D\Ga^N_{\a})(f)]\,.
\eea
Thereby, since $\det[(D\Ga^N_{0})(f)]=1$, we arrive to
\begin{equation}\label{DefFlowForm}
\det[(D\Ga^N_{\a})(f)]=e^{\int_0^\a d\a' \Lambda_{\a'}(f)}\,
\quad\mbox{with}\quad
\Lambda_{\a'}(f):=\left.\frac{d}{d\e}\det[(D\Ga^N_{\e})(\Ga^N_{\a'}(f))]\right|_{\e=0}\,.
\end{equation}
Thus the \eqref{DiffDive} will be a consequence of the following identity
 \begin{equation}\label{eq:div1New}
\left[\frac{d}{d\e}\det[(D\Ga^N_{\e})(g)]\right]_{\e=0}
=
      \dive  i P_N \left(\mathcal I [g] g \right) \, ,
\end{equation}
valid for all $g \in E_N$. To prove \eqref{eq:div1New} we apply $D$ to
equation \eqref{eq:gauge-troncataDiff} which becomes
\begin{equation}\nonumber
\frac{d}{d\e} (D \Ga_\e^N) (g) =  D \, i P_N \left(\mathcal I[\Ga^N_{\e} g ]\Ga^N_{\e} g \right)\,.
\end{equation}
thus, since $(D \Ga_0^N)(g)=\mathbb{I}$, we get
\begin{equation}\nonumber
(D \Ga_\e^N) (g) = \mathbb{I}  +  \e \left(  D \, i P_N \left(\mathcal I[\Ga^N_{\e} g ]\Ga^N_{\e} g \right) \Big|_{\e = 0} \right)  + \mathcal{O}(\e) 
= \mathbb{I}  +  \e  D \, i P_N \left(\mathcal I[ g ]  g \right)  + \mathcal{O}(\e) \,, 
\end{equation}
from which we obtain
\begin{equation}\nonumber
\det[ (D \Ga_\e^N) (g)] = 1 + \e  \Tr   D \, i P_N  \left(\mathcal I[ g ]  g \right) + \mathcal{O}(\e) \,.
\end{equation}
This immediately implies the identity \eqref{eq:div1New}, recalling the definition of $D$ given in \eqref{Def:Diff} and 
that of divergence in \eqref{Def:dive}.
The second part of the statement is a consequence of \eqref{DiffDive} and of the inequality \eqref{eq:div}.
\end{proof}

%
%
%
%
%
%



\section{Quasi-invariance of $\g_k$ under the gauge map}\label{sect:gauge2}

The goal of this section is to prove Theorem \ref{th:gauge}.
We recall the definition of the restricted measure
\begin{equation}\label{RestrGauss}
\tilde\g_k(A)=\g_k(A\cap\{f \in L^2 : \mu[f] \leq R_0\})\,,\quad R_0>0\,. 
\end{equation}
\begin{remark}\label{AcRemark}
The absolutely continuity of $\hat \r_k$ w.r.t. $\g_k$ follows by Theorem \ref{th:gauge}, 
since $\hat \r_k := \tilde \rho_{k} \circ \Ga_{\a_k}$ and $\tilde \rho_{k}$ is absolutely continuous w.r.t. $\tilde \g_k$ by
Proposition \ref{sec5prop1}.
\end{remark}

We need the following accessory lemma.

\begin{lemma}\label{lemma:bound-exp}
Let $k\geq2$ and $R_{0} > 0$ sufficiently small. There there is $b>0$ such that  
\begin{equation}\label{MainTTEst}
\int \exp\left(b\left| \int u_N^{(k)}u_N^{(\b_2)} u_N^{(\b_3)} u_N^{(\b_4)} \right|\right)\tilde\g_k(df) < C \,, 
\end{equation}
for all $N \in \N$ and for all $ \b_2 \geq \b_3 \geq \b_4 \geq 0$ with $\b_2+ \b_3 + \b_4 = k-1$, where $u_N$ may denote either $P_N f$ or $P_N \bar f$. The 
constant $C$ is independent on $N$.
\end{lemma}

The proof uses a slight modification of the argument for proving the same statement for $k=1$ and $\b_3=\b_4=0$ we learned from Thomann and Tzvetkov (unpublished). In the sequel we just underline the few differences given by the case $k\geq2$ and for more details we refer to \cite[Theorem 3.3]{brereton}.

\begin{proof}
We prove that for $R_{0} > 0$ sufficiently small there is $c>0$ depending only on $k$ such that
\begin{equation}\label{eq:bound-exp}
\tilde\g_k\left( \left| \int u_N^{(k)} u_N^{(\b_2)} u_N^{(\b_3)}  u_N^{(\b_4)} \right| \geq t \right)  \lesssim e^{-ct} \, .
\end{equation} 
This immediately implies the statement taking $0< b < c$.
We treat separately the regimes $t\leq\sqrt N$ and $t>\sqrt N$. 

For $t>\sqrt N$ we will prove directly
(as usual $\chi$ is the characteristic function)
\be\label{eq:gauge-tail}
\tilde\g_k\left( \left| \int u_N^{(k)} u_N^{(\b_2)} u_N^{(\b_3)} u_N^{(\b_4)} \right| \geq t \right)\chi_{\{t>\sqrt N\}}\lesssim e^{-ct} \, ,
\ee
for small $R_{0} > 0$ and some $c>0$. Hereafter $c$ will denote several positive small constants, possibly decreasing from line to line.

For $t\leq \sqrt N$ we use the decomposition
\begin{align}
&\tilde\g_k\left( \left| \int u_N^{(k)} u_N^{(\b_2)} u_N^{(\b_3)} u_N^{(\b_4)} \right| \geq t \right)\chi_{\{t\leq\sqrt N\}}\nn\\
&\leq\tilde\g_k\left( \left| \int u_T^{(k)} u_T^{(\b_2)} u_T^{(\b_3)} u_T^{(\b_4)} \right| \geq \frac t2 \right)\chi_{\{t\leq\sqrt N\}}\label{eq:gauge-head1}\\
&+\tilde\g_k\left( \left| \int u_N^{(k)} u_N^{(\b_2)} u_N^{(\b_3)} u_N^{(\b_4)} 
- \int u_T^{(k)} u_T^{(\b_2)} u_T^{(\b_3)} u_T^{(\b_4)}\right| \geq \frac t2 \right)\chi_{\{t\leq\sqrt N\}}\label{eq:gauge-head2}\,,
\end{align}
with $T:=\lfloor t^2 \rfloor$. The term (\ref{eq:gauge-head1}) enjoys the tail estimate (\ref{eq:gauge-tail}). For the
term (\ref{eq:gauge-head2}) \cite[Lemma 5.3]{GLV16} yields 
$$
\tilde\g_k\left( \left| \int u_N^{(k)} u_N^{(\b_2)} u_N^{(\b_3)} u_N^{(\b_4)} 
- \int u_T^{(k)} u_T^{(\b_2)} u_T^{(\b_3)} u_T^{(\b_4)}\right| \geq t \right)\chi_{\{t\leq\sqrt N\}}\lesssim e^{-ct}\,. 
$$
(In fact \cite[Lemma 5.3]{GLV16} is formulated only for $\b_2=k-1, \b_3=\b_4=0$. The extension to our more general case is however straightforward).

It remains to prove (\ref{eq:gauge-tail}). First we show 
\begin{equation}\label{eq:stima-LittlePaley}
\left|  \int u_N^{(k)} u_N^{(\b_2)} u_N^{(\b_3)} u_N^{(\b_4)} \right| 
\lesssim \left( \sum_{j\geq0} 2^{j(k-\frac12)} \| \Delta_j u_N \|_{L^2}   \right)
\prod_{\ell=2}^{4} \left( \sum_{j\geq0} 2^{j(\b_\ell + \frac12)} \| \Delta_j u_N \|_{L^2}   \right) 
 \,,
\end{equation}
where we recall $\D_j:=P_{2^{j}}-P_{2^{j-1}}$ are Paley-Littlewood projectors.
Notice that, since $\b_\ell \leq k-1$ we have $\b_\ell + \frac12 \leq k - \frac12$. 
This bound follows noting that by orthogonality  
\begin{align}\nonumber
\left| \int u_N^{(k)} u_N^{(\b_2)} u_N^{(\b_3)} u_N^{(\b_4)} \right|
& 
\lesssim
\sum_{  j_{\ell} \geq 0, \,  j_1 \leq 3 j_{2}   } 2^{j_{1} k} 2^{j_2  \b_2} 
\int  |\Delta_{j_1} u_N| |\Delta_{j_2} u_N| |\Delta_{j_3} u_N| |\Delta_{j_4} u_N|
\\ \nonumber
& 
\lesssim
\sum_{ j_{\ell} \geq 0   } 2^{j_{1} (k - \frac12 )} 2^{j_{2} (\b_2 + \frac12 )} \int  |\Delta_{j_1} u_N| |\Delta_{j_2} u_N| |\Delta_{j_3} u_N| |\Delta_{j_4} u_N| \, ,
\end{align}
then H\"older's and Bernstein's inequalities lead to
\begin{align}\nonumber
& \left| \int u_N^{(k)} u_N^{(\b_2)} u_N^{(\b_3)} u_N^{(\b_4)} \right|
\lesssim
\sum_{   j_{\ell} \geq 0  } 2^{j_{1} (k - \frac12 )} 2^{j_{2} (\b_2 + \frac12 )} 
\| \Delta_{j_1} u_N \|_{L^{2}} \| \Delta_{j_2} u_N \|_{L^2} 
\| \Delta_{j_3} u_N \|_{L^{\infty}} \| \Delta_{j_4} u_N \|_{L^{\infty}}
\\ \nonumber
& 
\lesssim
\sum_{   j_{\ell} \geq 0 } 2^{j_{1} (k - \frac12 )} 2^{j_{2} (\b_2 + \frac12 )} 2^{j_{3} (\b_3 + \frac12 )} 2^{j_{4} (\b_4 + \frac12 )} 
\| \Delta_{j_1} u_N\|_{L^{2}} \| \Delta_{j_2} u_N \|_{L^2} 
 \| \Delta_{j_3} u_N \|_{L^{2}} \| \Delta_{j_4} u_N \|_{L^{2}}\end{align} 
 from which one deduces the \eqref{eq:stima-LittlePaley}.
Since
$$
\frac12 - \frac{1}{4k} + \sum_{\ell=2}^{4} \frac{1}{2k} \left( \b_{\ell} + \frac12 \right) = 1 \, ,
$$
by (\ref{eq:stima-LittlePaley}) and the union bound (notice that the sums in \eqref{eq:termine1}, \eqref{eq:termine2} are over a finite number of terms)
we arrive to
\begin{align}
\tilde\g_k\left( \left| \int u_N^{(k)}u_N^{(\b_2)} u_N^{(\b_3)} u_N^{(\b_4)} \right| \geq t \right)
& \leq
 \g_k\left(\sum_{j\geq} 2^{j(k-\frac12)} \| \Delta_j u_N \|_{L^2}\geq t^{\frac{1}{2} - \frac{1}{4k}}\right)\label{eq:termine1}
 \\ 
& + \sum_{\ell=2}^{4} \g_k \left( \sum_{j\geq0} 2^{j (\b_{\ell} + \frac12)} \| \Delta_j u_N \|_{L^2} \geq t^{  \frac{1}{2k}(\b_{\ell} + \frac12)} \right)\,.\label{eq:termine2}
\end{align}
First we note that the simple bound
$$
\sum_{0 \leq j < \frac{1}{2k} \ln_2 t } 2^{j(k-\frac12)} \| \Delta_j u_N  \|_{L^2} \lesssim  R_0 t^{ \frac{1}{2} - \frac{1}{4k} }
$$
yields for $R_0$ small enough
\be\label{eq:jpiccoli1}
\g_k\left( \sum_{0 \leq j  < \frac{1}{2k} \ln_2 t } 2^{j(k-\frac12)} \| \Delta_j u_N  \|_{L^2} \geq t^{\frac{1}{2} - \frac{1}{4k}} \right)=0\,.
\ee
Similarly
$$
\sum_{0 \leq j < \frac{1}{2k} \ln_2 t} 2^{ j(\b_{\ell} + \frac12)} \| \Delta_j u_N \|_{L^2} \lesssim R_0 t^{\frac{1}{2k}(\b_{\ell} + \frac12)}
$$
gives 
\be\label{eq:jpiccoli2}
\g_k\left( \sum_{0 \leq j <  \frac{1}{2k} \ln_2 t} 2^{ j(\b_{\ell} + \frac12)} \| \Delta_j u_N \|_{L^2} \geq t^{\frac{1}{2k}(\b_{\ell} + \frac12)}\right)=0\,.
\ee
We introduce now a sequence $\s_j$ so that
$\sum_{j\geq 0}\s_j\leq \frac14$.
Therefore
\be\label{eq:j-int1-1}
\g_k\left( \sum_{ j \geq \frac{1}{2k} \ln_2 t } 2^{j(k-\frac12)} \| \Delta_j u_N \|_{L^2} \geq t^{ \frac12 - \frac{1}{4k} } \right)
\leq \sum_{ j \geq \frac{1}{2k} \ln_2 t  }\g_k\left(  2^{2+j(k-\frac12)} \| \Delta_j u_N \|_{L^2} \geq \s_jt^{ \frac12 - \frac{1}{4k} }\right)\,,
\ee
and
\be\label{eq:j-int2-1}
\g_k\left( \sum_{j \geq \frac{1}{2k} \ln_2 t} 2^{j(\b_{\ell} + \frac12)} \| \Delta_j u_N \|_{L^2} \geq t^{\frac{1}{2k}(\b_{\ell} + \frac12)}\right)
\leq \sum_{j\geq \frac{1}{2k} \ln_2 t} \g_k\left(  2^{2+j(\b_{\ell} + \frac12)} \| \Delta_j u_N \|_{L^2} \geq \s_jt^{\frac{1}{2k}(\b_{\ell} + \frac12)}\right)\,.
\ee
Now we use (recall $t\geq \sqrt N$)
\be\label{eq:j-int1-2}
\g_k\left(  2^{2+j(k-\frac12)} \| \Delta_j u_N \|_{L^2} \geq \s_jt^{\frac12 - \frac{1}{4k}} \right)= 
P\left(\sqrt{\sum_{n=2^{j-1}}^{2^j} |g_n|^2}\geq 2^{\frac j2-k-2}\s_jt^{\frac12 - \frac{1}{4k}}\right)\lesssim e^{-c2^{j}\s^2_j  t^{1- \frac{1}{2k}}}
\ee
where $\{g_n,\bar g_n\}_{n\in\N}$ are i.i.d. standard complex Gaussian random variables and $P$ is the associated probability. In an analog manner we have 
\begin{align}\label{eq:j-int2-2}
\g_k & \left(  2^{2+j(\b_{\ell} + \frac12)} \| \Delta_j u_N \|_{L^2} \geq \s_jt^{\frac{1}{2k}(\b_{\ell} + \frac12)}\right)
 \\ \nonumber
&= P\left(\sqrt {\sum_{n=2^{j-1}}^{2^j} |g_n|^2}
\geq 2^{j(k - \b_{\ell} - \frac 12)-k-2}\s_jt^{\frac{1}{2k}(\b_{\ell} + \frac12 )}\right)\lesssim 
e^{-c2^{j(2k-2\b_{\ell}-1)}\s^2_j t^{\frac{1}{k}(\b_{\ell} + \frac12)}}\,. 
\end{align}
Therefore \eqref{eq:j-int1-1} and \eqref{eq:j-int1-2} give
\be\label{eq:j-int1-fin}
\g_k\left( \sum_{ j \geq \frac{1}{2k} \ln_2 t }  2^{j(k-\frac12)} \| \Delta_j u_N \|_{L^2} \geq t^{\frac12 - \frac{1}{4k}}\right) \lesssim e^{-ct}\,,
\ee
and likewise by \eqref{eq:j-int2-1} and \eqref{eq:j-int2-2} we obtain 
\be\label{eq:j-int2-fin}
\g_k\left( \sum_{ j \geq \frac{1}{2k} \ln_2 t }  2^{j(\b_{\ell} + \frac12)} \| \Delta_j u_N \|_{L^2} \geq \s_jt^{\frac{1}{2k}(\b_{\ell} + \frac12)}\right) \lesssim e^{-ct}\,.
\ee

Summarising, by (\ref{eq:jpiccoli1}), (\ref{eq:jpiccoli2}), (\ref{eq:j-int1-fin}), (\ref{eq:j-int2-fin}) we obtain as $t>\sqrt N$ 
$$
\g_k\left(\sum_{j\geq0} 2^{j(k-\frac12)} \| \Delta_j u_N \|_{L^2}\geq t^{\frac 12 - \frac{1}{4k}}\right)
\lesssim e^{-c t}\,,
\quad \g_k\left( \sum_{j\geq0} 2^{j (\b_{\ell} + \frac12)} \| \Delta_j u_N \|_{L^2} \geq t^{\frac{1}{2k}(\b_{\ell} + \frac12)}\right)\lesssim e^{-c t}\,,
$$
whence (\ref{eq:gauge-tail}) follows. 
\end{proof}

\begin{proposition}\label{lemma:Lp}
Let $k\geq 2$ and $R_0 >0$ sufficiently small. Then there exists $C >0$, which only depends on $k$ and $R_0$,  
so that for any $p>1$
\be\label{eq:bound-der-alpha}
\left\|\frac{d}{d\a}\|(\Ga_\a^N f)\|^2_{H^k}\Big|_{\a=0}\right\|_{L^p(\tilde\g_k)}\leq Cp\,.
\ee
\end{proposition}

\begin{proof}
Pairing in $\dot H^{k}$ equation \eqref{eq:gauge-troncataDiff} with $\Ga_\a^N f$, we get
\begin{align}\label{KeyOrth}\nn
\frac{d}{d\a}\|\Ga_\a^N f\|^2_{\dot{H}^k}\Big|_{\a=0}
&=2\Re\int \overline{\Ga_\a^N f}^{(k)} \left(\frac{d}{d\a}\Ga_\a^N f\right)^{(k)}\Big|_{\a=0}\nn\\
&=2\Re\int \overline{\Ga_\a^N f}^{(k)} P_N(i\mc I[\Ga_\a^N f]\Ga_\a^N f)^{(k)}\Big|_{\a=0}\nn\\
&=2\Re\int \overline{f}^{(k)} P_N(i\mc I[f]f)^{(k)}\nn\\
&=\frac{d}{d\a}\|\Ga_\a f\|^2_{\dot{H}^k}\big|_{\a=0}\,,
\end{align}
where the last identity follows by orthogonality and using equation \eqref{eq:gauge-troncataDiff} again.

Note that, by Lemma \ref{lem:gauge1} (used with $\psi=\ms G_{-\alpha}f$) and the representation of the integrals of motion
of the GDNLS given in \eqref{MainReprForm}, we have that
\begin{align}\label{eq:ReallyUsRepr}\nn
\|\mc \Ga_{\alpha}f\|^2_{\dot{H}^k}
&=\|f\|_{\dot{H}^k}^2
-2ik\alpha\mu\int f^{(k-1)}\bar f^{(k)}
\\
&
+i\alpha \sum_{\substack{ \b_2 + \b_3 + \b_4=k-1\\ \b_2 \geq \b_3 \geq \b_4 }} 
C_{\b_2, \b_3, \b_4}
\int (u^{(k)}  u^{(\b_2)} u^{(\b_3)}  u^{(\b_4)}) + \mc O(\a^2)
\,,
\end{align}
where as usual $u$ denotes either $f$ or $\bar f$.
Therefore, in order to estimate the derivative of $\|\Ga_\a^N y\|_{H^k}$ in $\a=0$, we need to estimate the terms in (\ref{eq:ReallyUsRepr}) which are linear in $\a$. More precisely, using \eqref{KeyOrth} and \eqref{eq:ReallyUsRepr}, we arrive to
\begin{align}\nn
& \left\|\frac{d}{d\a}\|\Ga_\a^N f\|_{H^k}\big|_{\a=0}\right\|_{L^p(\tilde\g_k)}
\\ \nn
& = 
\left\|2k\mu\int f^{(k-1)}\bar f^{(k)}
+ \sum_{\substack{ \b_2 + \b_3 + \b_4=k-1\\ \b_2 \geq \b_3 \geq \b_4 }} 
C_{\b_2, \b_3, \b_4}
\int (u^{(k)}  u^{(\b_2)} u^{(\b_3)}  u^{(\b_4)}) \right\|_{L^p(\tilde\g_k)}
\\ 
&
\leq 2 k  \left\|\mu\int f^{(k-1)}\bar f^{(k)}\right\|_{L^p(\tilde\g_k)}
+  \sum_{\substack{ \b_2 + \b_3 + \b_4=k-1\\ \b_2 \geq \b_3 \geq \b_4 }}
C_{\b_2, \b_3, \b_4} \left \|\int (u^{(k)}  u^{(\b_2)} u^{(\b_3)}  u^{(\b_4)})\right\|_{L^p(\tilde\g_k)}\,.
\label{eq2}
\end{align}
The first term of \eqref{eq2} is easier to handle, as
$$
\left\|\mu\int f^{(k-1)}\bar f^{(k)}\right\|_{L^p(\tilde\g_k)}\leq R_0\left\|\int f^{(k-1)}\bar f^{(k)}\right\|_{L^p(\tilde\g_k)}\lesssim R_0 p
\left\|\int f^{(k-1)}\bar f^{(k)}\right\|_{L^2(\g_k)}
$$
by hyper-contractivity and the $L^2(\g_k)$ norm on the r.h.s. is bounded by an absolute constant due to due to Lemma \ref{Wick2}.
To handle the remaining terms of \eqref{eq2} we use
Lemma \ref{lemma:bound-exp} and the elementary inequality $|x|^p/p^p\leq e^{|x|}$, which lead us to 
\begin{equation}\nn
\left \|\int (u^{(k)}  u^{(\b_2)} u^{(\b_3)}  u^{(\b_4)})\right\|_{L^p(\tilde\g_k)} \lesssim p   \, ,
\end{equation}
so that the proof is concluded. Notice that we have also used that $u_N^{(k)}u_N^{(\b_2)} u_N^{(\b_3)} u_N^{(\b_4)}$ converges $\tilde \g_k$-a.e. to $u^{(k)}u^{(\b_2)} u^{(\b_3)} u^{(\b_4)}$ (that has been proved in \cite{GLV16}) and Fatou's Lemma.
\end{proof}

Once we have that the norm in \eqref{eq:bound-der-alpha} grows at most linearly in $p$ the almost invariance of $\g_k$ w.r.t. the gauge map 
follows taking advantage (again) of the group property
of the map. We follow the argument of
\cite{sigma} and \cite{NLW}. Here we recall the key statements omitting the proofs when they are easily adapted from these papers.
As usual, we approximate $\tilde \g_k$ by a family of weighted measures $\tilde \g_{k,N}$ with density
\be\label{eq:jacobiOLD}
 \tilde \g_{k,N}(df) =  \chi_{\{ \mu [P_N f] \leq R_0 \} }(P_N f) \g_k(d f) \,.
\ee

\begin{lemma}\label{lemma:t0}
Let $k\geq2$. Then
\begin{equation}\label{TzvEst}
\frac{d}{d\a} \tilde\gamma_{k,N} (\Ga_{\a}^N (A))
\lesssim  p \,
\tilde \gamma_{k,N} (\Ga_{\a}^N (A))^{1-\frac{1}{p}} \left(1+\frac{\log N}{\sqrt N}\right)\,.
\end{equation}
for all $A\in\B(L^2(\T))$.
\end{lemma}

\begin{proof}
Let $A\in\B(L^2(\T))$. For all $\bar\a \in \R$ we have  
\begin{equation}
\frac{d}{d\a} \tilde\gamma_{k,N}( \Ga^N_\a (A))\Big|_{\a=\bar \a}
=
\frac{d}{d\a} \tilde\gamma_{k,N}( \Ga^N_\a (\Ga^N_{\bar \a}(A)))\Big|_{\a=0} \, .
\end{equation}
This is a simple consequence of the fact that $\Ga^N_\a$ is a one parameter group of transformations. One can adapt directly the 
argument in \cite{TV13b} (Proposition 5.4, step 2). 
Using Proposition~\ref{FlowMapBij}, the identity \eqref{eq:div1New} and the fact that $\| \Ga^N_\a f \|^2_{L^2} = \| P_N f \|_{L^2}$ for all $\a \in \R$ (see Lemma \ref{L2NormGauge}), 
we can compute
\begin{align}\nn
\frac{d}{d\a} & \tilde\g_{k,N}(\Ga^N_{\a}(\Ga^N_{\bar\a}))\Big|_{\a=0}
=
\frac{d}{d\a}  \int_{\Ga^N_{\bar\a}(A)} \tilde\g_{k,N}(d(\Ga^N_{\a} f))\Big|_{\a=0} \\ \nn
&=
\frac{d}{d\a} \int_{\Ga^N_{\bar\a}(A)} \tilde\g_{k,N}(df) |\det D\Ga^N_\a (f)|  
e^{- \frac12 ( \| \Ga^N_\a f \|^2_{\dot{H}^k} - \| f \|^2_{\dot{H}^k} ) }  \Big|_{\a=0} \\ \nn
&
=
\int_{\Ga^N_{\bar\a}(A)}  \tilde \g_{k,N}(df)  \frac{d}{d\a}  |\det D\Ga^N_\a (f)|  \Big|_{\a=0}
- \frac12
\int_{\Ga^N_{\bar\a}(A)}  \tilde\g_{k,n}(df) 
|\det D\Ga^N_\a (f)|   
\frac{d}{d\a}\|\Ga_\a^N f\|^2_{\dot{H}^k} \Big|_{\a=0} \\ \nn
&
=
i \int_{\Ga^N_{\bar\a}(A)}  \tilde\g_{k,N}(df)   
    \dive P_N \left(\mathcal I [f] f \right)\Big|_{\a=0} 
- \frac12
 \int_{\Ga^N_{\bar\a}(A)}      \tilde\g_{k,N}(df)
\frac{d}{d\a}\|\Ga_\a^N f\|^2_{\dot{H}^k} \Big|_{\a=0}\, .
\end{align}
%
Then by H\"older's inequality we obtain
\begin{align}
\frac{d}{d\a} \tilde\g_{k,N}( \Ga^N_\a (A))\Big|_{\a=\bar \a}
& 
\leq
\tilde\g_{k,N}(\Ga^N_{\bar\a}(A))^{1-\frac1p} \|   \dive P_N \left(\mathcal I [f] f \right) \|_{L^p(\g_k)}
\\ \nn
& 
+ \tilde\g_{k,N}(\Ga^N_{\bar\a}(A))^{1-\frac1p} \Big\|\frac{d}{d\a}\|(\Ga_\a^N f)\|^2_{H^k}\Big|_{\a=0}\Big\|_{L^p(\tilde\g_k)}
\end{align}
The first term on the r.h.s. is bounded using Lemma \ref{lemma:div}, as $\|\|f\|^2_{H^1}\|_{L^p(\g_k)}\lesssim p$ for any $p\geq1$, $k\geq2$. The second one is bounded by Lemma \ref{lemma:Lp} and we obtain \eqref{TzvEst}. 
\end{proof}

\begin{lemma}\label{TzvOhLemma1}
Let $k \geq 2$.
There exists an absolute constant $\bar \a > 0$ so that the following holds. 
For all $\e >0$ 
there exists $\delta(\e)$ so that the following holds. If $A \in\B(L^2(\T)) $ is such that
with $\g_{k,N}(A) < \delta $ then $\g_{k,N}( \Ga_\a^N A) < \e $
for all $|\alpha| < \bar\a$. 
\end{lemma}

\begin{proof}
It follows by Lemma \ref{lemma:t0}, using exactly the same argument of \cite[Proposition 5.3]{NLW}. 
\end{proof}

Now we want to pass to the limit $N \to \infty$ to get an a priory bound of $\g_{k}( \Ga_\a^N A)$ for small values of $\a$
and sets $A$ of small measure.

\begin{proposition}
Let $k \geq 2$.
There exists an absolute constant $\bar \a > 0$ so that the following holds. Let $R>1$ and $A \in \B(L^2(\T))$ with $A \subset B^1(R)$. 
Then for all $\e >0$ there exists $\d >0$ so that  
$$
\tilde\gamma_k (A) < \delta \Rightarrow \sup_{|\a| \leq \bar \a} \tilde\gamma_k (\Ga_\a A) < \e \,. 
$$  
\end{proposition}

\begin{proof}
It follows by 
Lemma \eqref{TzvOhLemma1} and Corollary \ref{StabCorImp}, using exactly the same argument of \cite[Lemma 5.5]{NLW}.
\end{proof}

We are now ready to prove Theorem \ref{th:gauge}.
Given $A \in \B(L^2(\T))$ with $A  \subset B^1(R)$ we have 
$$ 
\tilde\gamma_k (A \cap B^1(R) ) < \delta \Rightarrow \sup_{|\a| \leq \bar \a} \tilde\gamma_k (\Ga_\a (A \cap B^1(R))) < \e
$$
As $\e>0$ is arbitrary we get
$$ 
\tilde\gamma_k (A \cap B^1(R) ) = 0 \Rightarrow \sup_{|\a| \leq \bar \a} \tilde\gamma_k (\Ga_\a (A \cap B^1(R))) = 0 \, .
$$
Since $\a$ is independent on $R$ we can take the limit $R \to \infty$ to get
$$ 
\tilde\gamma_k (A) = 0 \Rightarrow \sup_{|\a| \leq \bar \a} \tilde\gamma_k (\Ga_\a (A)) = 0 \, ,
$$
for all $\tilde\g_k$-measurable subset $A \subseteq H^1(\T)$. Then the statement follows iterating this estimate, using again that $\bar \a$ is an absolute constant, and recalling that $H^1$ has full $\tilde\g_k$-measure for $k\geq2$.


\end{document}